\documentclass{statsoc}

\usepackage[a4paper]{geometry}
\usepackage{graphicx}
\usepackage[textwidth=8em,textsize=small]{todonotes}
\usepackage{amsmath}
\usepackage{natbib}
\usepackage{mathtools}
\usepackage[ruled,vlined]{algorithm2e}
\usepackage{amsmath, amssymb}
\usepackage{bbm}
\usepackage{amsfonts, multirow, epsfig, subfig, floatrow}
\usepackage{graphicx, pdflscape, verbatim, enumerate, colortbl, setspace}
\usepackage{setspace, color,bm}
\usepackage[normalem]{ulem}
\usepackage{cite}
\usepackage{multirow}
\usepackage{booktabs,array}
\usepackage{url}

\newtheorem{proposition}{Proposition}
\newtheorem{lemma}{Lemma}
\newtheorem{definition}{Definition}
\newtheorem{theorem}{Theorem}

\newtheorem{assumption}{Assumption}
\newtheorem{condition}{Condition}
\newtheorem{remark}{Remark}

\newtheorem{innercustomprop}{Proposition}
\newenvironment{customprop}[1]
  {\renewcommand\theinnercustomprop{#1}\innercustomprop}
  {\endinnercustomprop}
 
 \newtheorem{innercustomthm}{Theorem}
\newenvironment{customthm}[1]
  {\renewcommand\theinnercustomthm{#1}\innercustomthm}
  {\endinnercustomthm}
  
   \newtheorem{innercustomremark}{Remark}
\newenvironment{customremark}[1]
  {\renewcommand\theinnercustomremark{#1}\innercustomremark}
  {\endinnercustomremark}
  
   \newtheorem{innercustomlemma}{Lemma}
\newenvironment{customlemma}[1]
  {\renewcommand\theinnercustomlemma{#1}\innercustomlemma}
  {\endinnercustomlemma}  
  
\floatname{algorithm}{Procedure}

\makeatletter
\newcommand*{\rom}[1]{\expandafter\@slowromancap\romannumeral #1@}
\makeatother

\newcommand{\R}{\mathbb{R}}

\newcommand{\indep}{\rotatebox[origin=c]{90}{$\models$}}

\title[Estimating Optimal Treatment Rules with an Instrumental Variable]{Estimating Optimal Treatment Rules with an Instrumental Variable: A Partial Identification Learning Approach}
\author[Author 1 {\it et al.}]{Hongming Pu}
\address{University of Pennsylvania, Philadelphia, USA}
\author[H. Pu and B. Zhang]{Bo Zhang}
\address{University of Pennsylvania, Philadelphia, USA}

\begin{document}
\let\thefootnote\relax\footnotetext{\textit{Address for correspondence:} Bo Zhang, Department of Statistics, The Wharton School, University of Pennsylvania, Philadelphia, PA 19104 (e-mail: \textsf{bozhan@wharton.upenn.edu}).}

\begin{abstract}
Individualized treatment rules (ITRs) are considered a promising recipe to deliver better policy interventions. One key ingredient in optimal ITR estimation problems is to estimate the average treatment effect conditional on a subject's covariate information, which is often challenging in observational studies due to the universal concern of unmeasured confounding. Instrumental variables (IVs) are widely-used tools to infer the treatment effect when there is unmeasured confounding between the treatment and outcome. In this work, we propose a general framework of approaching the optimal ITR estimation problem when a valid IV is allowed to only partially identify the treatment effect. We introduce a novel notion of optimality called ``IV-optimality''. A treatment rule is said to be IV-optimal if it minimizes the maximum risk with respect to the putative IV and the set of IV identification assumptions. We derive a bound on the risk of an IV-optimal rule that illuminates when an IV-optimal rule has favorable generalization performance. We propose a classification-based statistical learning method that estimates such an IV-optimal rule, design computationally-efficient algorithms, and prove theoretical guarantees. We contrast our proposed method to the popular outcome weighted learning (OWL) approach via extensive simulations, and apply our method to study which mothers would benefit from traveling to deliver their premature babies at hospitals with high level neonatal intensive care units. \textsf{R} package \textsf{ivitr} implements the proposed method.

\end{abstract}


\section{Introduction}
\pagenumbering{arabic}
\subsection{Estimating Individualized Treatment Rules with a Valid Instrumental Variable}
Individualized treatment rules (henceforth ITRs) are now recognized as a general recipe for leveraging vast amount of clinical, prognostic, and socioeconomic status data to deliver the best possible healthcare or other policy interventions. Researchers across many disciplines have responded to this trend by developing novel data-driven strategies that estimate ITRs. Some seminal works include \citet{murphy2003optimal}, \citet{robins2004optimal}, \citet{qian2011performance}, \citet{zhang2012estimating, zhang2012robust}, and \citet{zhao2012estimating}, among others. See \citet{kosorok2019precision} and \citet{tsiatis2019dynamic} for comprehensive and up-to-date surveys. One key ingredient in ITR estimation problems is to estimate the average treatment effect conditional on patients' clinical and prognostic features. However, estimation of the \emph{conditional average treatment effect} (CATE) can be challenging in randomized control trials (henceforth RCTs) with high-dimensional covariates, limited sample size, and individual noncompliance, and observational studies due to the universal concern of \emph{unmeasured confounding} (\citealp{kallus2018confounding}; \citealp{kallus2018interval};
\citealp{cui2019semiparametric};
\citealp{zhang2020selecting}; \citealp{qiu2020optimal}).

In non-ITR settings, instrumental variables (IVs) are commonly used tools to infer treatment effects in observational studies where observed covariates cannot adequately adjust for confounding between the treatment and outcome. However, few works have explored using an IV to estimate optimal ITRs. Two exceptions are \citet{cui2019semiparametric} and \citet{qiu2020optimal}, both of which studied ITR estimation problems when an IV can be used to \emph{point identify} the conditional average treatment effect under assumptions introduced in \citet{wang2018bounded}. However, one limitation of their approaches is that their assumptions that allow a valid IV to point identify the CATE may be quite stringent in some scenarios and do not necessarily hold, especially in ITR estimation settings where treatment effect heterogeneity is expected. 

This article takes a distinct perspective. Although a valid IV in general cannot point identify the average treatment effect (ATE) and similarly the CATE, it can \emph{partially identify} them, in the sense that a lower and upper bound of ATE and CATE can be obtained with a valid IV. Depending on the quality of the putative IV and various identification assumptions, lengths of partial identification intervals may vary, and in the extreme case the intervals collapse to points, i.e., the ATE or CATE (\citealp{AIR1996}; \citealp{robins1996identification}; \citealp{balke1997bounds}; \citealp{manski2003partial}). See \citet{swanson2018partial} for an up-to-date literature review on partial identification of ATE using an IV. It is worth pointing out that a partial identification interval is fundamentally different from a confidence interval, in that a confidence interval shrinks to a point as sample size goes to infinity, while a partial identification interval remains, in the limit, an interval, and represents the intrinsic uncertainty of an IV analysis. 

A popular approach to ITR estimation problems in non-IV settings is to transform the problem into a weighted classification problem, where the sign of the CATE constitutes a subject's label \{-1, +1\}, and the magnitude the weight. In an IV setting, the point identified CATE is replaced with an interval $I$. When such an interval avoids $0$, say $I = [1, 3]$, it is clear that the subject benefits from receiving the treatment and should be labeled as such. However, the situation becomes complicated when the interval covers $0$, say $I = [-1, 3]$, in which case the true CATE can be anything between $-1$ and $3$, and such a subject can no longer be labeled as benefiting or not benefiting from the treatment. In this way, partial identification of the CATE in an IV analysis poses a fundamental challenge to optimal ITR estimation. 

We have three objectives in this article. First, we propose a general classification-based (\citealp{zhang2012estimating}; \citealp{zhao2012estimating}) framework for optimal ITR estimation problems with an IV. The putative IV is allowed to only partially identify the conditional average treatment effect. Second, we introduce a novel notion of optimality called \emph{IV-optimality}: a treatment rule is said to be \emph{IV-optimal} if it minimizes the maximum risk that could incur given a putative IV and a set of IV identification assumptions. One remarkable feature of ``IV-optimaltiy'' is that it is \emph{amenable} to different IVs and identification assumptions, and allows empirical researchers to weigh estimation precision against credibility. We also derive bounds on the risk of an IV-optimal rule, which illuminates when an IV-optimal ITR has favorable generalization performance. Finally, we derive an estimator of such an IV-optimal rule which we call the IV-PILE estimator, develop computationally-efficient algorithms, and prove theoretical guarantees. The worst-case risk of the IV-PILE estimator can be estimated and gives practitioners important guidance on the applicability of their estimated ITRs. Via extensive simulations, we demonstrate that our proposed method has favorable generalization performance compared to applying the outcome weighted learning (OWL) (\citealp{zhao2012estimating}) and similar methods to observational data in the presence of unmeasured confounding. Finally, we apply our developed method to study the differential impact of delivery hospital on neonatal health outcomes. \textsf{R} package \textsf{ivitr} implements the proposed methodology.

\subsection{A Motivating Example: Differential Impact of Delivery Hospital on Premature Babies' Health Outcomes}
\label{subsec: intro motivating example}
We consider a concrete example to carry forward our discussion. \citet{lorch2012differential} constructed a cohort-based retrospective study out of all hospital-delivered premature babies in Pennsylvania and California between $1995$ and $2005$ and Missouri between $1995$ and $2003$. Using the differential travel time as an instrumental variable, \citet{lorch2012differential} find that there is benefit to neonatal outcomes when premature babies are delivered at hospitals with high-level neonatal intensive care units (NICUs) compared to hospitals without high-level NICUs.

An IV analysis is crucial in this study and similar studies based on observational data. Studies that directly use the treatment, e.g., delivery hospitals in the NICU study, to infer the treatment effect often cannot sufficiently adjust for unmeasured and unrecorded factors, such as the severity of comorbidities or laboratory results (\citealp{lorch2012differential}). Although \citet{lorch2012differential} have found evidence supporting a positive treatment effect of mothers delivering premature infants at high-level NICUs, some important questions remain. First and foremost, it is of great scientific interest to understand which mothers had better be sent to hospitals with high-level NICUs as compared to which mothers are just as well off at low level NICUs. Given the current limited capacity of high level NICUs, an answer to this question would facilitate our understanding of which mothers and their premature babies are most in need of high-level NICUs, and provide insight into how to construct optimal perinatal regionalization systems, systems that designate hospitals by the scope of perinatal service provided and designate where infants are born or transferred according to the level of care they need at birth (\citealp{lasswell2010perinatal}; \citealp{kroelinger2018comparison}). Such scientific inquiries elicit estimating individualized treatment rules using observational data consisting of mothers' observed characteristics, treatment received, outcome of interest, and a valid instrumental variable. We will revisit this example and apply our developed methodology to it near the end of the article.

\section{ITR Estimation with an IV: from Point to Partial Identification}
\label{sec: recast the problem into semi-sup classification}
We briefly review the problem of estimating ITRs from a classification perspective, and discuss the key impediment to generalizing the estimation strategy from RCTs to observational data in Section \ref{subsec: estimate ITR from classification perspective}. Section \ref{subsec: IV assumption and PI} discusses how to leverage a valid IV to partially identify the conditional average treatment effect, and Section \ref{subsec: recast prob in semi-sup learning} proposes a general framework of approaching the ITR estimation problem with a valid IV.

\subsection{ITR Estimation from a Classification Perspective: from Randomized Control Trials to Observational Studies}
\label{subsec: estimate ITR from classification perspective}
Suppose that the data $\{(\mathbf{X}_i, A_i, Y_i), i = 1, \dots, n\}$ are collected from a two-arm randomized trial. The $d$-dimensional vector $\mathbf{X}_i \in \mathcal{X}$ encodes subject $i$'s prognostic characteristics, $A_i \in \mathcal{A} = \{-1, +1\}$ a binary treatment, and $Y_i \in \mathbb{R}$ an outcome of interest. Let $f(\cdot): \mathcal{X} \mapsto \mathbb{R}$ be a discriminant function such that the sign of $f(\cdot)$ yields the desired treatment rule. Let $\mu(1, \mathbf{X}) = \mathbb{E}[Y(1) \mid \mathbf{X}] = \mathbb{E}[Y \mid A = 1, \mathbf{X}]$ and $\mu(-1, \mathbf{X}) = \mathbb{E}[Y(-1) \mid \mathbf{X}] = \mathbb{E}[Y \mid A = -1, \mathbf{X}]$ denote the average potential outcomes conditional on $\mathbf{X}$ in each arm, and $C(\mathbf{X}) = \mu(1, \mathbf{X}) - \mu(-1, \mathbf{X})$ the \emph{conditional average treatment effect}, or CATE, following the notation in \citet{zhang2012estimating}. 

The value function of a particular rule $f(\cdot)$ is defined to be $V(f) = \mathbb{E}[Y(\text{sgn}\{f(\mathbf{X})\})]$. An optimal rule is the one that maximizes $V(f)$ among a class of functions $\mathcal{F}$, or equivalently minimizes the following risk:
\begin{equation}
    \label{eqn: classical objective function}
  \mathcal{R}(f) = \mathbb{E}\left[|C(\mathbf{X})|\cdot \mathbbm{1}\big\{\text{sgn}\{C(\mathbf{X})\} \neq \text{sgn}\{f(\mathbf{X})\}\big\}\right],
\end{equation}
where $\text{sgn}(x) = 1, ~\forall x > 0$ and $-1$ otherwise.

Let $B_i = \text{sgn}\{C(\mathbf{X}_i)\}$ be a latent class label that assigns $+1$ to subject $i$ if she would benefit from the treatment and $-1$ otherwise. If a correct treatment decision is made, i.e., $\text{sgn}\{f(\mathbf{X}_i)\} = \text{sgn}\{C(\mathbf{X}_i)\}$, there is no loss incurred; otherwise, the decision is not optimal and the corresponding loss has magnitude $W_i = |C(\mathbf{X}_i)|$. In this way, the optimal ITR estimation problem is recast as a weighted classification problem whose expected weighted misclassification error is specified in (\ref{eqn: classical objective function}). In a typical classification problem, the training data contain class labels and weights. In the context of ITR estimation problems, the contrast function $C(\mathbf{X}_i)$ for subject $i$ is first estimated from data, say as $\widehat{C}(\mathbf{X}_i)$, and the associated label and weight are then constructed accordingly: $\widehat{B}_i = \text{sgn}\{\widehat{C}(\mathbf{X}_i)\}$ and $\widehat{W}_i = |\widehat{C}(\mathbf{X}_i)|$. To summarize, the original data $\{(\mathbf{X}_i, A_i, Y_i), i= 1, \dots, n\}$ are transformed into $\{(\mathbf{X}_i, \widehat{B}_i, \widehat{W}_i), i = 1, \dots, n\}$, and a standard classification routine is then applied to this derived dataset in order to estimate a function $f(\cdot)\in\mathcal{F}$ that minimizes $\mathcal{R}(f)$. This framework, as discussed in more detail in \citet{zhang2012estimating}, covers many popular ITR methodologies. Notably, the popular \emph{outcome weighted learning} (OWL) approach (\citealp{zhao2012estimating}) is a particular instance of this general framework, where $C(\mathbf{X})$ is estimated via an inverse probability weighted estimator (IPWE) and a support vector machine (SVM) is used to perform classification.

One critical task in estimating optimal ITRs from this classification perspective is to estimate well the contrast $C(\mathbf{X})$. With a known propensity score as in a randomized control trial, an inverse probability weighted estimator (IPWE) can be used to unbiasedly estimate $C(\mathbf{X})$. However, this task becomes much more challenging when data come from observational studies. A key assumption in drawing causal inference from observational data is the so-called \emph{treatment ignorability} assumption (\citealp{rosenbaum1983central}), also known as the \emph{no unmeasured confounding assumption} (henceforth NUCA) (\citealp{Robins1992}), or \emph{treatment exogeneity} (\citealp{Imbens2004}). A version of the no unmeasured confounding assumption states that
\begin{equation*}
     \label{eqn: treatment ignorability assumption}
\begin{split}
     &F(Y(1), Y(-1) \mid A = a, \boldsymbol{X} = \boldsymbol{x}) = F(Y(1), Y(-1) \mid \boldsymbol{X} = \boldsymbol{x}),~\forall(a, \boldsymbol{x}),
\end{split}
\end{equation*}
where $F(\cdot)$ denotes the cumulative distribution function. In words, the treatment assignment is effectively randomized within strata formed by observed covariates. However, when the NUCA fails, the conditional average treatment effect \emph{may not} be unbiasedly estimated from observed data as $\mathbb{E}[Y(1) \mid \mathbf{X}]$ is not necessarily equal to $\mathbb{E}[Y \mid A = 1, \mathbf{X}]$. 

\subsection{Instrumental Variables: Assumptions and Partial Identification}
\label{subsec: IV assumption and PI}

An instrumental variable is a useful tool to estimate the treatment effect when the treatment and outcome are believed to be confounded by unmeasured confounders. We consider the potential outcome framework that formalizes an IV as in \citet{AIR1996}. Let $Z_i \in \{-1, +1\}$ be a binary IV associated with subject $i$, and $\mathbf{Z}$ a length-n vector containing all IV assignments. Let $A_i(\mathbf{Z})$ be the indicator of whether subject $i$ would receive the treatment or not under IV assignment $\mathbf{Z}$, $\mathbf{A}$ a length-n vector of treatment assignment status with $A_i(\mathbf{Z})$ being the $i\text{th}$ entry, and $Y_i(\mathbf{Z}, \mathbf{A})$ the outcome of subject $i$ under IV assignment $\mathbf{Z}$ and treatment assignment $\mathbf{A}$. We assume that the following \emph{core} IV assumptions hold (\citealp{AIR1996}):
\begin{enumerate}
    \item[(IV.A1)] Stable Unit Treatment Value Assumption (SUTVA): $Z_{i}=Z_{i}^{\prime}$ implies $A_{i}(\mathbf{Z}) = A_{i}(\mathbf{Z}^{\prime})$; $Z_{i} = Z_{i}^{\prime}$ and $A_{i}=A_{i}^{\prime}$ together imply $Y_{i}(\mathbf{Z}, \mathbf{A}) = Y_{i}(\mathbf{Z}^{\prime}, \mathbf{A}^{\prime})$.
    \item[(IV.A2)] Positive correlation between IV and treatment: $P(A = 1\mid Z=1,\mathbf{X}=\mathbf{x}) > P(A = 1\mid Z = -1, \mathbf{X} = \mathbf{x})$ for all $\mathbf{x}$.
    \item[(IV.A3)] Exclusion restriction (ER): $Y_i(\mathbf{Z}, \mathbf{A}) = Y_i(\mathbf{Z}^\prime, \mathbf{A})$ for all $\mathbf{Z}, \mathbf{Z}^\prime, \mathbf{A}$. 
    \item[(IV.A4)] IV unconfoundedness conditional on $\mathbf{X}$: $Z \indep A(z), Y(z, a) \mid \mathbf{X}$ for $z \in \{-1, +1\}$ and $a \in \{-1, +1\}$. 
\end{enumerate}
These four core IV assumptions do not allow point identification of the average treatment effect (ATE); however, they lead to the well-known Balke-Pearl bound on the ATE for a binary outcome $Y$ (\citealp{balke1997bounds}). See \citet{swanson2018partial} for other (possibly weaker) versions of assumptions that lead to the same Balke-Pearl bound. For a continuous but bounded outcome, \citet{manski1998monotone} derived a nonparametric bound under the monotone instrumental variable assumption. Additional assumptions are typically needed to tighten these bounds or to identify the treatment effect in an identifiable subgroup or the entire population.

In a binary IV analysis, a subject belongs to one of the four compliance classes: 1) an always-taker if $\{A(1), A(-1)\} = (1,1)$; 2) a complier if $\{A(1), A(-1)\}=(1,-1)$; 3) a never-taker if $\{A(1), A(-1)\}=(-1,-1)$; 4) a defier if $\{A(1), A(-1)\}=(-1,1)$. When the outcome is binary and the proportion of defiers is known, \citet{richardson2010analysis} discussed how to bound the ATE among all four compliance classes as a function of the defier proportion. An important instance is when there is no defiers and a valid IV can be used to identify the \emph{local average treatment effect}, i.e., the average treatment effect among compliers, as shown in the seminal work by \citet{AIR1996}. 

The fundamental limitation of an IV analysis is that even a valid IV provides \emph{no} information regarding the counterfactual outcomes of always-takers and never-takers: we simply do not have information about what would happen had always-takers been forced to forgo the treatment, and never-takers had they been forced to accept the treatment. This suggests another strategy to further bound the population average treatment effect by setting bounds to $\mathbb{E}[Y(-1) \mid \text{always-takers}]$ and $\mathbb{E}[Y(1) \mid \text{never-takers}]$. \citet{swanson2018partial} contains a detailed account of various proposals on how to set these bounds. Some versions of IV identification assumptions, for instance those established in \citet{wang2018bounded}, would allow a valid IV to \emph{point identify} the population average treatment effect. Estimating optimal treatment rules when the CATE can be point identified using a valid IV is studied in \citet{cui2019semiparametric} and \citet{qiu2020optimal}, and is not the focus of the current paper, although it is a special case of our general framework. The rest of this article focuses on how to estimate useful individualized treatment rules when an IV only partially identifies the CATE, possibly under various IV-specific identification assumptions. 

\subsection{Estimating Optimal ITR with an IV from a Partial Identification Perspective}
\label{subsec: recast prob in semi-sup learning}
We now describe a general framework of approaching the problem of estimating optimal treatment rules using a valid IV in observational studies. Suppose that we have i.i.d data $\{(\mathbf{X}_i, Z_i, A_i, Y_i), i = 1, \dots, n\}$ with a binary IV $Z_i$, binary treatment $A_i$, observed covariates $\mathbf{X}_i \in \mathbb{R}^d$, and outcome of interest $Y_i \in \mathbb{R}$. Let $\mathcal{C}$ denote a set of IV identification assumptions and $I_i =  [L(\mathbf{X}_i), U(\mathbf{X}_i)] \ni C(\mathbf{X}_i)$ a partial identification interval of the conditional average treatment effect $C(\mathbf{X}_i)$ associated with subject $i$ under $\mathcal{C}$. We view each subject as belonging to one of the following three latent classes (with class label $B_i$): 
\begin{enumerate}
    \item $B_i = +1$ if $I_i > \Delta$ in the sense that $x > \Delta, \forall x \in I_i$; 
    \item $B_i = -1$ if $I_i < \Delta$ in the sense that $x < \Delta, \forall x \in I_i$;
    \item $B_i = \text{NA}$ if $\Delta \in I_i $.
\end{enumerate}
In words, the class $B = +1$ consists of those who would benefit at least $\Delta$ from the treatment, $B = -1$ those who would not benefit more than $\Delta$, and $B = \text{NA}$ those for whom the putative IV and the set of identification assumptions $\mathcal{C}$ together cannot assert that subject $i$ would benefit at least $\Delta$ from the treatment or not. We will refer to subjects with $B_i \in \{-1, +1\}$ as labeled subjects and $B_i = \text{NA}$ unlabeled. 

\begin{remark}\rm
We let $\Delta$ denote a margin of practical relevance. In many ITR settings, the treatment may only do harm (or good), and we would recommend taking/not taking the treatment only when the margin is large. In our application, a high-level NICU might never do harm to mothers and preemies compared to a low-level NICU; however, given the limited capacity of high level NICUs and long travel times for some mothers to a high level NICUs, it may be more reasonable for mothers and their newborns to be sent to the nearest NICU, unless a high-level NICU is \emph{significantly} better in reducing the mortality. This trade-off is reflected by the margin $\Delta$ set according to expert knowledge. One can always let $\Delta = 0$ and the problem is reduced to the more familiar setting. 
\end{remark}

In practice, $I_i$ is estimated from the observed data, say as $\widehat{I}_i$, and $B_i$ is constructed accordingly, say as $\widehat{B}_i$. Consider the derived dataset $\mathcal{D} = \{(\mathbf{X}_i, \widehat{I}_i, \widehat{B}_i), i = 1, \dots, n \}$. Write 
\[
\mathcal{D} = \mathcal{D}_l \cup \mathcal{D}_{ul} = \{(\mathbf{X}_i, \widehat{I}_i, \widehat{B}_i), i = 1, \dots, l\} \cup
\{(\mathbf{X}_i, \widehat{I}_i), i = l + 1, \dots, n\},
\]
where $l$ subjects in $\mathcal{D}_l$ have labels $\widehat{B}_i \in \{-1, +1\}$, and $u = n - l$ subjects in $\mathcal{D}_{ul}$ are unlabeled. Our goal is to still learn an ``optimal'' treatment rule $f(\cdot)$ such that some properly defined misclassification error is minimized. 

\section{Examples of Partial Identification Bounds for a Binary Outcome}
\label{sec: examples of C and application to NICU}

\subsection{Balke-Pearl Bound}
\label{subsec: balke pearl bound}
Assume that the four core IV assumptions (IV.A1 - IV.A4) stated in Section \ref{subsec: IV assumption and PI} hold within strata formed by observed covariates, i.e.,
\begin{equation*}
    \mathcal{C}_{\text{BP}} = \{\text{IV.A1 - IV.A4 hold within strata of}~\mathbf{X}\}.
\end{equation*}
The Balke-Pearl bounds state that the conditional average treatment effect $C(\mathbf{X})$ is lower bounded by (\citealp{balke1997bounds}; \citealp{cui2019semiparametric}):
\begin{equation}
\label{eqn: bp lower bound}
    L(\mathbf{X})  = \max \left\{\begin{array}{lr} p_{-1, -1 \mid -1, \mathbf{X}} + p_{1, 1 \mid 1, \mathbf{X}} - 1\\
        
        p_{-1, -1 \mid 1, \mathbf{X}} + p_{1, 1 \mid 1, \mathbf{X}} - 1\\
        
        p_{1, 1 \mid -1, \mathbf{X}} + p_{-1, -1 \mid 1, \mathbf{X}} - 1\\
        
        p_{-1, -1 \mid -1, \mathbf{X}} + p_{1, 1 \mid -1, \mathbf{X}} - 1  \\
        
         2p_{-1, -1 \mid -1, \mathbf{X}} +  p_{1, 1 \mid -1, \mathbf{X}} +  p_{1, -1 \mid 1, \mathbf{X}} +  p_{1, 1 \mid 1, \mathbf{X}} - 2 \\
         
         p_{-1, -1 \mid -1, \mathbf{X}} +  2p_{1, 1 \mid -1, \mathbf{X}} +  p_{-1, -1 \mid 1, \mathbf{X}} +  p_{-1, 1 \mid 1, \mathbf{X}} - 2 \\
         
          p_{1, -1 \mid -1, \mathbf{X}} +  p_{1, 1 \mid -1, \mathbf{X}} +  2p_{-1, -1 \mid 1, \mathbf{X}} +  p_{1, 1 \mid 1, \mathbf{X}} - 2 \\
         
          p_{-1, -1 \mid -1, \mathbf{X}} +  p_{-1, 1 \mid -1, \mathbf{X}} +  p_{-1, -1 \mid 1, \mathbf{X}} +  2p_{1, 1 \mid 1, \mathbf{X}} - 2 \\
        \end{array}\right\},
\end{equation}
and upper bounded by
\begin{equation}
\label{eqn: bp upper bound}
    U(\mathbf{X})  = \min \left\{\begin{array}{lr} 
        
        1 - p_{1, -1 \mid -1, \mathbf{X}} - p_{-1, 1 \mid 1, \mathbf{X}}\\
        
        1 - p_{-1, 1 \mid -1, \mathbf{X}} - p_{1, -1 \mid 1, \mathbf{X}}\\
        
        1 - p_{-1, 1 \mid -1, \mathbf{X}} - p_{1, -1 \mid -1, \mathbf{X}} \\
        
        1 - p_{-1, 1 \mid 1, \mathbf{X}} - p_{1, -1 \mid 1, \mathbf{X}}\\
        
         2 - 2p_{-1, 1 \mid -1, \mathbf{X}} -  p_{1, -1 \mid -1, \mathbf{X}} -  p_{1, -1 \mid 1, \mathbf{X}} -  p_{1, 1 \mid 1, \mathbf{X}}\\
         
         2 - p_{-1, 1 \mid -1, \mathbf{X}} -  2p_{1, -1 \mid -1, \mathbf{X}} -  p_{-1, -1 \mid 1, \mathbf{X}} -  p_{-1, 1 \mid 1, \mathbf{X}} \\
         
          2 - p_{1, -1 \mid -1, \mathbf{X}} -  p_{1, 1 \mid -1, \mathbf{X}} -  2p_{-1, 1 \mid 1, \mathbf{X}} -  p_{1, -1 \mid 1, \mathbf{X}}\\
         
          2 - p_{-1, -1 \mid -1, \mathbf{X}} -  p_{-1, 1 \mid -1, \mathbf{X}} -  p_{-1, 1 \mid 1, \mathbf{X}} -  2p_{1, -1 \mid 1, \mathbf{X}}\\
        \end{array}\right\},
\end{equation}
where $p_{y, a \mid z, \mathbf{X}}$ is a shorthand for $P(Y = y,  A = a \mid Z = z, \mathbf{X})$. Note that all conditional probabilities $\{P(Y = y,  A = a \mid Z = z, \mathbf{X} = \mathbf{x}),~y=\pm 1, a = \pm 1, z = \pm 1\}$ can in principle be nonparametrically identified. In practice, we may estimate them by re-coding $2 \times 2 = 4$ combinations of $Y \in \{-1, +1\}$ and $A \in \{-1, +1\}$ as four categories and fitting a flexible and expressive multi-class classification routine, e.g., random forests (\citealp{breiman2001random}).

\subsection{Bounds as in \citet{siddique2013partially}}
\citet{siddique2013partially} considers an assumption (in addition to four core IV assumptions) that limits treatment heterogeneity in the following way: 
\begin{itemize}
    \item[(IV.A5)] Correct Non-Compliant Decision:
    \begin{equation*}
        \begin{split}
             &\mathbb{E}[Y(1) \mid A = 1, Z = -1] - \mathbb{E}[Y(-1) \mid A = -1, Z = -1] \geq 0, \\
              &\mathbb{E}[Y(-1) \mid A = -1, Z = 1] - \mathbb{E}[Y(1) \mid A = -1, Z = 1] \geq 0.
        \end{split}
    \end{equation*}
\end{itemize}
In words, this assumption states that for those who take a treatment different from the encouragement (i.e., $A \neq Z$), their decisions are on average favorable. Under the four core IV assumptions and this extra assumption, the bound on ATE can be further tightened. Let
\begin{equation*}
    \mathcal{C}_{\text{Sid}} = \{\text{IV.A1 - IV.A4 \text{plus} IV.A5 hold within strata of}~\mathbf{X}\}.
\end{equation*}
Under IV identification set $\mathcal{C}_{\text{Sid}}$, the conditional average treatment effect is lower bounded by (\citealp{siddique2013partially}; \citealp{swanson2018partial}):
\begin{equation}
    \label{eqn: siddique lower bound}
    L(\mathbf{X}) = \max \left\{\begin{array}{lr} p_{1, 1 \mid 1, \mathbf{X}} + p_{1, -1 \mid 1, \mathbf{X}}\\
    p_{1, 1 \mid -1, \mathbf{X}} 
    \end{array}\right\} - \min \left\{\begin{array}{lr} p_{1, -1 \mid -1, \mathbf{X}} + p_{1 \mid -1, \mathbf{X}}\\
    p_{1, -1 \mid 1, \mathbf{X}} + p_{1 \mid 1, \mathbf{X}}
    \end{array}\right\},
\end{equation}
and upper bounded by
\begin{equation}
    \label{eqn: siddique upper bound}
    U(\mathbf{X}) = \min \left\{\begin{array}{lr} p_{1, 1 \mid 1, \mathbf{X}} + p_{-1 \mid 1, \mathbf{X}}\\
    p_{1, 1 \mid -1, \mathbf{X}} + p_{-1 \mid -1, \mathbf{X}}
    \end{array}\right\} - \max \left\{\begin{array}{lr} p_{1, -1 \mid -1, \mathbf{X}} + p_{1, 1 \mid -1, \mathbf{X}}\\
    p_{1, -1 \mid 1, \mathbf{X}} 
    \end{array}\right\},
\end{equation}
where $p_{y, a \mid z, \mathbf{X}}$ again stands for $P(Y = y, A = a \mid Z = z, \mathbf{X})$, and $p_{a \mid z, \mathbf{X}}$ is a shorthand for $P(A = a \mid Z = z, \mathbf{X})$. Observe that $P(A = a \mid Z = z, \mathbf{X}) = \sum_{y \in \{0, 1\}} P(Y = y, A = a \mid Z = z, \mathbf{X})$, and we can again estimate the lower bound and upper bound by estimating $\{P(Y = y,  A = a \mid Z = z, \mathbf{X} = \mathbf{x}),~y=\pm 1, a = \pm 1, z = \pm 1\}$.

\begin{remark}[Assumption Set $\mathcal{C}$]\rm
There are many other IV identification assumptions that help reduce the length of partial identification intervals in one way or another. Again, we would refer readers to \citet{baiocchi2014instrumental} and \citet{swanson2018partial} for bounds other than those considered above. We would like to point out that IV identification assumptions are typically not verifiable (although they might lead to testable implications), and depend largely on expert knowledge. Moreover, certain assumptions may be inappropriate in the context of ITR estimation problems, e.g., assumptions that largely restrict treatment heterogeneity, and should be made with caution.
\end{remark}

\begin{remark}[Continuous but Bounded $Y$]\rm
When $Y$ is continuous, partial identification bounds on $Y$ require additional assumptions that bound the support of $Y$. In Supplementary Material A, we further review assumptions and partial identification results that allow a valid IV to partially identify the counterfactual mean (and hence the ATE and CATE) of a continuous but bounded outcome.
\end{remark}

\section{An IV-Partial Identification Learning (IV-PILE) Approach to Estimating Optimal ITRs: IV-Optimality, Risk, and Optimization}

\label{sec: method and computation}
\subsection{IV-Optimality}
\label{subsec: IV-optimality}
Without loss of generality, we assume $\Delta=0$. Let $f(\cdot): \mathcal{X} \mapsto \mathbb{R}$ be a discriminant function and $\text{sgn}\{f(\cdot)\}$ a decision rule to be learned. Recall that the risk function to be minimized in an ITR estimation problem is \[
\mathcal{R}(f) = \mathbb{E}\left[|C(\mathbf{X})|\cdot \mathbbm{1}\big\{\text{sgn}\{f(\mathbf{X})\} \neq \text{sgn}\{C(\mathbf{X})\}\big\}\right].
\]As has been argued extensively, this optimal rule is in general \emph{not} identifiable when the collected observed covariates $\boldsymbol{X}$ cannot adequately address the confounding between the treatment and outcome.

To proceed, we define a new notion of optimality and a new estimand to target.
\begin{definition}[IV-Optimality]\rm
\label{def: IV-optimal}
A treatment rule $f(\cdot) \in \mathcal{F}$ is said to be IV-optimal if it is \emph{optimal with respect to the putative IV and assumption set $\mathcal{C}$} in the following sense:
\begin{equation*}
\begin{split}
 f &= \underset{f \in \mathcal{F}}{\text{argmin}} ~\mathcal{R}_{\text{upper}}(f;L(\cdot),U(\cdot)) \\
& =  \underset{f \in \mathcal{F}}{\text{argmin}}~\mathbb{E}\left[\sup_{C'(\mathbf{X}) \in [L(\mathbf{X}), U(\mathbf{X})]}  |C'(\mathbf{X})|\cdot \mathbbm{1}\big\{\text{sgn}\{f(\mathbf{X})\} \neq \text{sgn}\{C'(\mathbf{X})\}\big\}\right],
    \label{eqn: risk function with indicator}
\end{split}
\end{equation*}
where $[L(\mathbf{X}), U(\mathbf{X})]$ is the partial identification interval under the putative IV and identification assumption set $\mathcal{C}$.
\end{definition}

Proposition \ref{proposition: exchangeability} asserts that $\mathbb{E}[~\cdot~]$ and $\sup$ operators in Definition \ref{def: IV-optimal} are exchangeable.
\begin{proposition}\rm
\label{proposition: exchangeability}
\begin{equation*}
    \begin{split}
        f &= \underset{f \in \mathcal{F}}{\text{argmin}}~\mathbb{E}\left[\sup_{C'(\mathbf{X}) \in [L(\mathbf{X}), U(\mathbf{X})]}  |C'(\mathbf{X})|\cdot \mathbbm{1}\big\{\text{sgn}\{f(\mathbf{X})\} \neq \text{sgn}\{C'(\mathbf{X})\}\big\}\right] \\
&= \underset{f \in \mathcal{F}}{\text{argmin}}\sup_{C'(\cdot):~ L \preceq C' \preceq U} \mathcal{R}(f; C'(\cdot)),
    \end{split}
\end{equation*}
where $f_1 \preceq f_2$ denotes $f_1(\mathbf{x}) \leq f_2(\mathbf{x})$ for all $\mathbf{x}$. 
\end{proposition}
\begin{proof}
All proofs in the article are in Supplementary Materials C, D and E.
\end{proof}

\begin{remark}\rm
The risk function $\mathcal{R}_{\text{upper}}(f;L(\cdot),U(\cdot))$ considered in Definition \ref{def: IV-optimal} represents the expected worst-case weighted misclassification error among all $C'(\mathbf{X})$ compatible with $L(\mathbf{X})$ and $U(\mathbf{X})$ informed by the putative IV and IV identification assumptions. $\mathcal{R}_{\text{upper}}(f;L(\cdot),U(\cdot))$ is a natural upper bound on the risk $\mathcal{R}(f)$. Proposition \ref{proposition: exchangeability} further shows that $f$ in Definition \ref{def: IV-optimal} can be understood as a \emph{min-max} estimate.
\end{remark}

\begin{remark}\rm
When $C(\cdot) = L(\cdot) = U(\cdot)$, $\mathcal{R}_{\text{upper}}(f)$ would reduce to $\mathcal{R}(f)$, and IV-optimality reduces to the usual notation of optimality considered in \citet{zhang2012estimating}, \citet{zhao2012estimating}, \citet{cui2019semiparametric}, and \citet{qiu2020optimal}. 
\end{remark}

This new optimality criterion has at least three desirable features. First, it is always well-defined for any valid IV and under minimal IV identification assumptions. Second, it facilitates using IV identification assumptions as ``leading cases, not truths" (\citealp[Page 72]{tukey1986sunset}). According to Tukey, a statistical procedure is ``safe" if it is valid in a wide range of scenarios. The statistical procedure targeting the ``IV-optimal" rule is therefore ``safe" in the sense that the estimand is well-defined and can be learned under a \emph{wide range} of IV identification assumptions and mild modeling assumptions. In sharp contrast, \citet{cui2019semiparametric} and \citet{qiu2020optimal} aimed at learning the \emph{optimal} ITR with an IV; though the optimal ITR is always well-defined, it \emph{cannot} be learned even with a valid IV unless some often stringent IV identification assumptions are met. Third and perhaps most importantly, the notion of ``IV-optimality" leaves to IV identification assumptions ``the task of stringency" (\citealp[Page 72]{tukey1986sunset}), and captures the intuition that the quality of the estimated ITR should depend on the quality of the instrumental variable. According to Definition \ref{def: IV-optimal} and Proposition \ref{proposition: exchangeability}, an ``IV-optimal" ITR is more stringent, in the sense that it is ``closer" to the true underlying optimal ITR and has smaller risk and better generalization performance if the putative IV together with IV identification assumptions can help narrow down the partial identification intervals. This is the case, for instance, when the putative IV is a very strong one and the compliance rate is very high, or when assumptions in addition to the core IV assumptions, e.g., the correct non-compliant decision assumption (IV.A5), apply to the putative IV. We study more closely the risk of an IV-optimal ITR in Section \ref{subsec: Bayes rule and true generalization error}.

\begin{remark}[Multiple IVs and Weak IVs]\rm
In many empirical studies, researchers have multiple putative IVs, e.g., excess tuition \emph{and} excess distance in a study of the effect of community college on educational attainment (\citealp{rouse1995democratization}), and it is often unclear which one of these IVs, or if any of them, satisfies the point identification assumptions required in \citet{cui2019semiparametric} and \citet{qiu2020optimal} to identify the optimal ITR. However, these multiple IVs can be used to estimate their respective ``IV-optimal" ITRs, possibly under different, IV-specific, identification assumptions, and the quality of each resulting ``IV-optimal" ITR depends on how much each of these multiple IVs can narrow down the partial identification intervals and pinpoint the CATE. Multiple IVs can even be combined into a single stronger IV, and this stronger IV is likely to yield an ``IV-optimal" ITR that is more stringent and has better generalization performance compared to using any of the multiple IVs alone. On the other hand, if researchers only have a very weak IV, the corresponding partial identification intervals may be excessively long and non-informative, and as a result, the ``IV-optimal" ITR may be far from the optimal ITR in its generalization performance. Indeed, with a weak IV, researchers should expect little information to be learned about the treatment effect and perhaps the wisest thing to do is switching to a stronger IV.
\end{remark}

\begin{remark}\rm
\label{remark: we also develop plug-in estimators}
Although not the primary focus of this paper, one can directly minimize the expectation in Definition \ref{def: IV-optimal} in a pointwise manner by estimating $L(\mathbf{X})$ and $U(\mathbf{X})$, just like one can estimate $C(\mathbf{X})$ and then take $\text{sgn}\{\widehat{C}(\mathbf{X})\}$ to be the estimated optimal ITR in non-IV settings. These methods are called \emph{indirect methods} in the literature as they indirectly specify the form of the optimal ITR through postulated models for various aspects of $C(\mathbf{X})$ in non-IV settings (\citealp{zhao2019efficient}), and $L(\mathbf{X})$ and $U(\mathbf{X})$ in our setting. In Supplementary Material G, we construct simple plug-in estimators for an IV-optimal ITR based on this idea, and prove that this straightforward plug-in estimator is in fact minimax optimal. We pursue a classification perspective as in \citet{zhang2012estimating} and \citet{zhao2012estimating} here rather than the indirect methods because of the following consideration. In many practical scenarios, we would like to have control over the complexity of the estimated ITR. This in general cannot be fulfilled by indirect methods unless we specify some simple models to estimate $L(\mathbf{X})$ and $U(\mathbf{X})$ in the first place; however, $L(\mathbf{X})$ and $U(\mathbf{X})$ are unlikely to admit simple parametric forms in our settings as they are complicated combinations of maxima and/or minima of many conditional probabilities (see Section \ref{sec: examples of C and application to NICU}). On the other hand, if we use flexible machine learning tools to estimate conditional probabilities involved in $L(\mathbf{X})$ and $U(\mathbf{X})$, the corresponding ITR is often complicated and lacks interpretability. It may also suffer from the problem of overfitting (\citealp{zhao2012estimating}; \citealp{wang2016learning}; \citealp{zhao2019efficient}). These considerations motivate us to adopt the classification perspective as in \citet{zhang2012estimating} and \citet{zhao2012estimating}. A great appeal of the classification perspective is that by decoupling the task of estimating $L(\mathbf{X})$ and $U(\mathbf{X})$ from that of estimating the IV-optimal rule, the estimated ITR no longer suffers from the aforementioned problems. In principle, one can leverage flexible machine learning tools to estimate relevant conditional probabilities and hence $L(\mathbf{X})$ and $U(\mathbf{X})$, while still learning a parsimonious IV-optimal ITR within a pre-specified function class, e.g., the class of linear functions.
\end{remark}

\subsection{Bayes Decision Rule and Bayes Risk}
\label{subsec: Bayes rule and true generalization error}

Define the \emph{Bayes risk} $\mathcal{R}^\ast_{\text{upper}} = \inf_{f} \mathcal{R}_{\text{upper}}(f)$, where infimum is taken over all measurable functions. A decision rule $f$ is called a \emph{Bayes decision rule} if it attains the Bayes risk, i.e., $\mathcal{R}_{\text{upper}}(f^\ast) = \mathcal{R}^\ast_{\text{upper}}$. Proposition \ref{prop: Bayes decision rule} gives a representation of the Bayes decision rule.

\begin{proposition}\rm
\label{prop: Bayes decision rule}
Consider the risk function $\mathcal{R}_{\text{upper}}(f)$ defined in Definition \ref{def: IV-optimal}. Let
\[
\eta(\mathbf{x}) = |U(\mathbf{x})|\cdot\mathbbm{1}\big\{L(\mathbf{x}) > 0\big\} - |L(\mathbf{x})|\cdot\mathbbm{1}\big\{U(\mathbf{x}) < 0\big\} + (|U(\mathbf{x})| - |L(\mathbf{x})|)\cdot\mathbbm{1}\big\{[L(\mathbf{x}), U(\mathbf{x})] \ni 0\big\}.
\]
Consider a decision rule $f^\ast(\mathbf{x})$ such that 
\[
    \text{sgn}\{f^\ast(\mathbf{x})\} = \text{sgn}\{\eta(\mathbf{x})\} = \begin{cases}
+1, \quad&\text{if}~~\eta(\mathbf{x})\geq 0,\\
-1, &\text{if}~~\eta(\mathbf{x})< 0.
\end{cases}
\]
Let $\mathcal{R}^\ast_{\text{upper}} = \mathcal{R}_{\text{upper}}(f^\ast)$. $f^\ast$ is the Bayes decision rule and $\mathcal{R}^\ast_{\text{upper}}$ is the Bayes risk such that \[
\mathcal{R}_{\text{upper}}(f) \geq \mathcal{R}^\ast_{\text{upper}},\forall ~f~\text{measurable}.
\]
\end{proposition}

Proposition \ref{prop: excess risk} further derives the excess risk of a measurable decision rule $f$.
\begin{proposition}\rm
\label{prop: excess risk}
For any measurable decision rule $f$, its excess risk is
\begin{equation*}
    \mathcal{R}_{\text{upper}}(f)-\mathcal{R}^\ast_{\text{upper}}=\mathbb{E}\left[\mathbbm{1}\big\{\text{sgn}\{f(\mathbf{X})\}\neq \text{sgn}\{f^\ast(\mathbf{X})\}\big\}\cdot|\eta(\mathbf{X})|\right],
\end{equation*}
where $\eta(\mathbf{x})$ is defined in Proposition \ref{prop: Bayes decision rule}.
\end{proposition}

\subsection{Risk of IV-Optimal Rules}
An IV-optimal rule targets $\mathcal{R}_{\text{upper}}$. What can be said about the risk of an IV-optimal rule? Proposition \ref{prop: approx error} provides insight into this important question.
\begin{proposition}[Risk of IV-Optimal Rules]\label{prop: approx error}\rm
\leavevmode
\begin{enumerate}
    \item For any measurable $f$, we have
\begin{equation*}
    0\leq \mathcal{R}_{\text{upper}}(f)-\mathcal{R}(f)\leq \mathbb{E}\left[U(\mathbf{X})-L(\mathbf{X})\right].    
\end{equation*}
\item Let $f^{\ast}$ be the Bayes decision rule targeting $\mathcal{R}_{\text{upper}}$ in Proposition \ref{prop: Bayes decision rule} such that $\mathcal{R}^\ast_{\text{upper}} = \mathcal{R}_{\text{upper}}(f^\ast)$. The risk of $f^\ast$ satisfies
\begin{equation*}\small
\begin{split}
&\mathcal{R}(f^{\ast}) = \mathbb{E}\left[|C(\mathbf{X})|\cdot \mathbbm{1}\big\{\text{sgn}\{C(\mathbf{X})\} \neq \text{sgn}\{f^\ast(\mathbf{X})\}\big\}\right] \\
\leq &\mathbb{E}\left[\underbrace{\mathbbm{1}\left\{L(\mathbf{X})<0<U(\mathbf{X})\right\}}_{\text{\rom{1}}}\cdot\underbrace{\left\{U(\mathbf{X})-L(\mathbf{X})\right\}}_{\text{\rom{2}}}\cdot
\underbrace{\left\{\frac{1-\rho(\mathbf{X};U, L)}{2}\right\}}_{\text{\rom{3}}}\cdot\underbrace{\mathbbm{1}\left\{\rho^{c}(\textbf{X};U,L,C)>\rho(\textbf{X};U,L)\right\}}_{\text{\rom{4}}}\right],
\end{split}
\end{equation*}
where $\rho(\mathbf{x};U,L)=\frac{|U(\mathbf{x})+L(\mathbf{x})|/2}{|U(\mathbf{x})-L(\mathbf{x})|/2}$ and $\rho^{c}(\mathbf{x};U,L,C)=\left|\frac{C(\mathbf{x})-(U(\mathbf{x})+L(\mathbf{x}))/2}{|U(\mathbf{x})-L(\mathbf{x})|/2}\right|$.
\end{enumerate}
\end{proposition}

The first part of Proposition \ref{prop: approx error} states that for any decision rule $f$, $\mathcal{R}_{\text{upper}}(f)$ is no larger than the risk $\mathcal{R}(f)$ by a margin of $\mathbb{E}[U(\mathbf{X}) - L(\mathbf{X})]$. From this result, it is transparent that as $U(\mathbf{x})$ converges uniformly to $L(\mathbf{x})$, the gap between $\mathcal{R}_{\text{upper}}(f)$ and $\mathcal{R}(f)$ goes to $0$ uniformly in $f$.

The second part of the proposition further states that if $f^\ast$ is a Bayes decision rule for $\mathcal{R}_{\text{upper}}$ that attains the Bayes risk $\mathcal{R}^\ast_{\text{upper}}$, its generalization error $\mathcal{R}(f^\ast)$ is upper bounded by the expectation of the product of four terms, each of which bears its own meaning. Fix $\mathbf{x} \in \mathcal{X}$ and the partial identification interval $[L(\mathbf{x}), U(\mathbf{x})]$. The first term \rom{1} $= \mathbbm{1}\{L(\mathbf{x}) < 0 < U(\mathbf{x})\}$ measures if the interval $[L(\mathbf{x}), U(\mathbf{x})]$ covers $0$. If $[L(\mathbf{x}), U(\mathbf{x})]$ does not cover $0$, such an $\mathbf{x}$ would not contribute to the risk of $f^\ast$. The second term \rom{2} $= U(\mathbf{x}) - L(\mathbf{x})$ measures the length of the interval. Not surprisingly, if the interval $[L(\mathbf{x}), U(\mathbf{x})]$ covers $0$, the narrower it is, the less it would contribute to the risk of $f^\ast$. The third term \rom{3} $= \{1 - \rho(\mathbf{x}; U, L)\}/2$ measures how symmetric $[L(\mathbf{x}), U(\mathbf{x})]$ is about $0$. Suppose that $[L(\mathbf{x}), U(\mathbf{x})]$ is such that $L(\mathbf{x}) < 0 < U(\mathbf{x})$, i.e., $[L(\mathbf{x}), U(\mathbf{x})]$ covers $0$ (the first term \rom{1} is $1$) and has a nontrivial interval length (the second term \rom{2} is not $0$). If $[L(\mathbf{x}), U(\mathbf{x})]$ is symmetric about $0$, i.e., $L(\mathbf{x}) = -U(\mathbf{x})$, $\rho(\mathbf{x}; U, L)$ would attain its minimum at $0$; on the other hand, $\rho(\mathbf{x}; U, L)$ could be arbitrarily close to $1$ if either $L(\mathbf{x})$ is arbitrarily close to $0$ from the left, or $U(\mathbf{x})$ is arbitrarily close to $0$ from the right. In other words, $\rho(\mathbf{x}; U, L)$ (and hence the third term \rom{3}) measures the skewness of the interval $[L(\mathbf{x}), U(\mathbf{x})]$ with respect to $0$. The fourth term \rom{4} measures how symmetric $[L(\mathbf{x}), U(\mathbf{x})]$ is about $C(\mathbf{x})$, relative to how symmetric the interval is around $0$. Observe that $\rho^{c}(\mathbf{x};U,L,C)$ is analogous to $\rho(\mathbf{x}; U, L)$, except that $\rho^{c}(\mathbf{x};U,L,C)$ measures the skewness of the interval $[L(\mathbf{x}), U(\mathbf{x})]$ with respect to $C(\mathbf{x})$. 
$\rho^{c}(\mathbf{x};U,L,C)$ would attain its minimum at $0$ if the interval is symmetric about $C(\mathbf{x})$, and its maximum at $1$ if the interval barely covers $C(\mathbf{x})$, i.e., $L(\mathbf{x}) = C(\mathbf{x})$ or $U(\mathbf{x}) = C(\mathbf{x})$. Therefore, if the interval is more symmetric about $C(\mathbf{x})$ than it is about $0$, the fourth term \rom{4} is $0$; otherwise, it is $1$. To conclude, the risk of $f^\ast$ is small if with high probability, the partial identification interval does not cover $0$, is short in length, is asymmetric about $0$, and is more symmetric about $C(\mathbf{x})$ than about $0$. This upper bound on the risk of $f^\ast$ can also be understood as the maximum gap between the generalization performance of $f^\ast$ and an optimal ITR. Figure \ref{fig: bad/good performance} summarizes the above discussion with a graphical illustration. 
\vspace{-0.3 cm}
\begin{figure}[ht]
   \centering
     \subfloat{\includegraphics[width = 0.47\columnwidth]{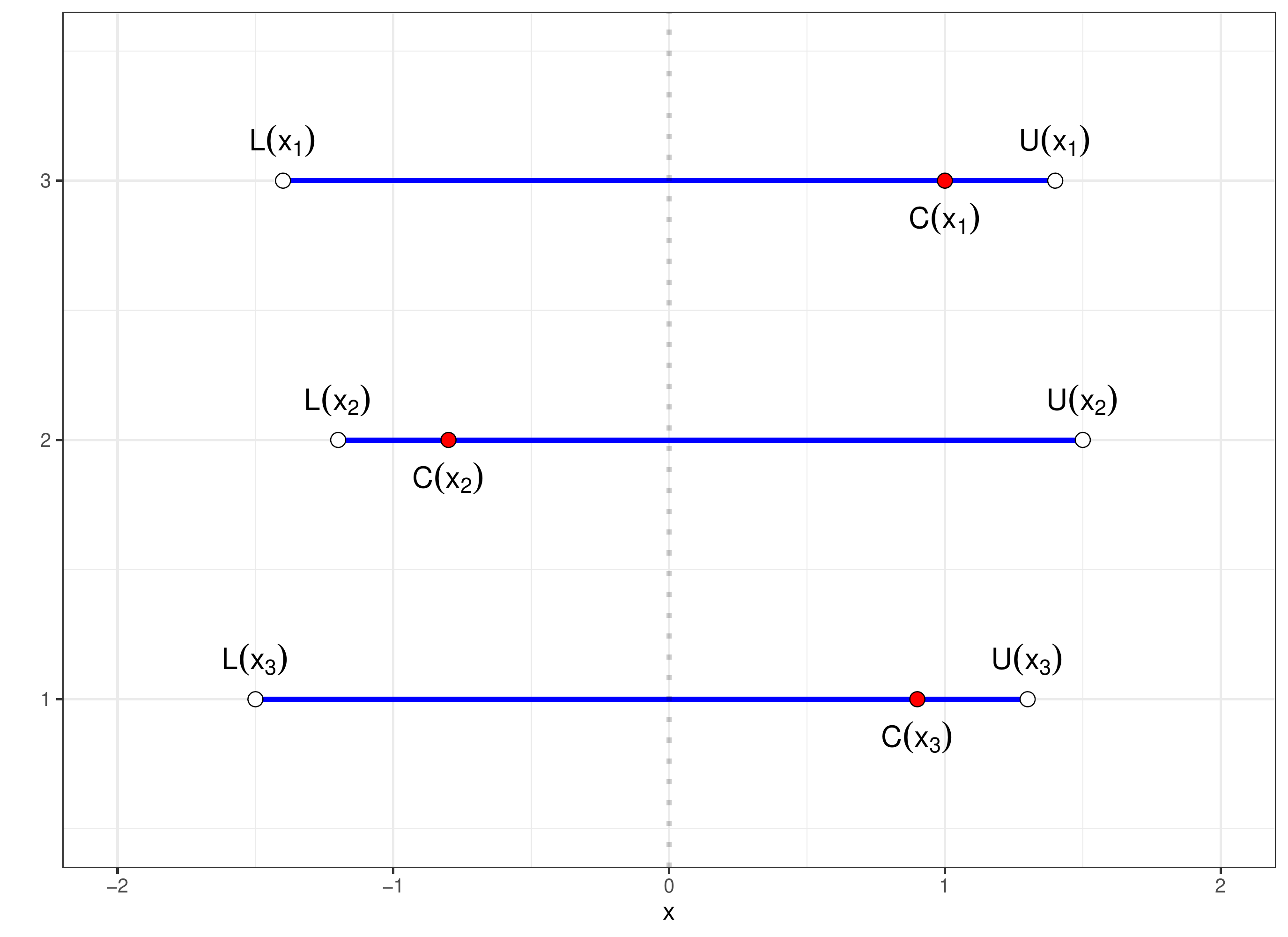}}
     \subfloat{\includegraphics[width = 0.47\columnwidth]{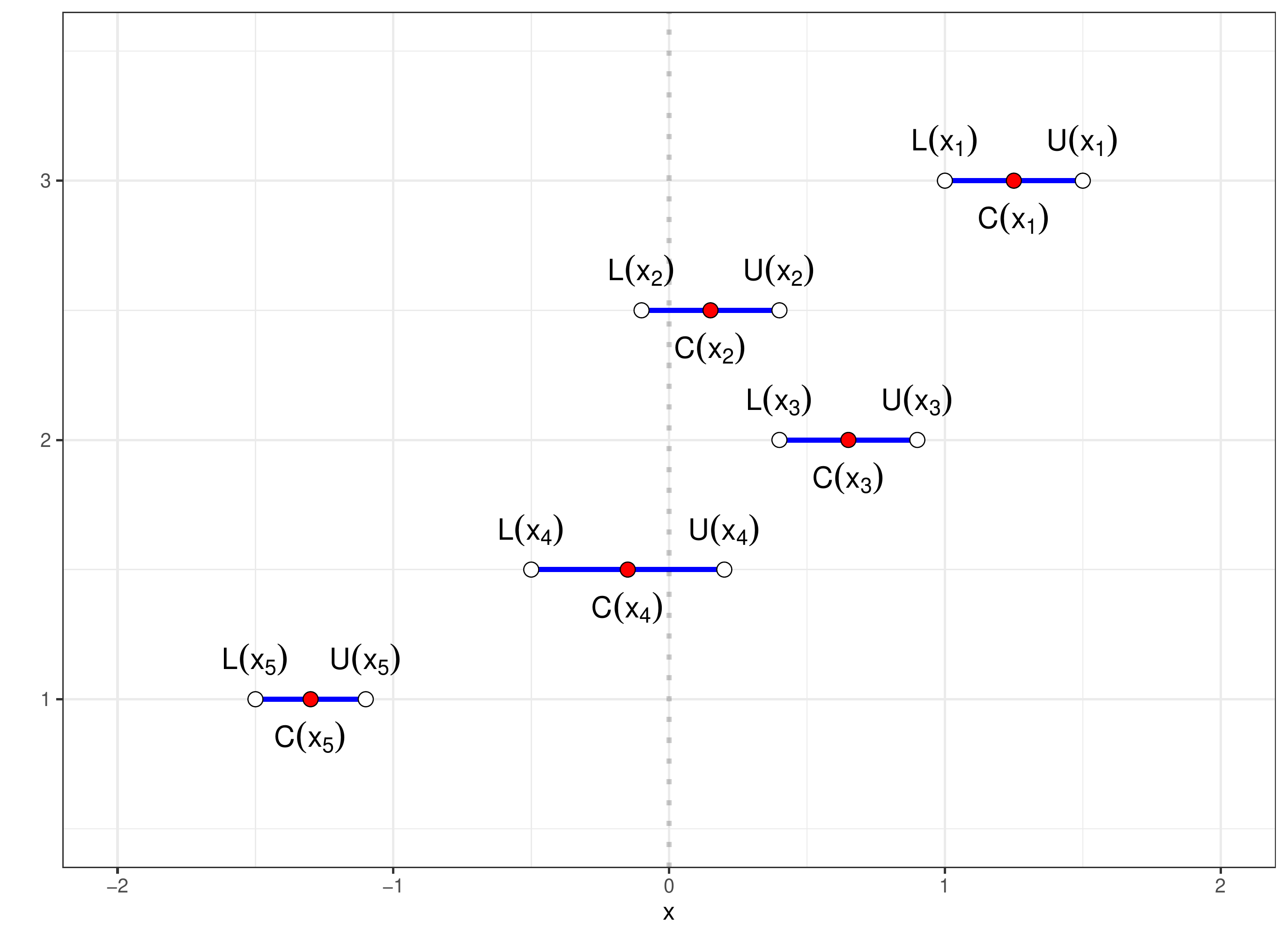}}
     \caption{An illustration of the second part of Proposition \ref{prop: approx error}. The left panel plots a scenario when IV-optimal Bayes rule $f^\ast_{\text{Bayes}}$ may have a large risk: partial identification intervals cover $0$, are excessively long, symmetric about $0$, and asymmetric about $C(\mathbf{x})$. The right panel plots a scenario with favorable generalization performance: many partial identification intervals avoid $0$, are short in length, asymmetric about $0$, and symmetric about $C(\mathbf{x})$.}
\label{fig: bad/good performance}
\end{figure}

\subsection{Risk Decomposition, Structural Risk Minimization, and Surrogate Loss}
\label{subsec: risk decomp and structural risk minimization}
The risk function $\mathcal{R}_{\text{upper}}(f)$ can be decomposed into two parts: $\mathcal{R}_{\text{label, upper}}(f) $, corresponding to the risk associated with the labeled part, and $\mathcal{R}_{\text{unlabel, upper}}(f) $, corresponding to that associated with the unlabeled part: 
\small\begin{equation}
\label{eqn: risk decomp}
\begin{split}
    &\mathcal{R}_{\text{upper}}(f) = \mathbb{E}\left[\max_{C'(\mathbf{X}) \in [L(\mathbf{X}), U(\mathbf{X})] }  |C'(\mathbf{X})|\cdot \mathbbm{1}\big\{\text{sgn}\{f(\mathbf{X})\} \neq \text{sgn}\{C'(\mathbf{X})\}\big\}\right]  \\
    =&~ \underbrace{\mathbb{E}\left[|U(\mathbf{X})| \cdot \mathbbm{1}\big\{\text{sgn}\{f(\mathbf{X})\} \neq 1\big\}\cdot\mathbbm{1}\{L(\mathbf{X}) > 0\} \right] + \mathbb{E}\left[|L(\mathbf{X})| \cdot \mathbbm{1}\big\{\text{sgn}\{f(\mathbf{X})\} \neq -1\big\}\cdot\mathbbm{1}\big\{U(\mathbf{X}) < 0\big\} \right] }_{\mathcal{R}_{\text{label, upper}}(f)} \\
    +&~\underbrace{\mathbb{E}\bigg[\max\left[|U(\mathbf{X})| \cdot \mathbbm{1}\big\{\text{sgn}\{f(\mathbf{X})\} \neq 1\big\},~ |L(\mathbf{X})| \cdot \mathbbm{1}\big\{\text{sgn}\{f(\mathbf{X})\} \neq -1\big\}\right]\cdot \mathbbm{1}\big\{[L(\mathbf{X}), U(\mathbf{X})] \ni 0\big\}\bigg]}_{\mathcal{R}_{\text{unlabel, upper}}(f)}\\
    =&~ \mathcal{R}_{\text{label, upper}}(f) + \mathcal{R}_{\text{unlabel, upper}}(f).
\end{split}
\end{equation}

\begin{remark}\rm
It may be tempting to replace $\max_{C'(\mathbf{X}) \in [L(\mathbf{X}), U(\mathbf{X})]}$ with $\min_{C'(\mathbf{X}) \in [L(\mathbf{X}), U(\mathbf{X})]}$ in Definition \ref{def: IV-optimal}; however, the definition would then become vacuous as is easily seen from \eqref{eqn: risk decomp}. Fix $\mathbf{x} \in \mathcal{X}$ such that $[L(\mathbf{x}), U(\mathbf{x})] \ni 0$. The risk conditional on $\mathbf{x}$ is always minimized by letting $\text{sgn}\{f(\mathbf{x})\} = \text{sgn}\{C'(\mathbf{x})\}$; however, since $[L(\mathbf{x}), U(\mathbf{x})] \ni 0$ and $C'(\mathbf{x}) \in [L(\mathbf{x}), U(\mathbf{x})]$, $\text{sgn}\{C'(\mathbf{x})\}$ can be either $+1$ or $-1$, suggesting that the risk conditional on $\mathbf{x}$ is always $0$ no matter what value $f(\mathbf{x})$ takes on. In other words, $\mathcal{R}_{\text{unlabel, upper}}(f) = 0 ~\forall f$, with $\max$ replaced with $\min$ in \eqref{eqn: risk decomp}, and unlabeled data become superfluous.
\end{remark}

Decomposition (\ref{eqn: risk decomp}) motivates estimating $f(\cdot) \in \mathcal{F}$ using the following structural risk minimization approach (\citealp{vapnik1992principles}):
\begin{equation}\small
\begin{split}
    \label{eqn: empirical version of risk function indicator}
    \widehat{f}(\cdot) = \underset{f \in \mathcal{F}}{\text{argmin}} &\sum_{i = 1}^l\widehat{U}(\mathbf{X}_i)\cdot
    \mathbbm{1}\big\{\text{sgn}\{f(\mathbf{X}_i)\} \neq 1\big\}\cdot \mathbbm{1}\big\{\widehat{L}(\mathbf{X}_i) > 0\big\}\\
    &+ \{-\widehat{L}(\mathbf{X}_i)\}\cdot \mathbbm{1}\big\{\text{sgn}\{f(\mathbf{X}_i)\} \neq -1\big\}\cdot \mathbbm{1}\big\{\widehat{U}(\mathbf{X}_i) < 0\big\}\\
    & + \sum_{i = l+1}^n \max\left[\widehat{U}(\mathbf{X}_i) \cdot \mathbbm{1}\big\{\text{sgn}\{f(\mathbf{X}_i)\} \neq 1\}, ~-\widehat{L}(\mathbf{X}_i) \cdot \mathbbm{1}\big\{\text{sgn}\{f(\mathbf{X}_i)\} \neq -1\big\}\right]\\
    & +\frac{n\lambda_n}{2} ||f||^{2},
    \end{split}
\end{equation}
where $[\widehat{L}(\mathbf{X}_i), \widehat{U}(\mathbf{X}_i)]$ is an estimated partial identification interval of $[L(\mathbf{X}_i), U(\mathbf{X}_i)]$, $\|\cdot\|$ denotes some norm of $f(\cdot)$, and $\lambda_n$ is a possibly data-dependent tuning parameter that penalizes the norm of $f(\cdot)$ to reduce overfitting. For instance, if we assume that $f(\cdot)$ resides in a \emph{reproducing kernel Hilbert space} (RKHS) $\mathcal{H}_{\mathcal{K}}$, then $\|\cdot\|$ corresponds to the norm associated with $\mathcal{H}_{\mathcal{K}}$. The complexity of $f(\cdot)$ is restricted by penalizing its norm.

It is well known in machine learning and optimization literature that directly minimizing the empirical risk as above is difficult due to the non-continuity and non-convexity of the indicator function, and it is customary to rewrite the loss function by replacing the \emph{0-1-based loss} with a convex upper bound. Table \ref{tbl: loss and surrogate loss} summarizes the original 0-1-based loss and our choice of surrogate loss corresponding to each $[L(\mathbf{x}), U(\mathbf{x})]$ configuration. Figure \ref{fig: surrogate loss and 0-1 loss} in Supplementary Material B.1 further plots the original loss and the corresponding surrogate loss in each case. Observe that the surrogate loss is indeed a continuous convex upper bound of the original discontinuous loss function in all cases. Moreover, it can be shown that our designed surrogate loss function is continuous in both $L$ and $U$ values.

\begin{table}[h!]
\centering
\caption{\small Original 0-1-based loss and the corresponding surrogate loss. }
\label{tbl: loss and surrogate loss}
 \resizebox{\columnwidth}{!}{
 \fbox{
 \begin{tabular}{cp{51mm}c} 
 \\ [-0.8em]
 $[L(\mathbf{x}), U(\mathbf{x})]$ & ~~~~~~Original Loss & Surrogate Loss  \\[0.5 em]
 \hline \\
 $[L(\mathbf{x}), U(\mathbf{x})] > 0$&$|U(\mathbf{x})|\cdot\mathbbm{1}\{\text{sgn}\{f(\mathbf{x})\} \neq 1\}$&$|U(\mathbf{x})|\cdot \{1 - f(\mathbf{x})\}^+$\\ \\
 $[L(\mathbf{x}), U(\mathbf{x})] < 0$&$|L(\mathbf{x})|\cdot\mathbbm{1}\{\text{sgn}\{f(\mathbf{x})\} \neq -1\}$&$|L(\mathbf{x})|\cdot \{1 + f(\mathbf{\mathbf{x}})\}^+$\\ \\
 $[L(\mathbf{x}), U(\mathbf{x})] \ni 0$, $|U(\mathbf{x})| \geq |L(\mathbf{x})|$ & $\max[U(\mathbf{x}) \cdot \mathbbm{1}\{\text{sgn}\{f(\mathbf{x})\} \neq 1\},\newline~ -L(\mathbf{x}) \cdot \mathbbm{1}\{\text{sgn}\{f(\mathbf{x})\} \neq -1\}]$&$|L(\mathbf{x})| + (|U(\mathbf{x})| - |L(\mathbf{x})|)\cdot\{1 - f(\mathbf{x})\}^+$ \\ \\
 $[L(\mathbf{x}), U(\mathbf{x})] \ni 0$, $|U(\mathbf{x})| < |L(\mathbf{x})|$ & $\max[U(\mathbf{x}) \cdot \mathbbm{1}\{\text{sgn}\{f(\mathbf{x})\} \neq 1\},\newline~ -L(\mathbf{x}) \cdot \mathbbm{1}\{\text{sgn}\{f(\mathbf{x})\} \neq -1\}]$&$|U(\mathbf{x})| + (|L(\mathbf{x})| - |U(\mathbf{x})|)\cdot\{1 + f(\mathbf{x})\}^+$ \\ \\[-0.5em]
 \end{tabular}}}
\end{table}

\begin{remark}\rm
When $[L(\mathbf{x}), U(\mathbf{x})] > 0$ or $[L(\mathbf{x}), U(\mathbf{x})] < 0$, the surrogate loss is a scaled hinge loss. When $[L(\mathbf{x}), U(\mathbf{x})] \ni 0$, the surrogate loss is a lifted and scaled hinge loss.
\end{remark}

Let $\phi(x) = (1 - x)^+$. Under the surrogate loss, the objective function (\ref{eqn: empirical version of risk function indicator}) becomes:
\begin{equation}
\begin{split}
    \label{eqn: empirical version of risk function}
    \widehat{f}(\cdot) = \underset{f \in \mathcal{F}}{\text{argmin}} \sum_{i = 1}^l &\left[\widehat{U}(\mathbf{X}_i)\cdot \phi\{f(\mathbf{X}_i)\}\cdot \mathbbm{1}\big\{\widehat{L}(\mathbf{X}_i) > 0\big\} +
    \{-\widehat{L}(\mathbf{X}_i)\}\cdot \phi\{-f(\mathbf{X}_i)\}\cdot \mathbbm{1}\big\{\widehat{U}(\mathbf{X}_i) < 0\big\}\right] \\
    + \sum_{i = l+1}^n &\bigg[\left[|\widehat{L}(\mathbf{X}_i)| + (|\widehat{U}(\mathbf{X}_i)| - |\widehat{L}(\mathbf{X}_i)|)\cdot\phi\{f(\mathbf{X}_i)\}\right]\cdot\mathbbm{1}\big\{|\widehat{U}(\mathbf{X}_i)| \geq |\widehat{L}(\mathbf{X}_i)|\big\}\\
    +&\left[|\widehat{U}(\mathbf{X}_i)| + (|\widehat{L}(\mathbf{X}_i)| - |\widehat{U}(\mathbf{X}_i)|)\cdot\phi\{-f(\mathbf{X}_i)\}\right]\cdot\mathbbm{1}\big\{|\widehat{L}(\mathbf{X}_i)| > |\widehat{U}(\mathbf{X}_i)|\big\}\bigg] +\frac{n\lambda_n}{2} ||f||^{2}.
    \end{split}
\end{equation}

Let $\mathcal{R}^{\text{h}}_{\text{upper}}(f)$ denote the risk associated with the surrogate loss, $f^{\ast}_{\text{h}}$ the Bayes decision rule that minimizes $\mathcal{R}^{\text{h}}_{\text{upper}}(f)$, and $\mathcal{R}^{\text{h}, \ast}_{\text{upper}}(f)$ the corresponding Bayes risk. Theorem \ref{thm: relation between surrogate and true loss} establishes a relationship between $\mathcal{R}_{\text{upper}}(f) - \mathcal{R}^\ast_{\text{upper}}$ and $\mathcal{R}^{\text{h}}_{\text{upper}}(f) - \mathcal{R}^{\text{h},\ast}_{\text{upper}}$, so that we can transfer assessments of statistical error in terms of the excess risk $\mathcal{R}^{\text{h}}_{\text{upper}}(f) - \mathcal{R}^{\text{h},\ast}_{\text{upper}}$ into assessments of error in terms of $\mathcal{R}_{\text{upper}}(f) - \mathcal{R}^{\ast}_{\text{upper}}$, the excess risk of genuine interest (\citealp{bartlett2006convexity}).

\begin{theorem}\rm
\label{thm: relation between surrogate and true loss}
For any measurable function $f$, we have
\begin{equation}
    \mathcal{R}_{\text{upper}}(f) - \mathcal{R}^{\ast}_{\text{upper}} \leq \mathcal{R}^{\text{h}}_{\text{upper}}(f) - \mathcal{R}^{\text{h},\ast}_{\text{upper}}.
\end{equation}
\end{theorem}

Theorem \ref{thm: relation between surrogate and true loss} reassures us that using the surrogate loss displayed in Table \ref{tbl: loss and surrogate loss} does not hinder the search for a function that achieves the optimal Bayes risk $\mathcal{R}^\ast_{\text{upper}}$, and it is appropriate to employ surrogate-loss-based computationally efficient algorithms. In Supplementary Materials B.2 and B.3, we derive linear/nonlinear $\widehat{f}$ when $f(\cdot)$ is in a reproducing kernel Hilbert space and show that the associated optimization problem can be transformed into a particular instance of weighted SVM (\citealp{vapnik2013nature}) and readily solved using standard solvers.

\subsection{IV-PILE Algorithm}
\label{subsec: IV pile alg}

Before delving into theoretical properties, we summarize the IV-PILE algorithm in Algorithm \ref{alg}.

\begin{algorithm} 
\SetAlgoLined
\caption{Pseudo Algorithm for IV-PILE} \label{alg}
\vspace*{0.12 cm}
\KwIn{$\{(\mathbf{X}_i, Z_i, A_i, Y_i), i=1,\cdots,n\}$ \text{and} IV identification assumption set $\mathcal{C}$;}
\vspace*{0.12 cm}
\ShowLn Obtain appropriate estimates of $L(\mathbf{X}_i)$ and $U(\mathbf{X}_i)$, denoted as $\widehat{L}_i$ and $\widehat{U}_i$, under IV identification assumption set $\mathcal{C}$. Parametric models or more flexible and expressive estimators like random forests can be used.\\
\vspace*{0.12 cm}
\ShowLn Compute the label $\widehat{e}_i \in \{-1, +1\}$ associated with each observation:
\[
\widehat{e}_i = \mathbbm{1}\big\{\widehat{U}_i < 0\big\} - \mathbbm{1}\big\{\widehat{L}_i > 0\big\} - \text{sgn}\{|\widehat{U}_i| - |\widehat{L}_i|\}\cdot \mathbbm{1}\big\{[\widehat{L}_i, \widehat{U}_i] \ni 0\big\},
\]
for $i = 1, \cdots, n$;\\
\vspace*{0.12 cm}
\ShowLn Compute the weight $\widehat{w}_i$ associated with each observation:
\begin{equation*}
    \begin{split}
        \widehat{w}_i = |\widehat{U}_i|\cdot\mathbbm{1}\big\{\widehat{L}_i > 0\big\} + |\widehat{L}_i|\cdot\mathbbm{1}\big\{\widehat{U}_i < 0 \big\}
        + \big||\widehat{U}_i| - |\widehat{L}_i |\big|\cdot\mathbbm{1}\big\{[\widehat{L}_i, \widehat{U}_i] \ni 0\big\},
    \end{split}
\end{equation*}
for $i = 1, \cdots, n$;\\
\vspace*{0.12 cm}
\ShowLn Solve a weighted SVM problem with labels and weights computed in Step 2 and 3 using a Gaussian kernel. Let $\widehat{f}$ be the solution;\\
\vspace*{0.12 cm}
\ShowLn Return $\widehat{f}$.
\end{algorithm}

\section{Theoretical Results}
\label{sec: theory}
\subsection{IV-PILE Estimator via Sample Splitting}
To facilitate theoretical analysis of the IV-PILE estimator, we study an alternative sample-splitting estimator that is very close to the IV-PILE estimator. Let $I_{1},I_{2}$ denote an equal-size mutually exclusive random partition of indices $\{1,\dots,n\}$ such that $|I_{1}|\asymp |I_{2}|\asymp n/2$. Samples with indices in $I_{1}$ are used to construct estimates $\widehat{L}(\mathbf{x})$ and $\widehat{U}(\mathbf{x})$ for functions $L(\mathbf{x})$ and $U(\mathbf{x})$. We then plug $\widehat{L}$ and $\widehat{U}$ into expressions for the weight $w$ and label $e$ (see Algorithm \ref{alg} and Supplementary Material B.2) to construct $\widehat{w}(\cdot)$ and $\widehat{e}(\cdot)$, and use the other half of samples to obtain the following IV-PILE estimator:
\begin{equation*}
\begin{split}
  \widehat{f}_{n}^{\lambda_n} =  \underset{f\in \mathcal{F}}{\text{argmin}}~&\frac{1}{|I_2|}\sum_{i\in I_2} \widehat{w}(\mathbf{X}_i) \cdot\{1 +\widehat{e}(\mathbf{X}_i)\cdot f( \mathbf{X}_i)\}^+ + \frac{\lambda_{n}}{2}||f||^{2}. 
\end{split}    
\end{equation*}
Sample splitting here helps remove the dependence between estimating $\widehat{w}$ and $\widehat{e}$ and constructing the IV-PILE estimator, which in turn helps weaken the assumptions needed to establish convergence rate results by getting rid of the entropy conditions on $L(\mathbf{x})$ and $U(\mathbf{x})$'s function classes. Similar sample splitting technique can be also found in \citet{bickel1982adaptive}, \citet{zheng2011cross}, \citet{chernozhukov2016double},
\citet{robins2017minimax}, and \citet{zhao2019efficient}, among many others.

\subsection{Theoretical Properties}
Let $\inf_{f\in \mathcal{F}}\mathcal{R}^{\text{h}}_{\text{upper}}(f)$ denote the minimal risk among rules in $\mathcal{F}$ and define the approximation error incurred by optimizing over $\mathcal{F}$ as 
$\mathcal{A}(\mathcal{F})= \inf_{f\in \mathcal{F}}\mathcal{R}^{\text{h}}_{\text{upper}}(f)-\mathcal{R}^{\text{h},\ast}_{\text{upper}}(f)$.
In this section we establish the convergence rate properties of $\mathcal{R}_{\text{upper}}(\widehat{f}_{n}^{\lambda_n}) - \inf_{f\in \mathcal{F}}\mathcal{R}^{\text{h}}_{\text{upper}}(f)$. We consider the following assumptions.
\begin{assumption}[Existence of a Finite Minimizer]\rm
\label{assump: existence of fnite minimizer}
\[\exists f_\mathcal{F}^\ast\in \mathcal{F} ~\text{s.t.}~ \mathcal{R}^{\text{h}}_{\text{upper}}(f_\mathcal{F}^\ast) = \inf_{f\in \mathcal{F}}\mathcal{R}^{\text{h}}_{\text{upper}}(f).
\]
\end{assumption}

\begin{assumption}[Boundedness Conditions I]\rm
\label{assump: boundedness}
\[
\exists M_1 >0\ s.t.\ |L(\mathbf{X})|,|U(\mathbf{X})|,|Y|\leq M_1~\text{with probability}~1,
\]
\[
\exists M_2 ~\text{s.t.}~\forall i, |\mathbf{X}[i]|\leq M_2 ~\text{with probability}~1.
\]
\end{assumption}

\begin{assumption}[Boundedness Conditions II]\rm
\label{assump: boundedness of L U estimates}
Assume that the estimates of $L$ and $U$, i.e., $\widehat{L}$ and $\widehat{U}$, satisfy
\[\exists M_3 ~\text{s.t.}~ |\widehat{L}(\mathbf{X})|,|\widehat{U}(\mathbf{X})|\leq M_3 ~\text{with probability}~1.
\]
\end{assumption}
For any $\epsilon >0$, let $N\{\epsilon,\mathcal{F},L_{2}(P)\}$ denote the covering number of $\mathcal{F}$, i.e., $N\{\epsilon,\mathcal{F},L_{2}(P)\}$ is the minimal number of closed $L_{2}(P)$-balls of radius $\epsilon$ that is required to cover $\mathcal{F}$.  We consider the following assumption on entropy condition:
\begin{assumption}[Entropy Condition]\rm
\label{assump: entropy}
There exists a subclass $\mathcal{F}_0$ of $\mathcal{F}$ and  constants $C_0>0$, $0<v<2$ such that $\mathcal{F}=\{a\cdot f:f\in \mathcal{F}_0,a\in\R\}$, the minimizer $\widehat{f}_{n}^{\lambda_n}$ is in $(C_0/\lambda_n)\cdot\mathcal{F}_0\coloneqq\{a\cdot f:f\in \mathcal{F}_0, a\in \R, |a|\leq C_0\cdot \lambda^{-1}_n\}$ with probability 1, and the $\sup_{P}\log N\{\epsilon,\mathcal{F}_0,L_2(P)\}\leq O(\epsilon^{-v})$ for all $0<\epsilon<1$, where the supremum is taken over all probability measures $P$.
\end{assumption}
Assumption \ref{assump: existence of fnite minimizer} says that there exists an $f$ with finite norm that minimizes the risk over $\mathcal{F}$. This is a standard assumption made in the statistical learning literature. Assumption \ref{assump: boundedness} requires that $L$, $U$, and $Y$ are bounded, and coordinates of observed covariates $\mathbf{X}$ are bounded with probability $1$. When $Y$ is a binary outcome, e.g., mortality as in the NICU study, $L$ and $U$ obtained via the Balke-Pearl bound and the Siddique bound are trivially bounded. When $Y$ is continuous, the partial identification literature typically requires $Y$ to be bounded (\citealp{swanson2018partial}), and the partial identification interval endpoints $L$ and $U$ are therefore also bounded. Boundedness of $Y$ is reasonable for many health outcomes, e.g., length of stay in hospital, the cholesterol level, etc. Assumption \ref{assump: boundedness of L U estimates} says that estimates $\widehat{L}$ and $\widehat{U}$ are bounded when $L$ and $U$ are bounded. This holds for any reasonable estimates of $L$ and $U$. Finally, the assumption on entropy condition is satisfied for many popular RKHS, e.g., the RKHS induced by the Gaussian kernel $k(x, x') := \exp(-\|x - x'\|^2/\sigma^2)$, where we take $\mathcal{F}_0$ to be the convex hull of $\{f_{x'}:f_{x'}(x)=k(x,x'),x'\in \mathcal{X}\}$, the minimizer $\widehat{f}_{n}^{\lambda_n}$ takes an explicit form that satisfies the specified condition (see Supplementary Material B), and the covering number condition follows from \citet[Corollary 9.5]{kosorok2007introduction}.

\begin{assumption}[Rate of Convergence of $L$ and $U$]\rm
\label{assump: rate of convergence of L and U}
Assume that $\widehat{L}$ and $\widehat{U}$ converge to  $L$ and $U$ at the following rates:
$$\mathbb{E}\left[\big|\widehat{L}(\mathbf{X}) - L(\mathbf{X})\big|\right] = O(n^{-\alpha}),$$
$$\mathbb{E}\left[\big|\widehat{U}(\mathbf{X})-U(\mathbf{X})\big|\right] = O(n^{-\beta}).$$
\end{assumption}

Consider a binary outcome and the associated Balke-Pearl bound. To obtain estimates $\hat{L}$ and $\hat{U}$ that satisfy Assumption \ref{assump: rate of convergence of L and U}, we first estimate $\{{p}_{y, a \mid z, \mathbf{X}},~y=\pm 1, a=\pm 1, z=\pm 1\}$ (let the estimates be $\{\hat{p}_{y, a \mid z, \mathbf{X}},~y=\pm 1, a=\pm 1, z=\pm 1\}$) and then plug these estimates into \eqref{eqn: bp lower bound} and \eqref{eqn: bp upper bound} to obtain $\hat{U}$ and $\hat{L}$. In Supplementary Material D.2, we prove that if $K$ functions are all $n^{-\theta}$ estimable then their linear combinations, maximum, and minimum are also $n^{-\theta}$ estimable. $L(\mathbf{X})$ and $U(\mathbf{X})$ in the Balke-Pearl bound are both maximum/minimum of a series of linear combinations of $p_{y, a \mid z, \mathbf{X}}$; hence, if we have the following condition hold
\begin{condition}[Convergence of Conditional Probabilities]
\label{cond: converg of prob}
\[
\exists\theta, \mathbb{E}\left[\left|\hat{p}_{y, a \mid z, \mathbf{X}}-p_{y, a \mid z, \mathbf{X}}\right|\right] = O(n^{-\theta}), ~y=\pm 1, a=\pm 1, z=\pm 1,
\]
\end{condition}
then we can deduce that Assumption \ref{assump: rate of convergence of L and U} holds for $\alpha = \beta = \theta$. Condition \ref{cond: converg of prob} holds in many scenarios. For instance, if we fit parametric models to estimate $p_{y, a \mid z, \mathbf{X}},~y=\pm 1, a=\pm 1, z=\pm 1$, and models are correctly specified, then this condition holds for $\theta=\frac{1}{2}$. We can also use flexible and expressive nonparametric regression methods to estimate functions $p_{y, a \mid z, \mathbf{X}}$. Assuming that functions $\{p_{y, a \mid z, \mathbf{X}},~y=\pm 1, a=\pm 1, z=\pm 1\}$ are in a H\"{o}lder ball with smoothness parameter $\alpha$ and a constant radius, then Condition \ref{cond: converg of prob} holds for $\theta={-\frac{\alpha}{d+2\alpha}}$, where $d$ is the dimension of $\mathbf{X}$, when $\{p_{y, a \mid z, \mathbf{X}},~y=\pm 1, a=\pm 1, z=\pm 1\}$ are estimated via wavelets (\citealp{donoho1998minimax}; \citealp{cai2012minimax}) or a variant of the random forests algorithm known as Mondrian forests (\citealp{mourtada2018minimax}). Similar results hold when we use Siddique bounds for a binary outcome and Manski-Pepper bounds for a continuous but bounded outcome; see Supplementary Material D.2 for details.
  
 \begin{assumption}[Norm Condition]\label{assump: norm}\rm
 There exists a constant $M_4$ s.t. for any $f\in \mathcal{F}$:
 \begin{align*}
     ||f||\geq M_4||f||_{\infty}.
 \end{align*}
\end{assumption}

 Assumption \ref{assump: norm} guarantees the sup-norm of minimizer $\widehat{f}_{n}^{\lambda_n}$ can be controlled by penalizing $f$.

Under Assumption \ref{assump: existence of fnite minimizer}-\ref{assump: norm}, it can be shown that  $ ||\widehat{f}_n ^{\lambda_n}||_{\infty} \leq \frac{2\sqrt{M_1 \vee M_3}}{M_4\sqrt{\lambda_n}}$.
 Define $B_{n}$ to be the set of functions $f$ s.t. $f\in \mathcal{F}$ and $||f||_{\infty}\leq  \frac{2\sqrt{M_3 \vee M_1}}{M_4\sqrt{\lambda_n}}$, and $B_n^*$ to be the intersection of $B_n$ and $(C_0/\lambda_n)\cdot\mathcal{F}_0$.
 Lemma \ref{lemma: difference caused by hat w and w} and \ref{lemma: lemma 2} below develop properties of functions in $B_n$.
From now on, we consider $\lambda_n = o(1)$. We use $\mathbb{E}_{\textbf{Z}}$ to denote taking expectation with respect to the random variable $\textbf{Z}$, and $\mathbb{E}$ with respect to all random variables. 
\begin{lemma}\rm
\label{lemma: difference caused by hat w and w}
Let
\[
    {w}_i = |{U}_i|\cdot\mathbbm{1}\big\{{L}_i > 0\big\} + |{L}_i|\cdot\mathbbm{1}\big\{{U}_i<0 \big\} + \big||{U}_i| - |{L}_i |\big|\cdot\mathbbm{1}\big\{[{L}_i, {U}_i] \ni 0\big\},
\]
\[
    {e}_i = \mathbbm{1}\big\{U_i < 0\big\}- \mathbbm{1}\big\{L_i > 0\big\} - \text{sgn}\big\{|{U}_i | - |{L}_i|\big\}\cdot \mathbbm{1}\big\{[{L}_i, {U}_i] \ni 0\big\},
\]
and
\[
    l(x;w,e,f)=w(x)\{1+e(x)f(x)\}^+,
\]
where $L_i = L(\mathbf{X}_i)$, and $U_i = U(\mathbf{X}_i)$, and $\widehat{w}$ and $\widehat{e}$ be defined in (\ref{eqn: w_i def}) and (\ref{eqn: e_i def}). We have
\begin{equation}    
\mathbb{E}_{\textbf{X}_{[I_1]}}\sup_{f\in B_n }\left|\mathbb{E}_{\textbf{X}_{[I_2]}}\left\{\frac{1}{|I_2 |}\sum_{i\in I_2} l(\textbf{X}_i ; w,e,f) -\sum_{i\in I_2} \frac{1}{|I_2 |}l(\textbf{X}_i ; \widehat{w},\widehat{e},f)\right\}\right|\leq O\left(n^{-(\alpha \wedge \beta)}/\sqrt{\lambda_n}\right),
\end{equation}
where $\textbf{X}_{[I_2]}$ and $\textbf{X}_{[I_1]}$ denote $\{\textbf{X}_i, i\in I_2\}$ and $\{\textbf{X}_i, i\in I_1 \}$,  respectively. 
\end{lemma}

Function $l(\cdot;w,e,f)$ in Lemma \ref{lemma: difference caused by hat w and w} denotes the loss function where $w$ and $e$ are set at truth, and $l(\cdot;\widehat{w},\widehat{e},f)$ the loss function where $w$ and $e$ are estimated. Lemma \ref{lemma: difference caused by hat w and w} effectively bounds the risk induced by estimating  $w(\cdot)$ and $e(\cdot)$. Lemma \ref{lemma: lemma 2} below further quantifies the risk induced by estimating the risk function using its empirical analogue. 
We may define the 

\begin{lemma}\rm
\label{lemma: lemma 2}
Let $w(\cdot)$ and $e(\cdot)$ be defined as in Lemma \ref{lemma: difference caused by hat w and w}. We have
\begin{equation}
    \begin{split}
\mathbb{E}_{\textbf{X}_{[I_1]}}\mathbb{E}_{\textbf{X}_{[I_2 ]}}\sup_{f\in B_n^*}\left|\mathbb{E}l(\textbf{X} ; w,e,f) - \sum_{i\in I_2}\frac{1}{|I_2 |} l(\textbf{X}_i ; w,e,f)\right| \leq O\left(\frac{1}{\sqrt{n}\cdot\lambda_n}\right).
    \end{split}
\end{equation}
\end{lemma}
Lemma \ref{lemma: difference caused by hat w and w} and \ref{lemma: lemma 2} facilitate the derivation of the convergence rate of $\mathcal{R}^{\text{h}}_{\text{upper}}(\widehat{f}_{n}^{\lambda_n})$, as is formally stated in Theorem \ref{thm: conv rate}.
\begin{theorem}\rm
\label{thm: conv rate}
Assume that Assumption \ref{assump: existence of fnite minimizer} to \ref{assump: norm} hold. We have
\begin{equation}
 \mathcal{R}^{\text{h}}_{\text{upper}}(\widehat{f}_{n}^{\lambda_n}) - \inf_{f\in \mathcal{F}}\mathcal{R}^{\text{h}}_{\text{upper}}(f)\leq O(\lambda_n+n^{-\frac{1}{2}}\lambda_n^{-1}+\lambda_n^{-\frac{1}{2}}(n^{-\alpha}+n^{-\beta})).
\end{equation}
\end{theorem}
Theorem \ref{thm: conv rate} gives concrete upper bounds of $\mathcal{R}^{\text{h}}_{\text{upper}}(\widehat{f}_{n}^{\lambda_n}) - \inf_{f\in \mathcal{F}}\mathcal{R}^{\text{h}}_{\text{upper}}(f)$, which implies that for a wide range of $\lambda_n$ satisfying $\frac{\ln(\ln(n))}{n^{(\frac{1}{2}\wedge 2\alpha \wedge 2\beta)}}\leq \lambda_n\leq \frac{1}{\ln(\ln(n))} $, as $n$ goes to infinity, $\mathcal{R}^{\text{h}}_{\text{upper}}(\widehat{f}_{n}^{\lambda_n})$ converges to  $\inf_{f\in \mathcal{F}}\mathcal{R}^{\text{h}}_{\text{upper}}(f)$,  the minimal risk over $\mathcal{F}$, and $\mathcal{R}^{\text{h}}_{\text{upper}}(\widehat{f}_{n}^{\lambda_n})- \mathcal{R}^{\text{h},\ast}_{\text{upper}}(f)$ converges to approximation error $\mathcal{A}(\mathcal{F})$.  
Combining Theorem \ref{thm: relation between surrogate and true loss} with Theorem \ref{thm: conv rate}, we have the following proposition that establishes the convergence rate of $\mathcal{R}_{\text{upper}}(\widehat{f}_{n}^{\lambda_n}) - \mathcal{R}^{\ast}_{\text{upper}}-\mathcal{A}(\mathcal{F})$, i.e., the excess risk minus approximation error under the true $0$-$1$-based loss.
\begin{proposition}\rm
\label{prop: rate on true loss}
Under Assumption \ref{assump: existence of fnite minimizer} to \ref{assump: norm}, we have
\begin{equation}
\mathcal{R}_{\text{upper}}(\widehat{f}_{n}^{\lambda_n}) \leq \mathcal{R}^{\ast}_{\text{upper}} + \mathcal{A}(\mathcal{F})+ O(\lambda_n+n^{-\frac{1}{2}}\lambda_n^{-1}+\lambda_n^{-\frac{1}{2}}(n^{-\alpha}+n^{-\beta})).
\end{equation}
\end{proposition}
Proposition \ref{prop: rate on true loss} shows that $\widehat{f}_{n}^{\lambda_n}$ is a consistent estimator in the sense that its risk converges to the Bayes risk plus the approximation error of the chosen function class $\mathcal{F}$, as long as $\lambda_n$ is selected from $\left[\frac{\ln(\ln(n))}{n^{(\frac{1}{2}\wedge 2\alpha \wedge 2\beta)}} , \frac{1}{\ln(\ln(n))}\right]$. In other words, $\widehat{f}_{n}^{\lambda_n}$ converges in its risk to the best IV-optimal ITR in the function class $\mathcal{F}$ for a wide range of $\lambda_n$ choices.

\section{Simulation Studies}
\label{sec: simulation}
In Section \ref{subsec: experiment benchmark}, we define the benchmark of our experiments. In Section \ref{subsec: simulation OWL fails}, we demonstrate that outcome weighted learning (OWL) may have performance in the presence of unmeasured confounding. In Section \ref{subsec: simulation setup} and \ref{subsec: simulation results}, we investigate the performance of our proposed IV-PILE estimator and contrast it to OWL. Section \ref{subsec: additional simu} summarizes results from additional simulations, details of which can be found in Supplementary Material F.2 through F.7.

\subsection{Experiment Benchmark}
\label{subsec: experiment benchmark}
There are at least two relevant benchmarks against which we can evaluate the performance of a candidate ITR in the presence of unmeasured confounding. Let $f_{\text{opt}}$ be an ITR that only has access to the observed covariates $\mathbf{X}$, assigns $A = +1$ when $\text{CATE}(\mathbf{X}) = \mathbb{E}[Y(1) - Y(0) \mid \mathbf{X}] \geq 0$, and assigns $A = -1$ otherwise. Let $V(f_{\text{opt}})$ denote the value of $f_{\text{opt}}$. One may compare a candidate ITR $f$ to the benchmark $f_{\text{opt}}$ by calculating:
\[
\mathcal{R}(f, f_{\text{opt}}) := V(f_{\text{opt}}) - V(f) = \mathbb{E}\left[|\text{CATE}(\mathbf{X})|\cdot \mathbbm{1}\left\{\text{sgn}\{\text{CATE}(\mathbf{X})\} \neq \text{sgn}\{f(\mathbf{X})\}\right\}
\right].
\]
Another relevant benchmark ITR is an ``omniscient" ITR $f_{\text{omni}}$ that has access to \emph{both} observed covariate $\mathbf{X}$ \emph{and} the unmeasured confounder $U$, assigns $A = +1$ when $\text{CATE}(\mathbf{X}, U) = \mathbb{E}[Y(1) - Y(0) \mid \mathbf{X},U] \geq 0$, and assigns $A = -1$ otherwise. One may compare a candidate ITR $f$ to the benchmark $f_{\text{omni}}$ by calculating:
\[
\mathcal{R}(f, f_{\text{omni}}) := V(f_{\text{omni}}) - V(f) =
  \mathbb{E}\left[|\text{CATE}(\mathbf{X}, U)|\cdot \mathbbm{1}\left\{\text{sgn}\{\text{CATE}(\mathbf{X}, U)\} \neq \text{sgn}\{f(\mathbf{X})\}\right\}
\right].
\]
Value functions of two benchmarks $f_{\text{opt}}$ and $f_{\text{omni}}$ differ by a constant that depends only on the data-generating process but not on the candidate ITR and encodes the loss of information due to not observing the unmeasured confounder, i.e.,
\[
V(f_{\text{omni}}) =
V(f_{\text{opt}}) + C_{\text{DGP}}
\] 
for some constant $C_{\text{DGP}} \geq 0$. We then immediately have 
\[
\mathcal{R}(f, f_{\text{opt}}) + C_{\text{DGP}} = \mathcal{R}(f, f_{\text{omni}}).
\]
We note that neither $f_{\text{opt}}$ nor $f_{\text{omni}}$ is identified from observed data under our general partial identification framework, and an IV-optimal ITR neither converges to $f_{\text{opt}}$ nor $f_{\text{omni}}$. Moreover, for the purpose of comparing two or more candidate ITRs, it is equivalent to report $\mathcal{R}(f, f_{\text{opt}})$ or $\mathcal{R}(f, f_{\text{omni}})$
as they differ only by a constant $C_{\text{DGP}}$, i.e., $\mathcal{R}(f_1, f_{\text{omni}}) \leq \mathcal{R}(f_2, f_{\text{omni}})$ implies $\mathcal{R}(f_1) \leq \mathcal{R}(f_2)$ and vice versa for two candidate ITRs $f_1$ and $f_2$ and a fixed data-generating process.

If we further assume identification assumptions that enable us to point identify $\text{CATE}(\mathbf{X})$ as in \citet{cui2019semiparametric} and \citet{qiu2020optimal}, then $f_{\text{opt}}$ becomes identified from observed data, estimated IV-optimal ITR would converge to $f_{\text{opt}}$, and $\mathcal{R}(f, f_{\text{opt}})$ may be of greater importance; however, we would like to point out that even in this special case, $\mathcal{R}(f, f_{\text{omni}})$ is still of interest as $C_{\text{DGP}}$ illuminates the gains from collecting more covariates.

In the simulation studies, we will present $\mathcal{R}(f, f_{\text{omni}})$ for each candidate ITR $f$ and data-generating process, and calculate $C_{\text{DGP}}$ for each data-generating process. Weighted misclassification error rate refers to $\mathcal{R}(f, f_{\text{omni}})$. Readers can then immediately deduce $\mathcal{R}(f, f_{\text{opt}}) = \mathcal{R}(f, f_{\text{omni}}) - C_{\text{DGP}}$. 

\subsection{Failure of the Outcome Weighted Learning (OWL) in the Presence of Unmeasured Confounding}
\label{subsec: simulation OWL fails}

We illustrate how ITR-estimation methods could dramatically fail in the presence of unmeasured confounding. We consider the following simple data-generating process of covariates and treatment:
\begin{equation*}
\begin{split}
    \label{simu model: X and A}
    &X_1, X_2, U \sim \text{Unif}~[-1, 1],\\
    &P(A = 1 \mid X_1, X_2, U) = \text{expit}(1 + X_1 - X_2 + \lambda U),
\end{split}
\end{equation*}
where $(X_1, X_2)$ are observed covariates and $U$ an unmeasured covariate. We consider two outcomes, one continuous $(Y_1)$ and the other binary $(Y_2)$:
\begin{equation*}
    \begin{split}
        &Y_1 = 1 + X_1 + X_2 + \xi U + 0.442(1 - X_1 - X_2 + \delta U)\cdot A + \epsilon, \quad \epsilon \sim N(0, 1),\\
        &P(Y_2 = 1 \mid X_1, X_2, U, A) = \text{expit}\{1 - X_1 + X_2 + \xi U + 0.442(1 - X_1 + X_2 + \delta U)\cdot A\}.
    \end{split}
\end{equation*}
Observed data consists only of $(X_1, X_2, A, Y_1)$ (or $(X_1, X_2, A, Y_2)$) since $U$ is not observed. Parameters $(\lambda, \xi, \delta)$ control the degree of unmeasured confounding, and $(\lambda, \xi, \delta) = (0, 0, 0)$ corresponds to no unmeasured confounding. We adapted the strategy proposed in \citet{zhao2012estimating} and \citet{zhang2012estimating} to the setting of observational studies by fitting a propensity score model $\widehat{\pi}(a)$ based on $X_1$ and $X_2$ alone and estimating the conditional average treatment effect $C(\mathbf{X})$ using an IPW estimator based on $\widehat{\pi}(a)$ in a training dataset consisting of $n_{\text{train}} = 300$ subjects. We then labeled each subject as $+1$ or $-1$ based on $\text{sgn}\{\widehat{C}(\mathbf{X})\}$ and attached to her a weight of magnitude $|\widehat{C}(\mathbf{X})|$. A support vector machine with Gaussian RBF kernel was then applied to this derived dataset. For various $(\lambda, \xi, \delta)$ combinations, we repeated the experiment $500$ times and reported the average weighted misclassification error rate $\mathcal{R}(f, f_{\text{omni}})$ evaluated on a testing dataset of size $100,000$, for both outcomes. 

For the continuous outcome $Y_1$, $\mathcal{R}(f, f_{\text{omni}})$ is less than $0.05$ when $(\lambda, \xi, \delta) = (0, 0, 0)$, suggesting a favorable generalization performance of outcome weighted learning (OWL) when NUCA holds. However, the error rate jumps to almost $0.30$ when $(\lambda, \xi, \delta) = (4, 0, 4)$. To get a sense of how poor the performance is, note that a classifier based on random coin flips yields an average error of $0.49$. Similar qualitative trends hold for the binary outcome $Y_2$. The average weighted misclassification error is $0.02$ when NUCA holds, $0.06$ when $(\lambda, \xi, \delta) = (4, 0, 4)$, and $0.07$ if the trained classifier is replaced with one based on random coin flips. Figure \ref{fig: failure of OWL} summarizes the results. The pattern suggests that a naive procedure assuming no unmeasured confounding could fail dramatically when this assumption is violated. Similar qualitatively results hold when we use $\mathcal{R}(f, f_{\text{opt}})$ in place of $\mathcal{R}(f, f_{\text{omni}})$.

\begin{figure}[ht]
   \centering
     \subfloat{\includegraphics[width = 0.49\columnwidth]{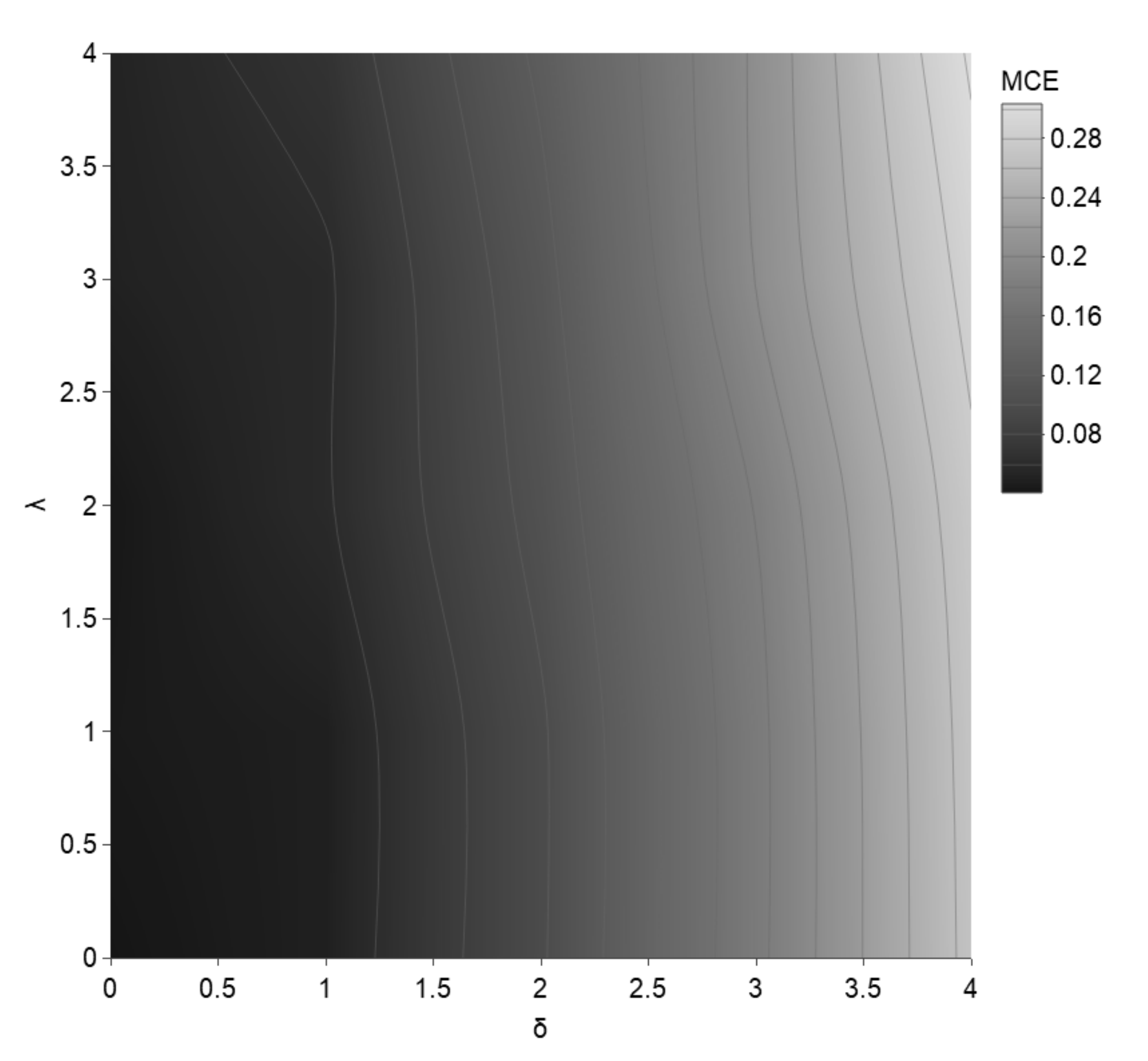}}
     \subfloat{\includegraphics[width = 0.49\columnwidth]{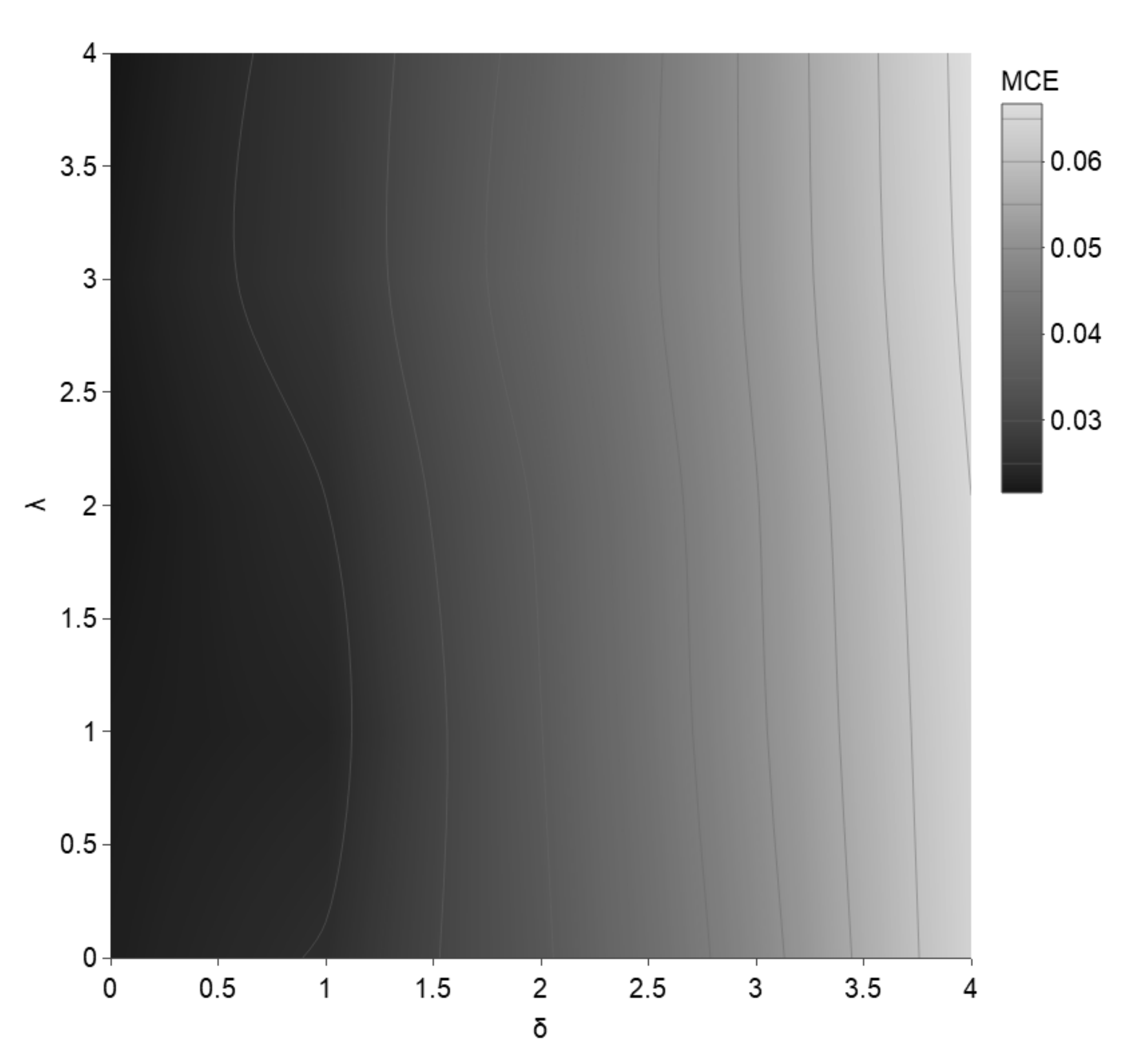}}
     \caption{\small Estimated \emph{weighted misclassification error} on the testing dataset for various $(\lambda, \delta)$ combinations with $\xi = 0$. The left panel plots the result for the continuous outcome $Y_1$. The right panel plots the result for the binary outcome $Y_2$. Size of training data is $300$ and experiment is repeated $500$ times. A classifier based on random coin flips yields an error of $0.49$ in the continuous case and $0.07$ in the binary case.}
\label{fig: failure of OWL}
\end{figure}

\subsection{Comparing IV-PILE to OWL: Experiment Setup}
\label{subsec: simulation setup}
\subsubsection{Data Generating Process}
\label{subsubsec: data generating process}
We considered the following data-generating process with a binary IV $Z$, a binary treatment $A$, a binary outcome $Y$, a $10$-dimensional observed covariates $\mathbf{X}$, and an unmeasured confounder $U$:
\begin{equation*}
    \begin{split}
        &Z \sim \text{Bern}(0.5), ~ X_1,\cdots,X_{10} \sim \text{Unif}~[-1, 1], ~ U \sim \text{Unif}~[-1, 1],\\
        &P(A = 1 \mid \mathbf{X}, U, Z) = \text{expit}\{8Z + X_1 - 7X_2 + \lambda (1 + X_1) U\},\\
        &P(Y = 1 \mid A, \mathbf{X}, U) = \text{expit}\{g_1(\mathbf{X}, U) + g_2(\mathbf{X}, U, A)\},
    \end{split}
\end{equation*}
with the following choices of $g_1(\mathbf{X}, U)$:
\begin{equation*}
    \begin{split}
    &\text{Model}~(1): \qquad g_1(\mathbf{X}, U) = 1 - X_1 + X_2 + \xi U, \\
    &\text{Model}~(2): \qquad g_1(\mathbf{X}, U) = 1 - X_1^2 + X_2^2 + \xi X_1 X_2 U, 
    \end{split}
\end{equation*}
and $g_2(\mathbf{X}, U)$:
\begin{equation*}
    \begin{split}
    &\text{Model}~(1): \qquad g_2(\mathbf{X}, U, A) = 0.442(1 - X_1 + X_2 + \delta U)A, \\
    &\text{Model}~(2): \qquad g_2(\mathbf{X}, U, A) = (X_2 - 0.25X_1^2 - 1 + \delta U)A.
    \end{split}
\end{equation*}
In the above specifications, $\lambda$ controls the level of interaction between $U$ and $\mathbf{X}$ on $P(A = 1)$, and $\delta$ controls the level of interaction between $U$ and $A$ on the outcome. Assumptions underpinning naive methods (OWL, EARL, etc) hold when $(\lambda, \xi, \delta) = (0, 0, 0)$. Moreover, the data-generating process considered here does \emph{not} necessarily satisfy the IV identification assumptions in \citet{cui2019semiparametric} and \citet{qiu2020optimal}. We direct interested readers to Supplementary Material F.1 for more details.

\subsubsection{IV Identification Assumptions and Estimators of $L(\mathbf{x})$ and $U(\mathbf{x})$}
We considered the IV identification set $\mathcal{C}_{\text{BP}}$ discussed in Section \ref{subsec: balke pearl bound}. Note that $Z \sim \text{Bern}(0.5)$ trivially satisfies $\mathcal{C}_{\text{BP}}$. Under $\mathcal{C}_{\text{BP}}$, $L(\mathbf{X})$ and $U(\mathbf{X})$ are calculated as in \eqref{eqn: bp lower bound} and \eqref{eqn: bp upper bound}, and require estimating conditional probabilities $P(Y = y, A = a \mid Z = z, \mathbf{X} = \mathbf{x})$ for $2 \times 2 \times 2 = 8$ ($y = \pm 1$, $a = \pm 1$, $z = \pm 1$) different $(y, a, z)$ combinations. These conditional probabilities do not involve $U$ and are identified from the observed data. In general, these conditional probabilities may not admit simple and familiar parametric form, and researchers are advised to use some flexible estimation routines, e.g., random forest (\citealp{breiman2001random}), to fit the data. We fit all conditional probabilities using random forest models with default settings as implemented in the \textsf{R} package \textsf{randomForest} (\citealp{RF_package}). We also considered estimating these conditional probabilities using simple (but misspecified) multinomial logistic regression models. Likewise, when implementing a naive OWL method, we estimated the propensity score model using a random forest and a logistic regression model.

\subsubsection{Training and Testing Dataset}
\label{subsubsec: train and test dataset}
Our training dataset consisted of $n_{\text{train}} = 300$ or $500$ independent samples of $(Z, \mathbf{X}, A, Y)$. Although the true data-generating process involved the unmeasured confounder $U$, we did \emph{not} observe or make use of $U$ throughout the training process. The testing dataset consisted of $n_{\text{test}} = 100,000$ independently drawn copies of $(\mathbf{X}, U)$. Their true conditional average treatment effects were calculated (since we had \emph{both} $\mathbf{X}$ and $U$ information). We reported estimated  $\mathcal{R}(f, f_{\text{omni}})$ of various candidate ITRs on this testing dataset. We also calculated and reported $C_{\text{DGP}}$ for each data-generating process so that $\mathcal{R}(f, f_{\text{opt}}) = \mathcal{R}(f, f_{\text{omni}}) - C_{\text{DGP}}$ can be immediately inferred.

\subsection{Numerical Results}
\label{subsec: simulation results}
We report simulation results in this section. We considered three classifiers:
\begin{enumerate}
    \item IV-PILE-RF: IV-PILE with relevant conditional probabilities in $L(\mathbf{X})$ and $U(\mathbf{X})$ estimated via random forests and classification performed with a Gaussian kernel;
    \item OWL-RF: OWL with the propensity score estimated via random forests and classification performed with a Gaussian kernel;
    \item COIN-FLIP: Classifier based on random coin flips.
\end{enumerate}

Table \ref{tbl: simulation 1} reports the estimated weighted misclassification error rate $\widehat{\mathcal{R}}(f, f_{\text{omni}})$ of IV-PILE-RF ($\widehat{\mathcal{R}}(\text{IV-PILE-RF}, f_{\text{omni}})$), OWL-RF ($\widehat{\mathcal{R}}(\text{OWL-RF}, f_{\text{omni}})$), and COIN-FLIP ($\widehat{\mathcal{R}}(\text{COIN-FLIP}, f_{\text{omni}})$) for different $(\lambda, \xi, \delta)$, $g_1(\mathbf{X}, U)$, and $n_{\text{train}}$ combinations when $g_2(\mathbf{X}, U, A)$ is taken to be Model (1). Supplementary Material F.1 reports the same numerical results when $g_2(\mathbf{X}, U, A)$ is taken to be Model (2). 

Table \ref{tbl: simulation 1} suggests two consistent trends that align well with our theory and intuition. First, we would expect that $\widehat{\mathcal{R}}(\text{IV-PILE-RF}, f_{\text{omni}})$ and $\widehat{\mathcal{R}}(\text{IV-PILE-RF}, f_{\text{opt}})$ would not go to $0$ even when $n_{\text{train}} \rightarrow \infty$. This is verified by noting that increasing $n_{\text{train}}$ does \emph{not} drive either error on the testing dataset to $0$ (See Supplementary Material F.3 for results when $n_{\text{train}}$ is larger than $500$). Moreover, we observe that $\widehat{\mathcal{R}}(\text{OWL-RF}, f_{\text{omni}})$ and $\widehat{\mathcal{R}}(\text{OWL-RF}, f_{\text{opt}})$ also remain large as $n_{\text{train}}$ grows, which reflects that the problem of unmeasured confounding is fundamental, and does \emph{not} go away as the training sample size grows. Second, the IV-PILE estimator seems to be robust and \emph{outperforms} the naive OWL estimator in all simulation settings considered here. However, we would like to point out that this is \emph{not} suggesting that our approach \emph{always} outperforms OWL or similar methods. Unlike our proposed approach, when the assumptions underpinning these methods are not met, no guarantees can provided about their performance, and it is difficult to predict what would happen to these methods in practice. 
\begin{table}[ht]
\centering
\caption{\small Estimated \emph{weighted misclassification error rate} for different $(\lambda, \delta, \xi)$ and $g_1(\mathbf{X}, U)$ combinations. We take $g_2(\mathbf{X}, U, A)$ to be Model (1) throughout. Training data sample size $n_{\text{train}} = 300$ or $500$. Each number in the cell is averaged over $500$ simulations. Standard errors are in parentheses.}
\label{tbl: simulation 1}
\resizebox{\columnwidth}{!}{%
\fbox{%
\begin{tabular}{ccccccc} \\[-0.8em]
   &\multicolumn{2}{c}{IV-PILE-RF} &\multicolumn{2}{c}{OWL-RF} &COIN-FLIP & $C_{\text{DGP}}$\\[-0.9em] \\ \\[-1em] \\[-0.8em]
 $\xi = 0$, $g_1 = \text{Model}~(1)$& $n_{\text{train}} = 300$ & $n_{\text{train}} = 500$ & $n_{\text{train}} = 300$ & $n_{\text{train}} = 500$ \\[-0.9em] \\ \\[-1em] \cline{1-1} \\[-0.8em]
$(\lambda, \delta) = (0.5, 0.5)$  &0.005 (0.000) &0.005 (0.000) & 0.014 (0.004) & 0.015 (0.003) &0.031 (0.000) & 0.001\\
$(\lambda, \delta) = (1.0, 1.0)$  &0.008 (0.000) &0.008 (0.000) & 0.018 (0.004) &0.018 (0.003) &0.033 (0.000) & 0.005 \\ 
$(\lambda, \delta) = (1.5, 1.5)$  &0.014 (0.001) &0.013 (0.000) & 0.022 (0.004) & 0.023 (0.003) &0.037 (0.000) & 0.010\\ 
$(\lambda, \delta) = (2.0, 2.0)$  &0.020 (0.000) &0.020 (0.000) & 0.029 (0.003) &0.029 (0.003) &0.043 (0.000) & 0.016 \\ \\[-0.8em]\hline 
\\[-0.8em]
$\xi = 1$, $g_1 = \text{Model}~(1)$ \\[-0.9em] \\ \cline{1-1} \\[-0.9em]
$(\lambda, \delta) = (0.5, 0.5)$  &0.005 (0.000) &0.004 (0.000) &0.013 (0.004) &0.014 (0.003) &0.029 (0.000) & 0.001\\
$(\lambda, \delta) = (1.0, 1.0)$  &0.007 (0.000) &0.007 (0.000) &0.016 (0.003) &0.016 (0.003) &0.029 (0.000) & 0.005\\
$(\lambda, \delta) = (1.5, 1.5)$  &0.012 (0.000) &0.012 (0.000) &0.019 (0.003) &0.020 (0.002) &0.031 (0.000) & 0.009\\
$(\lambda, \delta) = (2.0, 2.0)$  &0.018 (0.000) &0.018 (0.000) &0.025 (0.003) &0.025 (0.002) &0.035 (0.000) & 0.016\\ \\[-0.8em]\hline 
\\[-0.8em]
 $\xi = 0$, $g_1 = \text{Model}~(2)$ \\[-0.9em] \\ \\[-1em] \cline{1-1} \\[-0.8em]
$(\lambda, \delta) = (0.5, 0.5)$  &0.005 (0.000) &0.005 (0.000) & 0.017 (0.006) & 0.018 (0.005) &0.040 (0.000) &0.001\\
$(\lambda, \delta) = (1.0, 1.0)$  &0.007 (0.000) &0.007 (0.000) & 0.020 (0.006) &0.021 (0.005) &0.042 (0.000) &0.004 \\ 
$(\lambda, \delta) = (1.5, 1.5)$  &0.012 (0.000) &0.012 (0.000) & 0.024 (0.006) & 0.025 (0.005) &0.045 (0.000) & 0.008\\ 
$(\lambda, \delta) = (2.0, 2.0)$  &0.019 (0.000) & 0.019(0.000) & 0.031 (0.006) &0.031 (0.005) &0.050 (0.000) &0.014 \\ \\[-0.8em]\hline 
\\[-0.8em]
$\xi = 1$, $g_1 = \text{Model}~(2)$ \\[-0.9em] \\ \cline{1-1} \\[-0.9em]
$(\lambda, \delta) = (0.5, 0.5)$  &0.004 (0.000) &0.004 (0.000) &0.017 (0.004) &0.017 (0.005) &0.040 (0.000) & 0.001\\
$(\lambda, \delta) = (1.0, 1.0)$  &0.007 (0.000) &0.007 (0.000) &0.019 (0.006) &0.020 (0.005) &0.042 (0.000) & 0.004\\
$(\lambda, \delta) = (1.5, 1.5)$  &0.012 (0.000) &0.012 (0.000) &0.024 (0.006) &0.025 (0.005) &0.045 (0.000) & 0.009\\
$(\lambda, \delta) = (2.0, 2.0)$  &0.019 (0.000) &0.019 (0.000) &0.030 (0.005) &0.031 (0.005) &0.049 (0.000) & 0.015\\
\end{tabular}}}
\end{table}

\subsection{Additional Simulations}
\label{subsec: additional simu}
We conducted abundant additional simulations and reported results in Supplementary Material F.2 through F.7. In particular, in Supplementary Material F.2, we compared the performance of IV-PILE and OWL when relevant conditional probabilities in $L(\mathbf{X})$ and $U(\mathbf{X})$ were fit via a misspecified multinomial logistic regression model, and with random forest models with node sizes equal to $5$ or $10$. Qualitative behaviors described in Section \ref{subsec: simulation results} still held, and the IV-PILE algorithm seemed to be robust against misspecification of models used to fit relevant probabilities in $L(\mathbf{X})$ and $U(\mathbf{X})$. We also observed that using a larger node size in a random forest model seemed to largely improve the performance in some scenarios. In Supplementary Material F.3, we reported simulation results with larger training sample size $n_{\text{train}}$. In Supplementary Material F.4, we repeated a subset of simulation studies with a sample-splitting version of the IV-PILE algorithm, whose theoretical properties were studied in Section \ref{sec: theory}. We found that the sample-splitting version of the IV-PILE had slightly inferior finite-sample performance compared to the non-splitting version when $n_{\text{train}}$ is small; however, the sample-splitting version still largely outperformed the non-sample-splitting OWL estimator in simulation settings considered here. 

In Supplementary Material F.5, we varied the association between the IV $Z$ and the treatment $A$. We found that the performance of IV-PILE became worse when the association between the IV and the treatment became smaller. This observation aligns very well with Proposition \ref{prop: approx error} and our intuition because a weaker IV corresponds to less informative and wider partial identification intervals. In Supplementary Material F.6, we allowed the IV to violate the exclusion restriction assumption and have a direct effect on the outcome. We found that the IV-PILE estimator seemed to be robust to slight violation of the exclusion restriction assumption. Finally, in Supplementary Material F.7, we considered settings where the outcome was continuous but bounded, and estimated the partial identification interval using the Manski-Pepper bounds. We found that the qualitative conclusions that held for a binary outcome still held for the continuous but bounded outcome.

\section{Differential Impact of Delivery Hospital on Preemies' Outcomes Revisited}
\label{sec: application}
We now revisit the NICU data and apply our developed method to it. We considered data describing all premature babies in the Commonwealth of Pennsylvania from the year $1995$ to $2005$, with the following observed covariates describing mothers and their preemies: birth weight, gestational age in weeks, age of mother, insurance type of mother (fee for service, HMO, federal/state, other, uninsured), mother's race (white, African American, Hispanic, other), prenatal care, mother's education level, mother's parity, and the following covariates describing the zip code the mother lives in: median income, median home value, percentage of people who rent, percentage below poverty, percentage with high school degree, percentage with college degree. The treatment is $1$ if the baby is delivered at a hospital with high-level NICUs and $0$ otherwise. The treatment is believed to be still confounded because there are other important covariates not accounted for, e.g., the severity of mother's comorbidities. Mothers with severe comorbidities are more likely to be sent to high-level NICUs and their babies are at higher risk.

To resolve this concern, we followed \citet{lorch2012differential} and used the differential travel time, defined as mother's travel time to the nearest high-level NICU minus time to the nearest low-level NICU, as an instrumental variable. We constructed a binary IV $Z$ out of the excess travel time as follows: $Z = 1$ if excess travel time in its highest $10$ percentile, and $Z = -1$ if in its lowest $10$ percentile. The outcome of interest is premature infant mortality, including both fetal and non-fetal death. Our goal is to estimate an individualized treatment rule that recommends whether to send mothers to hospitals with high-level NICUs based on her observed covariates. In the case of having a premature baby, mothers should be directed to a hospital with high-level NICUs instead of the \emph{closest} hospital only when a high-level NICU is \emph{significantly} better for preemie mortality. As has been discussed in Section \ref{subsec: recast prob in semi-sup learning}, we let $\Delta$ control for what margin is significant. In practice, $\Delta$ should be informed based on some practical knowledge of the situation, e.g., how far the hospital with high-level NICUs is, and how urgent the situation is. 

Before proceeding to estimation, we saved $5,000$ data points as testing data and used the rest $25,702$ as training data. We considered two IV assumption sets: $\mathcal{C}_{\text{BP}}$ that underpins the Balke-Pearl bound and $\mathcal{C}_{\text{Sid}}$ that underpins the Siddique bound (see Section \ref{sec: examples of C and application to NICU}). Table \ref{tbl: size of D_l and D_ul real data} summarizes the sizes of the labeled and unlabeled parts of the training data, i.e., $|\mathcal{D}_l|$ and $|\mathcal{D}_{ul}|$, under assumption sets $\mathcal{C}_{\text{BP}}$ and $\mathcal{C}_{\text{Sid}}$ and for assorted $\Delta$ values. We observed that the additional ``Correct Non-Compliant Decision'' assumption in $\mathcal{C}_{\text{Sid}}$ helped significantly reduce the length of the partial identification intervals for our data, and therefore largely reduced $|\mathcal{D}_{ul}|$. Nevertheless, $|\mathcal{D}_{ul}| > 0$ in all cases for both $\mathcal{C}_{\text{BP}}$ and $\mathcal{C}_{\text{Sid}}$.

\begin{table}[h]
\centering
\caption{\small Size of $\mathcal{D}_l$ and $\mathcal{D}_{ul}$, the labeled and unlabeled part of the training data, under assumption sets $\mathcal{C}_{\text{BP}}$ and $\mathcal{C}_{\text{Sid}}$.}
\label{tbl: size of D_l and D_ul real data}
\begin{tabular}{lcc|cc|cc}\hline \\[-0.8em]
   &\multicolumn{2}{c}{$\Delta = 0$} &\multicolumn{2}{c}{$\Delta = 0.02$} &\multicolumn{2}{c}{$\Delta = 0.05$}\\[-0.9em] \\ \\[-1em] \\[-0.8em]
 Assumption Set & $|\mathcal{D}_l|$ & $|\mathcal{D}_{ul}|$ & $|\mathcal{D}_l|$ & $|\mathcal{D}_{ul}|$ & $|\mathcal{D}_l|$ & $|\mathcal{D}_{ul}|$ \\[-0.9em] \\ \\[-1em] \\[-0.8em]
$\mathcal{C}_{\text{BP}}$: Balke-Pearl Bound  &2,132 &23,617 & 4,926 & 20,776 &7,463 & 18,239 \\
$\mathcal{C}_{\text{Sid}}$: Siddique Bound  & 23,253 &2,449 &25,549 &153 & 25,535 & 167\\ \\[-0.8em]\hline 
\end{tabular}
\end{table}

We applied the developed IV-PILE approach to the training data, and used a $5$-fold cross validation to select the tuning parameters $\lambda_n$ (controlling the model complexity) and $\sigma$ (Gaussian kernel parameter) on a logarithmic grid from $10^{-3}$ to $10^3$. Let $\widehat{f}_{\text{BP}}$ denote the estimated ITR with optimal tuning parameters under the assumption set $\mathcal{C}_{\text{BP}}$ and $\widehat{f}_{\text{Sid}}$ under $\mathcal{C}_{\text{Sid}}$. We then applied $\widehat{f}_{\text{BP}}$ and $\widehat{f}_{\text{Sid}}$ on the $5,000$ testing data and estimated $\mathcal{R}_{\text{upper}}$. When $\Delta = 0$, we had $\widehat{\mathcal{R}}_{\text{upper}}(\widehat{f}_{\text{BP}}) = 0.149$ and $\widehat{\mathcal{R}}_{\text{upper}}(\widehat{f}_{\text{Sid}}) = 0.020$. This suggests that $\widehat{f}_{\text{BP}}$ would incur a weighted misclassification error \emph{at most} as large as $0.149$ under the four core IV assumptions. If we further assumed the ``Correct Non-Compliant Decision'' assumption as in $\mathcal{C}_{\text{Sid}}$, the estimated ITR $\widehat{f}_{\text{Sid}}$ would incur a weighted misclassification error \emph{at most} as large as $0.020$. It is not surprising that the expected worst-case loss is much smaller under $\mathcal{C}_{\text{Sid}}$, given that the additional assumption in $\mathcal{C}_{\text{Sid}}$ largely reduced the length of $[L(\mathbf{X}), U(\mathbf{X})]$. Importantly, although the optimal rule is not identifiable under $\mathcal{C}_{\text{upper}}$ or $\mathcal{C}_{\text{Sid}}$, we can estimate the worst-case risk associated with the IV-PILE estimator, and this gives practitioners important guidance: the learned treatment rule is potentially useful and beneficial when the identification assumption set is agreed upon by experts, and the worst-case risk is deemed reasonable.

\section{Discussion}
\label{sec: discussion and future works}
We study in detail the problem of estimating individualized treatment rules with a valid instrumental variable in this article. We have two major contributions. First, we point out the connection and a fundamental distinction between ITR estimation problems with and without an IV: both problems can be viewed as a classification task; however, the partial identification nature of an IV analysis creates a third latent class, those for whom we cannot assert if the treatment is beneficial or harmful, in addition to those who would ``benefit'' or ``not benefit'' from the treatment. This perspective provides a \emph{unifying} framework that facilitates thinking and framing the problem under distinct and problem-specific IV identification assumptions. 

Second, we approach this unique classification problem by defining a new notion of ``IV-optimality'': an IV-optimal rule minimizes the worst-case weighted misclassification error with respect to the putative IV and under the set of IV identification assumptions. IV-optimality is a sensible criterion that is always well-defined, and an IV-optimal rule can be estimated even under minimal IV identification assumptions and mild modeling assumptions. Our proposed IV-PILE estimator estimates such an IV-optimal rule, and may be advantageous compared to naively applying OWL or similar methods to observational data when NUCA fails, or when the putative IV does not allow point identifying the conditional average treatment effect. Although the focus of the article is estimating ITRs using observational data, the method developed here also applies to randomized control trials with individual noncompliance, as is commonly seen in clinical decision support systems. 

Works most related to our proposed approach are \citet{kallus2018confounding} and \citet{kallus2018interval}, both of which consider the problem of improving a baseline policy when a $\Gamma$-sensitivity analysis model is used to control the degree of unmeasured confounding. Both \citet{kallus2018confounding} and \citet{kallus2018interval} consider minimizing the maximum risk \emph{relative to} the baseline policy, and the CATE is also partially identified under the prescribed sensitivity analysis model. There are two main differences between their approach and ours. First, their approach necessarily requires a baseline policy/ITR and their derived policy/ITR is only guaranteed to do no worse than this baseline under their prescribed sensitivity analysis model. On the other hand, our approach does not require a baseline, and our method can be thought of as delivering a reasonably ``good" baseline policy/ITR. Second, their ``improved policy" always \emph{mimics} the baseline when the CATE under the sensitivity analysis model covers $0$. On the other hand, our IV-PILE estimator has a very different target, i.e., ``IV-optimality", and would recommend a treatment based on the partial identification region alone. One promising research direction is to study how to improve a baseline policy/ITR using one or several valid instrumental variables, instead of relying on a sensitivity analysis model.

Finally, we outline several broad future directions. First and foremost, it is of great interest to develop alternative optimality criteria when an IV only partially identifies the conditional average treatment effect, and compare these alternative criteria with IV-optimality. Second, it is of great importance to restrict the function class under consideration to some parsimonious and scientifically meaningful classes, e.g., the class of decision trees as considered in \citet{laber2015tree}, and the decision lists as considered in \citet{zhang2015using} and \citet{zhang2018interpretable}, and develop more \emph{interpretable} treatment rules under the IV setting. The ``IV-optimality'' developed in this article is still a relevant criterion for such an interpretable decision rule. Third, the ``min-max'' approach as developed in this article can be made less conservative in some settings. One possibility is to consider additional structural assumptions on $C(\mathbf{x})$. For instance, it is conceivable that $C(\mathbf{x})$, subject to $L(\mathbf{x}) \preceq C(\mathbf{x}) \preceq U(\mathbf{x})$, is smooth in $\mathbf{x}$. Lastly, instead of the single-decision setting considered in this article, it is interesting to consider multi-stage problems, i.e, \emph{dynamic treatment rules} (\citealp{murphy2001marginal}; \citealp{murphy2003optimal}; \citealp{robins2004optimal}; \citealp{moodie2007demystifying}; \citealp{zhao2011reinforcement}), and investigate learning optimal dynamic treatment rules with a potentially time-varying instrumental variable. 

\section*{Acknowledgement}
The authors would like to thank Dylan S. Small and reading group participants at the University of Pennsylvania for helpful thoughts and feedback. The authors would like to acknowledge Scott A. Lorch for access to the NICU data.

\begin{center}
{\large\bf SUPPLEMENTARY MATERIALS}
\end{center}

\begin{description}

\item[Online Supplementary Materials] 
Supplementary Material A reviews partial identification bounds results for continuous but bounded outcome. Supplementary Material B derives the linear and nonlinear decision rules and solves the associated optimization problems. Supplementary Material C contains proofs of Proposition 1-3 and Theorem 1. Supplementary Material D contains proof of Proposition 4 and how to estimate partial identification intervals. Supplementary Material E contains proofs of Lemmas and Theorem 2. Supplementary Material F contains additional simulation results. Supplementary Materials G constructs simple plug-in estimators of IV-optimal rules and proves that they are minimax optimal. \textsf{R} package \textsf{ivitr} implements the proposed method and can be downloaded via the Comprehensive \textsf{R} Archive Network (CRAN).
\end{description}

\bibliographystyle{rss}
\bibliography{IV_ITR}

\clearpage
\pagenumbering{arabic}

\setcounter{figure}{0}  
\setcounter{table}{0}  

  \begin{center}
    {\LARGE\bf Online Supplementary Materials for ``Estimating Optimal Treatment Rules with an Instrumental Variable: A Partial Identification Learning Approach" by Hongming Pu and Bo Zhang}
\end{center}
  \medskip
  
 \begin{abstract}
     Supplementary Material A reviews partial identification bounds results for continuous but bounded outcome (\citealp{manski1998monotone}). Supplementary Material B plots the surrogate loss function, derives the linear/nonlinear decision rules, and solves the associated optimization problems. Supplementary Material C contains proofs of Proposition 1-3 and Theorem 1. Supplementary Material D contains proof of Proposition 4 and how to estimate $L(\bm X)$ and $U(\bm X)$ so that Assumption 5 is satisfied. Supplementary Material E contains proofs of Lemmas and Theorem 2. Supplementary Material F contains a discussion of assumptions underpinning the identification results of \citet{cui2019semiparametric}, and additional simulation results, including simulation results when relevant conditional probabilities are estimated via simple parametric models or random forest with different node sizes, when a larger training sample size $n_{\text{train}}$ is used, when a sample-splitting version of the IV-PILE algorithm is used, when the IV strength varies, when the exclusion restriction assumption is mildly violated, and when the outcome is continuous but bounded and the Manski-Pepper bound is used. Finally, Supplementary Material G constructs simple plug-in estimators for IV-optimal rules and proves that they are minimax optimal.
 \end{abstract}

\begin{center}
    {\large\bf Supplementary Material A: Partial Identification Bounds for Continuous but Bounded Outcomes and Multilevel Instrumental Variables}
\end{center}

We first consider the simple case with a binary IV $Z$, a binary treatment $A$, a continuous outcome $Y$, and no measured confounder $\mathbf{X}$. \citet{manski1998monotone} considered the following \emph{monotone instrumental variable} (MIV) assumption and \emph{boundedness outcome assumption}:\\

\textbf{MIV Assumption}: $Z$ is a binary monotone instrumental variable in the sense of mean-monotonicity if, for $a \in \{0, 1\}$, 
\begin{equation}
    \label{eqn: MIV assumption}
    \mathbb{E}[Y(a) \mid Z = 1] \geq \mathbb{E}[Y(a) \mid Z = 0].
\end{equation}

\textbf{Boundedness Outcome Assumption}: $Y$ is a bounded outcome such that $Y \in [K_0, K_1]$.

Under the MIV assumption and the boundedness outcome assumption, \citet{manski1998monotone} showed that the marginal mean couterfactual outcome $\mathbb{E}[Y(a)]$, $a = 0, 1$, satisfies:
\begin{equation}
\label{eqn: manski-pepper bound}
\begin{split}
    &\sum_{z \in \{0, 1\}}P(Z = z)\left\{ \sup_{z_1 \leq z}\left[\mathbb{E}\{Y \mid Z = z_1, A = a\} \cdot P(A = a \mid Z = z_1) + K_0 \cdot P(A = 1 - a \mid Z = z_1)\right] \right\} \\
    &\hspace{7cm}\leq \mathbb{E}[Y(a)] \leq\\
    &\sum_{z \in \{0, 1\}}P(Z = z)\left\{ \inf_{z_2 \geq z}\left[\mathbb{E}\{Y \mid Z = z_2, A = a\} \cdot P(A = a \mid Z = z_2) + K_1 \cdot P(A = 1 - a \mid Z = z_2)\right]\right\} \\
\end{split}
\end{equation}
Once the upper and lower bound on $\mathbb{E}[Y(1)]$ and $\mathbb{E}[Y(0)]$ are obtained, the lower (upper) bound on the ATE = $\mathbb{E}[Y(1)] - \mathbb{E}[Y(0)]$ follows immediately by subtract the upper (lower) bound on $\mathbb{E}[Y(0)]$ from the lower (upper) bound on $\mathbb{E}[Y(1)]$. 
 
Now suppose that we further have observed covariates $\mathbf{X}$. We can assume that MIV assumption holds within each strata formed by the observed covariates $\mathbf{X}$:
\begin{equation}
    \label{eqn: MIV assumption with x}
    \mathbb{E}[Y(a) \mid \mathbf{X} = \mathbf{x}, Z = 1] \geq \mathbb{E}[Y(a) \mid \mathbf{X} = \mathbf{x}, Z = 0],
\end{equation}
and the Manski-Pepper bound \eqref{eqn: manski-pepper bound} can be modified by replacing $P(Z = z)$ with $P(Z = z \mid \mathbf{X})$, $\mathbb{E}\{Y \mid Z = z, A = a\}$ with $\mathbb{E}\{Y \mid Z = z, A = a, \mathbf{X}\}$, and $P(A = 1 - a \mid Z = z)$ with $P(A = 1 - a \mid Z = z_1, \mathbf{X})$. In practice, these conditional probabilities and expectations can be fit using flexible machine learning tools or simple parametric models as discussed in Section \ref{sec: examples of C and application to NICU}; in this way, partial identification intervals $[L(\mathbf{X}), U(\mathbf{X})]$ for a continuous but bounded outcome $Y$ can be obtained. MIV assumption and the Manski-Pepper bound \eqref{eqn: manski-pepper bound} can be further generalized to handle a multi-leveled IV. Let $\mathcal{Z}$ denote an ordered set of values $Z$ can take. IV $Z$ is said to be a monotone multi-leveled instrumental variable if for all $\mathbf{X} = \mathbf{x}$ and all $(z_1, z_2) \in \mathcal{Z} \times \mathcal{Z}$ such that $z_2 \geq z_1$, the following holds:
\begin{equation}
    \label{eqn: MIV assumption with multilevel Z}
    \mathbb{E}[Y(a) \mid \mathbf{X} = \mathbf{x}, Z = z_2] \geq \mathbb{E}[Y(a) \mid \mathbf{X} = \mathbf{x}, Z = z_1].
\end{equation}
The Manksi-Pepper bound \eqref{eqn: manski-pepper bound} can be adapted to a monotone multi-leveled IV by replacing $\sum_{z \in \{0, 1\}}$ with $\sum_{z \in \mathcal{Z}}$. For other partial identification bounds results for continuous but bounded outcome under slightly different IV identification assumptions, see \citet{kitagawa2009identification} and \citet{huber2017sharp}. 

Finally, we explicitly write down the Manski-Pepper bounds for a binary IV $Z$, a binary treatment $A$, and a continuous but bounded outcome $Y$ under the MIV assumption and the boundedness outcome assumptions for reference. Define
\begin{equation}
\begin{split}
    \psi_{Z = z, A = a}(Y, \mathbf{X}; K) = &\mathbb{E}\{Y \mid Z = z, A = a \mid \mathbf{X}\}\cdot P(A = a \mid Z = z, \mathbf{X})\\ 
    + &K \cdot P(A = 1 - a \mid Z = z, \mathbf{X}).
\end{split}
\end{equation}
We then have
\begin{equation}
    \begin{split}
        &P(Z = 0 \mid \mathbf{X}) \cdot \psi_{0, 0}(Y, \mathbf{X}; K_0) + P(Z = 1 \mid \mathbf{X})\cdot \max\{\psi_{0, 0}(Y, \mathbf{X}; K_0), \psi_{1, 0}(Y, \mathbf{X}; K_0)\} \\
        &\hspace{6 cm}\leq \mathbb{E}\{Y(0) \mid \mathbf{X}\} \leq \\
        &P(Z = 0 \mid \mathbf{X}) \cdot \min\{\psi_{0, 0}(Y, \mathbf{X}; K_1), \psi_{1, 0}(Y, \mathbf{X}; K_1)\} + P(Z = 1 \mid \mathbf{X}) \cdot \psi_{1, 0}(Y, \mathbf{X}; K_1),
    \end{split}
\end{equation}
and 
\begin{equation}
    \begin{split}
        &P(Z = 0 \mid \mathbf{X}) \cdot \psi_{0, 1}(Y, \mathbf{X}; K_0) + P(Z = 1 \mid \mathbf{X})\cdot \max\{\psi_{0, 1}(Y, \mathbf{X}; K_0), \psi_{1, 1}(Y, \mathbf{X}; K_0)\} \\
        &\hspace{6 cm}\leq \mathbb{E}\{Y(1) \mid \mathbf{X}\} \leq \\
        &P(Z = 0 \mid \mathbf{X}) \cdot \min\{\psi_{0, 1}(Y, \mathbf{X}; K_1), \psi_{1, 1}(Y, \mathbf{X}; K_1)\} + P(Z = 1 \mid \mathbf{X}) \cdot \psi_{1, 1}(Y, \mathbf{X}; K_1).
    \end{split}
\end{equation}
Finally, $\text{CATE}(\mathbf{X}) = \mathbb{E}\{Y(1) \mid \mathbf{X}\} - \mathbb{E}\{Y(0) \mid \mathbf{X}\}$ has a lower bound:
\begin{equation}
\label{eqn: Mansk-Pepper lower bound}
\begin{split}
    L(\mathbf{X}) = &P(Z = 0 \mid \mathbf{X}) \cdot \psi_{0, 1}(Y, \mathbf{X}; K_0) + P(Z = 1 \mid \mathbf{X})\cdot \max\{\psi_{0, 1}(Y, \mathbf{X}; K_0), \psi_{1, 1}(Y, \mathbf{X}; K_0)\} \\
    - &P(Z = 0 \mid \mathbf{X}) \cdot \min\{\psi_{0, 0}(Y, \mathbf{X}; K_1), \psi_{1, 0}(Y, \mathbf{X}; K_1)\} - P(Z = 1 \mid \mathbf{X}) \cdot \psi_{1, 0}(Y, \mathbf{X}; K_1),
\end{split}
\end{equation}
and an upper bound
\begin{equation}
\label{eqn: Mansk-Pepper upper bound}
\begin{split}
    U(\mathbf{X}) = &P(Z = 0 \mid \mathbf{X}) \cdot \min\{\psi_{0, 1}(Y, \mathbf{X}; K_1), \psi_{1, 1}(Y, \mathbf{X}; K_1)\} + P(Z = 1 \mid \mathbf{X}) \cdot \psi_{1, 1}(Y, \mathbf{X}; K_1) \\
    - &P(Z = 0 \mid \mathbf{X}) \cdot \psi_{0, 0}(Y, \mathbf{X}; K_0) - P(Z = 1 \mid \mathbf{X})\cdot \max\{\psi_{0, 0}(Y, \mathbf{X}; K_0), \psi_{1, 0}(Y, \mathbf{X}; K_0)\} ,
\end{split}
\end{equation}

In practice, we need to specify a range $[K_0, K_1]$ for the outcome of interest $Y$, and estimate $P(Z = z \mid \mathbf{X})$ and each part in $\psi_{Z = z, A = a}(Y, \mathbf{X}; K)$ using flexible machine learning tools or parsimonious parametric models. We evaluate the performance of the IV-PILE estimator when the outcome is continuous but bounded in Supplementary Material F.

\clearpage
\begin{center}
    {\large\bf Supplementary Material B: Deriving Linear/Nonlinear Decision Rules and Solving Associated Optimization Problems}
\end{center}

\subsection*{B.1: Plot of the Surrogate Loss}

\begin{figure}[h]
    \centering
    \caption{\small An illustration of the original 0-1-based loss and the corresponding surrogate loss for four types of $[L(\mathbf{x}), U(\mathbf{x})]$. Blue dashed lines represent the original 0-1-based loss. Red solid lines represent the surrogate loss.}
\label{fig: surrogate loss and 0-1 loss}
     \subfloat[${[}L(\mathbf{x}), U(\mathbf{x}){]} = {[}1, 3{]}$ ]{\includegraphics[width = 0.45\columnwidth, height = 6cm]{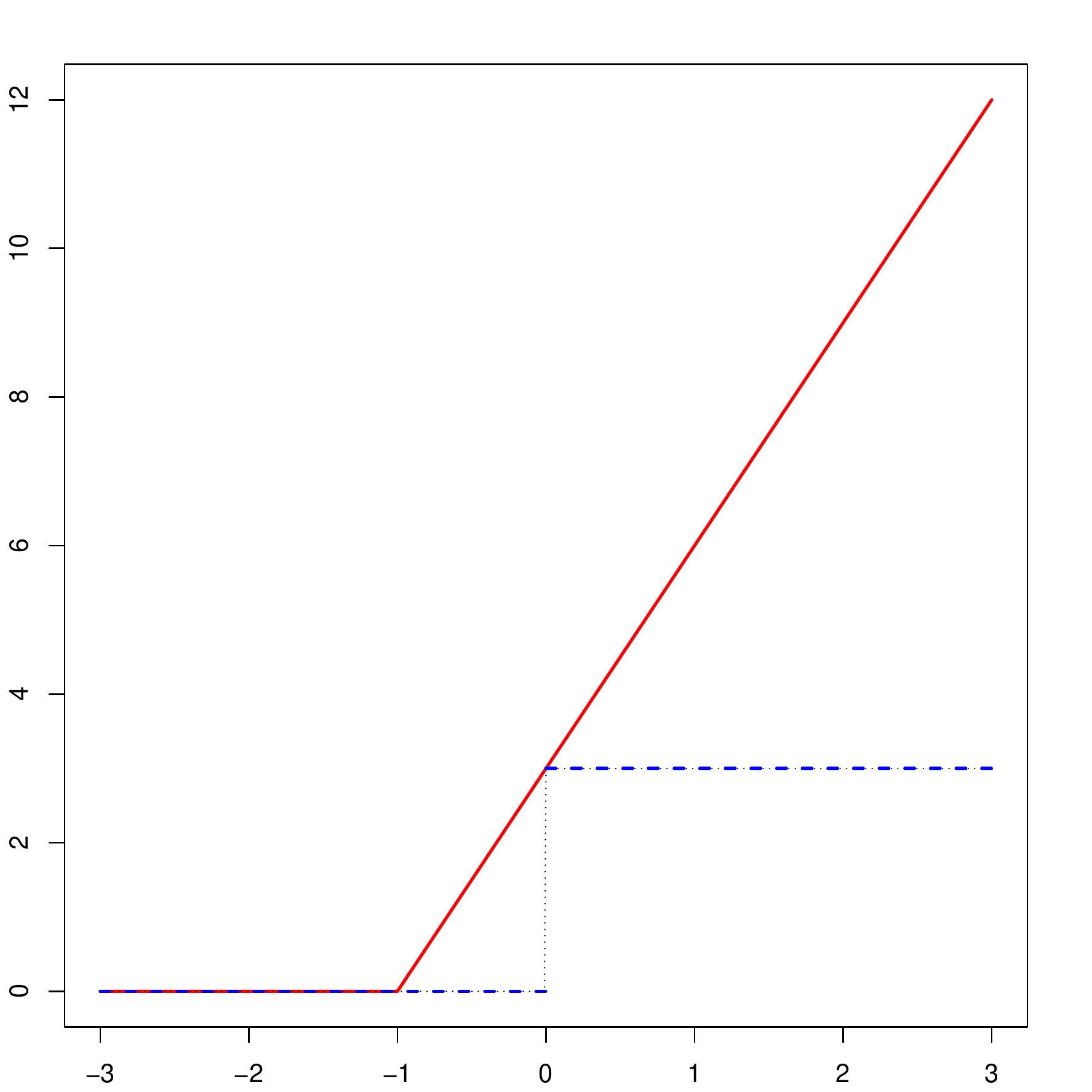}}
     \subfloat[${[}L(\mathbf{x}), U(\mathbf{x}){]} = {[}-3, -1{]}$]{\includegraphics[width = 0.45\columnwidth, height = 6cm]{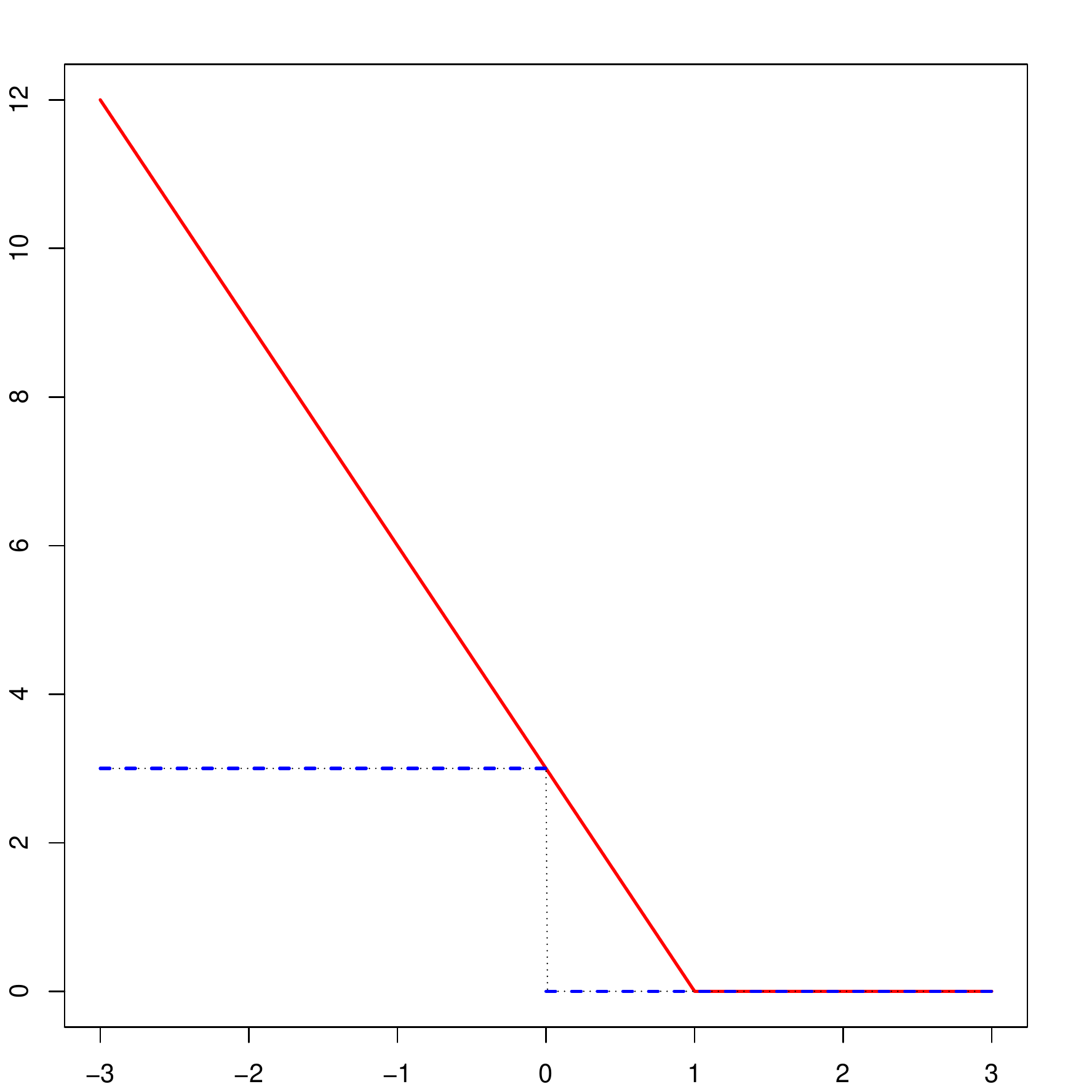}}\\
     \subfloat[${[}L(\mathbf{x}), U(\mathbf{x}){]} = {[}-1, 3{]}$]{\includegraphics[width = 0.45\columnwidth, height = 6cm]{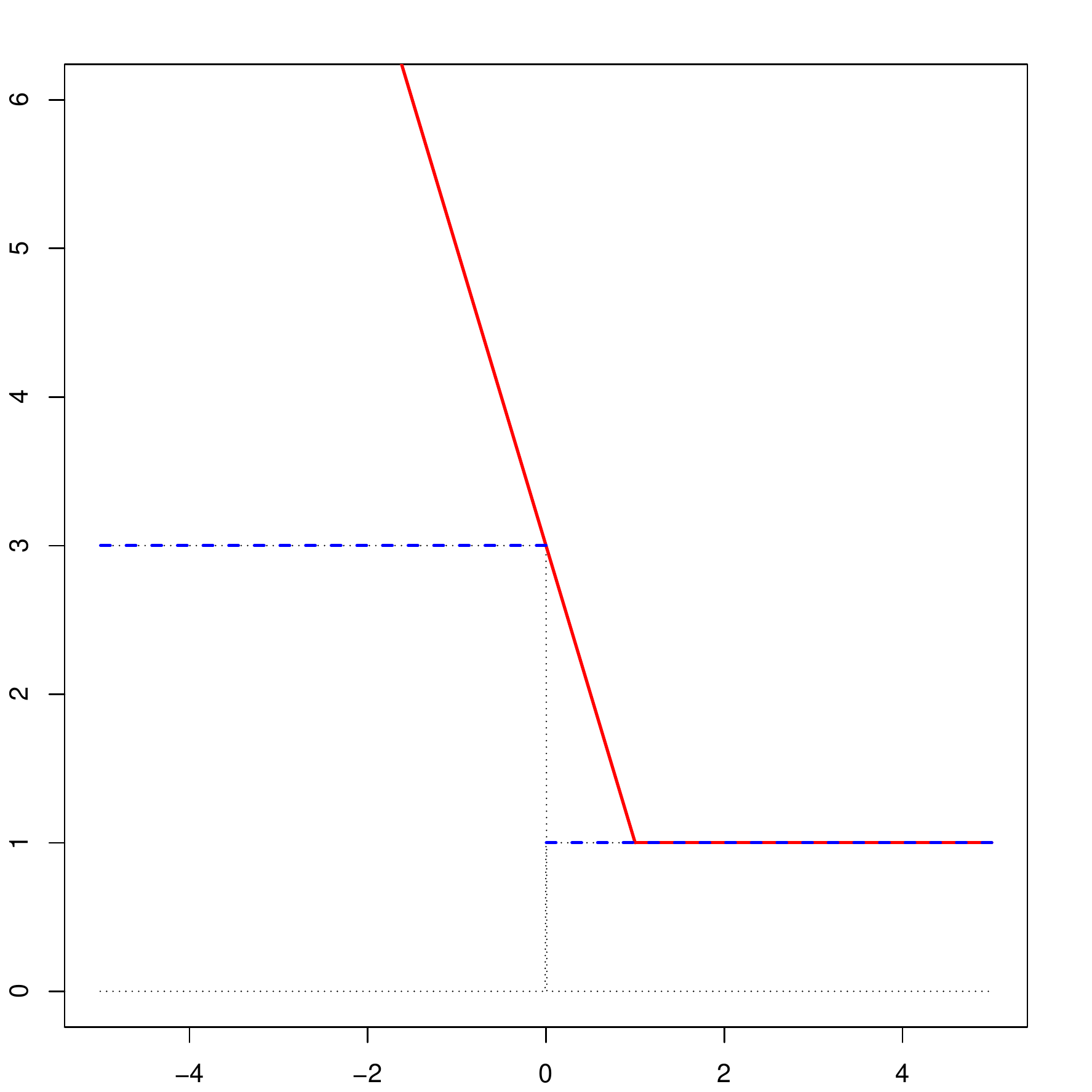}}
     \subfloat[${[}L(\mathbf{x}), U(\mathbf{x}){]} = {[}-3, 1{]}$]{\includegraphics[width = 0.45\columnwidth, height = 6cm]{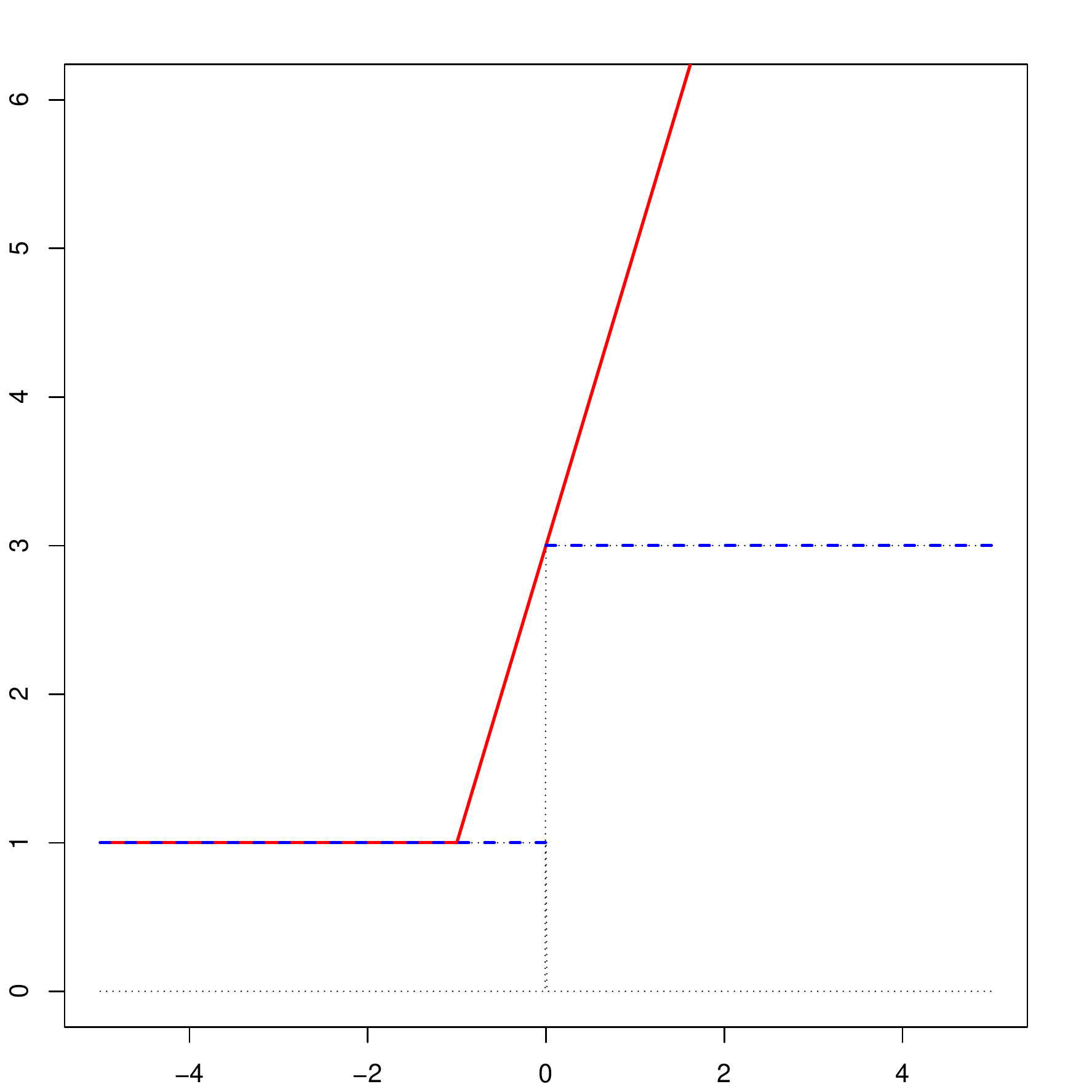}}
\end{figure}

\subsection*{B.2: Deriving Linear/Nonlinear Decision Rules} 
Recall that under the surrogate loss, the objective function is:
\begin{equation}
\begin{split}
    \label{eqn app: empirical version of risk function}
    \widehat{f}(\cdot) = \underset{f \in \mathcal{F}}{\text{argmin}} \sum_{i = 1}^l &\left[\widehat{U}(\mathbf{X}_i)\cdot \phi\{f(\mathbf{X}_i)\}\cdot \mathbbm{1}\big\{\widehat{L}(\mathbf{X}_i) > 0\big\} +
    \{-\widehat{L}(\mathbf{X}_i)\}\cdot \phi\{-f(\mathbf{X}_i)\}\cdot \mathbbm{1}\big\{\widehat{U}(\mathbf{X}_i) < 0\big\}\right] \\
    + \sum_{i = l+1}^n &\bigg[\left[|\widehat{L}(\mathbf{X}_i)| + (|\widehat{U}(\mathbf{X}_i)| - |\widehat{L}(\mathbf{X}_i)|)\cdot\phi\{f(\mathbf{X}_i)\}\right]\cdot\mathbbm{1}\big\{|\widehat{U}(\mathbf{X}_i)| \geq |\widehat{L}(\mathbf{X}_i)|\big\}\\
    &+ \left[|\widehat{U}(\mathbf{X}_i)| + (|\widehat{L}(\mathbf{X}_i)| - |\widehat{U}(\mathbf{X}_i)|)\cdot\phi\{-f(\mathbf{X}_i)\}\right]\cdot\mathbbm{1}\big\{|\widehat{L}(\mathbf{X}_i)| > |\widehat{U}(\mathbf{X}_i)|\big\}\bigg]\\
    &\hspace{-1 cm}+\frac{n\lambda_n}{2} ||f||^{2}.
    \end{split}
\end{equation}
We now derive an efficient algorithm that outputs $\widehat{f}$ as the solution to optimization problem (\ref{eqn app: empirical version of risk function}). Suppose for the moment that $f(\cdot)$ is a linear function of the form $f(\mathbf{x}) = \beta^T \mathbf{x} + \beta_0$. For ease of exposition, we write $\widehat{U}_i = \widehat{U}(\mathbf{X}_i)$ and $\widehat{L}_i = \widehat{L}(\mathbf{X}_i)$. Let $\mathcal{A}_{\text{p}}$ be an index set for subjects with $[\widehat{L}_i, \widehat{U}_i] > 0$, $\mathcal{A}_{\text{n}}$ those with $[\widehat{L}_i, \widehat{U}_i] < 0$, $\mathcal{A}_{\text{sp}}$ those with $[\widehat{L}_i, \widehat{U}_i] \ni 0$ and $|\widehat{L}_i|\leq |\widehat{U}_i|$, and $\mathcal{A}_{\text{sn}}$ those with $[\widehat{L}_i, \widehat{U}_i] \ni 0$ and $|\widehat{L}_i| > |\widehat{U}_i|$.  Objective function (\ref{eqn: empirical version of risk function}) becomes:
\begin{equation}
\begin{split}
    \label{eqn: objective linear case}
    (\widehat{\beta}, \widehat{\beta_0}) =  \underset{\beta, \beta_0}{\text{argmin}}~&\sum_{i \in \mathcal{A}_{\text{p}}} \widehat{U}_i\cdot\{1 - (\beta^T \mathbf{X}_i + \beta_0)\}^+ + \sum_{i \in \mathcal{A}_{\text{n}}} (- \widehat{L}_i)\cdot \{1 + (\beta^T \mathbf{X}_i + \beta_0)\}^+ \\
    &+\sum_{i \in \mathcal{A}_{\text{sp}}}\left[|\widehat{L}_i| + (|\widehat{U}_i| - |\widehat{L}_i|)\cdot\{1 - (\beta^T \mathbf{X}_i + \beta_0)\}^+ \right] \\
    &+\sum_{i \in \mathcal{A}_{\text{sn}}}\left[|\widehat{U}_i| + (|\widehat{L}_i| - |\widehat{U}_i|)\cdot\{1 + (\beta^T \mathbf{X}_i + \beta_0)\}^+
    \right]
    +\frac{n\lambda_n}{2}||f||^{2}.
\end{split}
\end{equation}
Let
\begin{equation}
\label{eqn: w_i def}
    \widehat{w}(\mathbf{X}_i) = |\widehat{U}_i|\cdot\mathbbm{1}\big\{\widehat{L}_i > 0\big\} + |\widehat{L}_i |\cdot\mathbbm{1}\big\{\widehat{U}_i<0 \big\} + \big||\widehat{U}_i | - |\widehat{L}_i |\big|\cdot\mathbbm{1}\big\{[\widehat{L}_i , \widehat{U}_i ] \ni 0\big\},~\forall i,
\end{equation}
and 
\begin{equation}
\label{eqn: e_i def}
    \widehat{e}(\mathbf{X}_i) = \mathbbm{1}\big\{\widehat{U}_i < 0\big\} - \mathbbm{1}\big\{\widehat{L}_i > 0\big\} - \text{sgn}\{|\widehat{U}_i | - |\widehat{L}_i |\}\cdot \mathbbm{1}\big\{[\widehat{L}_i , \widehat{U}_i ] \ni 0\big\},~\forall i.
\end{equation}
Optimization problem (\ref{eqn: objective linear case}) is reduced to the following weighted SVM form:
\begin{equation*}
\begin{split}
  (\widehat{\beta}, \widehat{\beta_0}) =  \underset{\beta, \beta_0}{\text{argmin}}~&\sum_{i = 1}^{n} \widehat{w}(\mathbf{X}_i) \cdot\{1 +\widehat{e}(\mathbf{X}_i)\cdot(\beta^T \mathbf{X}_i + \beta_0)\}^+,   
\end{split}    
\end{equation*}
which is equivalent to:
\begin{equation}
    \begin{split}
        \label{eqn: objective linear equivalence}
    &\underset{\beta, \beta_0}{\text{argmin}} ~~\sum_{i = 1}^n \mathcal{E}_i + \frac{n\lambda_n }{2}\|\beta\|^2 \\
\text{subject to}\quad & \mathcal{E}_i \geq \widehat{w}(\mathbf{X}_i) \cdot\{1 +\widehat{e}(\mathbf{X}_i)\cdot(\beta^T \mathbf{X}_i + \beta_0)\}, ~\forall i,\\
&\mathcal{E}_i \geq 0, ~\forall i, \\
    \end{split}
\end{equation}
where we used the fact that $\mathcal{E}_i \geq \max(a, b) \Leftrightarrow \{\mathcal{E}_i \geq a\} \cap \{\mathcal{E}_i \geq b\}$, and $\widehat{w}(\mathbf{X}_i)$ and $\widehat{e}(\mathbf{X}_i)$ are defined in (\ref{eqn: w_i def}) and (\ref{eqn: e_i def}), respectively.

As the optimization problem is transformed into a particular instance of weighted SVM (\citealp{vapnik2013nature}), it can be readily solved using standard solvers, just like the widely-used outcome weighted learning (OWL) approach proposed by \citet{zhao2012estimating}. Proposition \ref{prop: solution in linear case} gives the representation of the solution $(\widehat{\beta}, \widehat{\beta}_0)$ to the optimization problem defined in (\ref{eqn: objective linear equivalence}). 
\begin{customprop}{B1}\rm
\label{prop: solution in linear case}
Solution $\widehat{\beta}$ to the optimization problem (\ref{eqn: objective linear equivalence}) (and hence to (\ref{eqn: objective linear case})) has the following representation:
\begin{equation}
    \label{eqn: representation of beta in linear case}
    \widehat{\beta} = -\frac{1}{n\lambda_n}\sum_{i =1}^{n} q_i \widehat{w}(\mathbf{X}_i) \widehat{e}(\mathbf{X}_i) \mathbf{X}_i,    
\end{equation}
where $\widehat{w}(\mathbf{X}_i)$ and $\widehat{e}(\mathbf{X}_i)$ are defined in (\ref{eqn: w_i def}) and (\ref{eqn: e_i def}), and $\{q_i, i = 1, \cdots, n\}$ are solutions to the following quadratic programming problem:
\begin{equation}
\begin{split}
    \label{eqn: QP standard form}
    \underset{\mathbf{q}}{\text{min}}\qquad  &\frac{1}{2}\mathbf{q}^T D \mathbf{q} - d^T \mathbf{q} \\
    \text{subject to}\qquad\qquad &0\preceq \mathbf{q}\preceq 1,
\end{split}
\end{equation}
where
\[d=(w_1,w_2,\dots, w_n)\quad\text{and}\quad D_{ij} = \frac{1}{n\lambda_n}\widehat{w}(\mathbf{X}_i) \widehat{w}(\mathbf{X}_j) \widehat{e}(\mathbf{X}_i) \widehat{e}(\mathbf{X}_j)\langle X_i, X_j \rangle.
\]Finally, $\beta_0$ can be solved using the Karush-Kuhn-Tucker (KKT) conditions.  
\end{customprop}

Above derivation can be generalized to nonlinear decision rules. Suppose that $f(\cdot)$ resides in a \emph{reproducing kernel Hilbert space} (RKHS) $\mathcal{H}_{\mathcal{K}}$ associated with the kernel function $\mathcal{K}(\cdot, \cdot): \mathcal{X} \times \mathcal{X} \mapsto \mathbb{R}$. The Hilbert space $\mathcal{H}_{\mathcal{K}}$ is equipped with inner product $\langle \cdot, \cdot \rangle_{\mathcal{K}}$ and norm $\|\cdot\|_{\mathcal{K}}$. Proposition \ref{prop: representation in general case} gives the representation of the optimal rule for $f(x)$ in this case.

\begin{customprop}{B2}\rm
\label{prop: representation in general case}
Optimal decision function $f(x)$ has the following representation:
\begin{equation}
    \label{eqn: representation of beta in general case}
    f(x) = -\frac{1}{n\lambda_n}\sum_{i =1}^{n} q_i  \widehat{w}(\mathbf{X}_i) \widehat{e}(\mathbf{X}_i)  \langle\mathbf{X}_i, x\rangle_{\mathcal{K}} + \widehat{\beta}_{0},
\end{equation}
where
 $\{q_i, i=1,\dots, n\}$ are solutions to the same quadratic programming problem as in Proposition \ref{prop: solution in linear case} with $\langle \mathbf{X}_i, \mathbf{X}_j \rangle_{\mathcal{K}}$ in place of $\langle \mathbf{X}_i, \mathbf{X}_j \rangle$.
\end{customprop}

\subsection*{B.3: Proofs of Proposition \ref{prop: solution in linear case} and \ref{prop: representation in general case}}
We consider the linear case and the general case is analogous. Recall that the objective function in a linear decision boundary case can be rewritten as:
\begin{equation}
    \begin{split}
        \label{eqn: objective linear equivalence 2}
    &\underset{\beta, \beta_0}{\text{argmin}} ~~\sum_{i = 1}^n \mathcal{E}_i + \frac{n\lambda_n }{2}\|\beta\|^2 \\
\text{subject to}\qquad & \mathcal{E}_i \geq \widehat{w}(\mathbf{X}_i) \cdot\{1 +\widehat{e}(\mathbf{X}_i) (\beta^T \mathbf{X}_i + \beta_0)\},~\forall i,\\
&\mathcal{E}_i \geq 0, ~\forall i. \\
    \end{split}
\end{equation}

The Lagrangian of the above optimization problem is 
\begin{equation*}
\begin{split}
    L = &\sum_{i = 1}^n \mathcal{E}_i + \frac{n\lambda_n}{2}\|\beta\|^2 - \sum_{i =1}^{n}  p_i \mathcal{E}_i - \sum_{i =1}^{n} q_i \cdot\left[\mathcal{E}_i - \widehat{w}(\mathbf{X}_i) \cdot\{1 +\widehat{e}(\mathbf{X}_i) (\beta^T \mathbf{X}_i + \beta_0)\}\right], 
\end{split}
\end{equation*}
with $p_i, q_i \geq 0$, and $p_i , q_i$ defined for $i=1,\dots,n$.  

Let $\mathbf{p}$, $\mathbf{q}$, and $\mathbf{\mathcal{E}}$ denote the vector of $p_i$, $q_i$,  and $\mathcal{E}_i$, respectively. Let $g(\mathbf{p}, \mathbf{q}) = \underset{\beta, \beta_0, \mathbf{\mathcal{E}}}{\inf} L$. To minimize $L$, let \[
\frac{\partial L}{\partial \beta} = 0, \quad \frac{\partial L}{\partial \beta_0} = 0, \quad \frac{\partial L}{\partial \mathbf{\mathcal{E}}} = 0,
\] and we arrive at the following system of equations:
\begin{gather*}
        n\lambda_n\beta + \sum_{i =1}^{n} q_i  \widehat{w}(\mathbf{X}_i) \widehat{e}(\mathbf{X}_i)  \mathbf{X}_i =0,~\forall i, \\
        \sum_{i =1}^{n} q_i \widehat{w}(\mathbf{X}_i) \widehat{e}(\mathbf{X}_i) = 0,~\forall i \\
        1 - p_i - q_i = 0,~\forall i.  \\
\end{gather*}
Plug the above equations into the Lagrangian and we have the following dual problem: 
\begin{equation*}
\begin{split}
    \underset{0\preceq \mathbf{q}\preceq 1}{\max} & -\frac{1}{2n\lambda_n}\bigg|\bigg|\sum_{i =1}^{n} q_i  \widehat{w}(\mathbf{X}_i) \widehat{e}(\mathbf{X}_i)  \mathbf{X}_i\bigg|\bigg|^{2}+ \sum_{i =1}^{n} q_i \widehat{w}(\mathbf{X}_i). 
    \end{split}
\end{equation*}
 The dual problem can now be put in the following standard form of a quadratic programming (QP) problem with objective function $\underset{\mathbf{q}}{\text{min}}~~\frac{1}{2}\mathbf{q}^T D \mathbf{q} - d^T \mathbf{q}$, where
$$d=(w_1,w_2,\dots, w_n),$$
$$D_{ij} = \frac{1}{n\lambda_n}\widehat{w}(\mathbf{X}_i) \widehat{w}(\mathbf{X}_j)\widehat{e}(\mathbf{X}_i)\widehat{e}(\mathbf{X}_j)\langle X_i, X_j\rangle,$$
subject to $$0\preceq \mathbf{q}\preceq 1,$$
and linear equality and inequality constraints.

\bigskip
\bigskip
\bigskip
\begin{center}
{\large\bf Supplementary Material C: Proofs of Proposition 1-3 and Theorem \ref{thm: relation between surrogate and true loss}}
\end{center}

\subsection*{C.1: Proof of Proposition \ref{proposition: exchangeability}}
We will prove 
\[\mathbb{E}\left[\sup_{C'(\mathbf{X}) \in [L(\mathbf{X}), U(\mathbf{X})]}  |C'(\mathbf{X})|\cdot \mathbbm{1}\{\text{sgn}\{f(\mathbf{X})\}\} \neq \text{sgn}\{C'(\mathbf{X})\}\}\right]=\sup_{C'(\cdot):~ L \preceq C' \preceq U} \mathcal{R}(f; C'(\cdot)),
\label{eqn: proposition}\]
and Proposition \ref{proposition: exchangeability} follows immediately.\par
On one hand, for any function $C'(\cdot)$ s.t. $L(\cdot) \preceq C'(\cdot) \preceq U(\cdot)$, we have: 
\[\begin{split}
\mathcal{R}(f; C'(\cdot))&=\mathbb{E}\left[ |C'(\mathbf{X})|\cdot \mathbbm{1}\{\text{sgn}\{f(\mathbf{X})\} \neq \text{sgn}\{C'(\mathbf{X})\}\}\right]\\
&\leq 
\mathbb{E}\left[\sup_{C'(\mathbf{X}) \in [L(\mathbf{X}), U(\mathbf{X})]}  |C'(\mathbf{X})|\cdot \mathbbm{1}\{\text{sgn}\{f(\mathbf{X})\} \neq \text{sgn}\{C'(\mathbf{X})\}\}\right].\\
\end{split}
\]
\par
On the other hand, let $C^\ast(x)=L(x)\mathbbm{1}\{f(x)>0\}+U(x)\mathbbm{1}\{f(x)\leq 0  \}$. We have
\[\begin{split}
 &\mathbb{E}\left[\sup_{C'(\mathbf{X}) \in [L(\mathbf{X}), U(\mathbf{X})]}  |C'(\mathbf{X})|\cdot \mathbbm{1}\{\text{sgn}\{f(\mathbf{X})\} \neq \text{sgn}\{C'(\mathbf{X})\}\}\right] \\
 = &\mathcal{R}(f; C^\ast(\cdot)) \leq \sup_{C'(\cdot):~ L \preceq C' \preceq U} \mathcal{R}(f; C'(\cdot)).
 \end{split}
\]
Combine these two inequalities, and we have:
\[\mathbb{E}\left[\sup_{C'(\mathbf{X}) \in [L(\mathbf{X}), U(\mathbf{X})]}  |C'(\mathbf{X})|\cdot \mathbbm{1}\{\text{sgn}\{f(\mathbf{X})\} \neq \text{sgn}\{C'(\mathbf{X})\}\}\right] = \sup_{C'(\cdot):~ L \preceq C' \preceq U} \mathcal{R}(f; C'(\cdot)).
\]

\subsection*{C.2: Proof of Proposition \ref{prop: Bayes decision rule}}
Recall the risk function of a decision rule $f$ is 
\[
\mathcal{R}_{\text{upper}}(f) = \mathbb{E}\left[\max_{C'(\mathbf{X}) \in [L(\mathbf{X}), U(\mathbf{X})]}  |C'(\mathbf{X})|\cdot \mathbbm{1}\{\text{sgn}\{f(\mathbf{X})\} \neq \text{sgn}\{C'(\mathbf{X})\}\}\right].
\]
We now derive the Bayes decision rule. 
\begin{equation*}
    \begin{split}
      &\mathcal{R}_{\text{upper}}(f) \\
      =&  \mathbb{E}\left[\mathbb{E}\left[\max_{C'(\mathbf{X}) \in [L(\mathbf{X}), U(\mathbf{X})]} |C'(\mathbf{X})|\cdot \mathbbm{1}\big\{\text{sgn}\{f(\mathbf{X})\} \neq \text{sgn}\{C'(\mathbf{X})\}\big\}\mid \text{sgn}\{L(\mathbf{X})\}, \text{sgn}\{U(\mathbf{X})\}\right]\right]\\
      =&\mathbb{E}\bigg[|U(\mathbf{X})|\cdot \mathbbm{1}\big\{\text{sgn}\{f(\mathbf{X})\} \neq 1\big\}\cdot \mathbbm{1}\{L(\mathbf{X}) > 0\} + |L(\mathbf{X})|\cdot \mathbbm{1}\big\{\text{sgn}\{f(\mathbf{X})\} \neq -1\big\}\cdot \mathbbm{1}\{U(\mathbf{X}) < 0\} \\
      + &\max\big\{|L(\mathbf{X})|\cdot \mathbbm{1}\{\text{sgn}\{f(\mathbf{X})\}\} \neq -1\},~|U(\mathbf{X})|\cdot \mathbbm{1}\{\text{sgn}\{f(\mathbf{X})\} \neq 1\}\big\}\cdot \mathbbm{1}\big\{[L(\mathbf{X}), U(\mathbf{X})] \ni 0\big\}\bigg]. 
    \end{split}
\end{equation*}
Observe that 
\begin{equation*}
    \begin{split}
        &\max\big\{|L(\mathbf{X})|\cdot \mathbbm{1}\{\text{sgn}\{f(\mathbf{X})\}\} \neq -1\},~|U(\mathbf{X})|\cdot \mathbbm{1}\{\text{sgn}\{f(\mathbf{X})\} \neq 1\}\big\} \\
        =& |L(\mathbf{X})|\cdot \mathbbm{1}\big\{\text{sgn}\{f(\mathbf{X})\}\neq -1\big\} + |U(\mathbf{X})|\cdot \mathbbm{1}\big\{\text{sgn}\{f(\mathbf{X})\} \neq 1\big\}\\
        =&|L(\mathbf{X})|\cdot\left[1 - \mathbbm{1}\big\{\text{sgn}\{f(\mathbf{X})\} \neq 1\big\}\right] + |U(\mathbf{X})|\cdot \mathbbm{1}\big\{\text{sgn}\{f(\mathbf{X})\} \neq 1\big\}\\
        =&|L(\mathbf{X})| + (|U(\mathbf{X})| - |L(\mathbf{X})|)\cdot\mathbbm{1}\big\{\text{sgn}\{f(\mathbf{X})\} \neq 1\big\}.
    \end{split}
\end{equation*}
Then we can deduce that
\begin{equation*}
    \begin{split}
      &\mathcal{R}_{\text{upper}}(f) \\
      =&\mathbb{E}\bigg[|U(\mathbf{X})|\cdot \mathbbm{1}\big\{\text{sgn}\{f(\mathbf{X})\} \neq 1\big\}\cdot \mathbbm{1}\{L(\mathbf{X}) > 0\} + |L(\mathbf{X})|\cdot \big(1-\mathbbm{1}\big\{\text{sgn}\{f(\mathbf{X})\} \neq 1\big\}\big)\cdot \mathbbm{1}\{U(\mathbf{X}) < 0\} \\
      &+ \big[|L(\mathbf{X})| + (|U(\mathbf{X})| - |L(\mathbf{X})|)\cdot\mathbbm{1}\big\{\text{sgn}\{f(\mathbf{X})\} \neq 1\big\}\big]\cdot \mathbbm{1}\big\{[L(\mathbf{X}), U(\mathbf{X})] \ni 0\big\}\bigg]\\
      =&C_0+\mathbb{E}\bigg[\mathbbm{1}\big\{\text{sgn}\{f(\mathbf{X})\} \neq 1\big\}\cdot\\
      &
      \big[|U(\mathbf{X})|\cdot\mathbbm{1}\{L(\mathbf{X}) > 0\} - |L(\mathbf{X})|\cdot\mathbbm{1}\{U(\mathbf{X}) < 0\} + (|U(\mathbf{X})| - |L(\mathbf{X})|)\cdot\mathbbm{1}\{[L(\mathbf{X}), U(\mathbf{X})] \ni 0\}\big]\bigg],
    \end{split}
\end{equation*}
where $C_0$ is a constant that does not depend on $f$.

Recall
\begin{equation*}
    \begin{split}
        \eta(\mathbf{x}) &= |U(\mathbf{x})|\cdot\mathbbm{1}\{L(\mathbf{x}) > 0\} - |L(\mathbf{x})|\cdot\mathbbm{1}\{U(\mathbf{x}) < 0\} + (|U(\mathbf{x})| - |L(\mathbf{x})|)\cdot\mathbbm{1}\{[L(\mathbf{x}), U(\mathbf{x})] \ni 0\},
    \end{split}
\end{equation*}
and we have\[
\mathcal{R}_{\text{upper}}(f) = \mathbb{E}\left[\eta(\mathbf{X})\cdot\mathbbm{1}\big\{\text{sgn}\{f(\mathbf{X})\} \neq 1\big\} \right]+C_0.
\]
Clearly, the expectation is minimized by choosing $f = f^\ast$, where
\[
    \text{sgn}\{f^\ast(\mathbf{x})\} = \begin{cases}
+1, \quad&\text{if}~~\eta(\mathbf{x})\geq 0,\\
-1, &\text{if}~~\eta(\mathbf{x})< 0.
\end{cases}
\]

\subsection*{C.3: Proof of Proposition \ref{prop: excess risk}}
We compute the excess risk for an arbitrary measurable function $f$:
\begin{equation*}
    \begin{split}
        \mathcal{R}_{\text{upper}}(f) - \mathcal{R}^\ast_{\text{upper}}(f) &= \mathcal{R}_{\text{upper}}(f) - \mathcal{R}_{\text{upper}}(f^\ast)\\
        &=\mathbb{E}\bigg[\left[\mathbbm{1}\big\{\text{sgn}\{f(\mathbf{X})\} \neq 1\big\} - \mathbbm{1}\big\{\text{sgn}\{f^\ast(\mathbf{X})\} \neq 1\big\}\right]\cdot\eta(\mathbf{X})\bigg],
    \end{split}
\end{equation*}
where $f^\ast$ is constructed in Proposition \ref{prop: Bayes decision rule}, and $\eta(\mathbf{x})$ is defined in Proposition \ref{prop: Bayes decision rule}.

Observe that 
\begin{equation*}
    \begin{split}
        &\left[\mathbbm{1}\big\{\text{sgn}\{f(\mathbf{X})\} \neq 1\big\} - \mathbbm{1}\big\{\text{sgn}\{f^\ast(\mathbf{X})\} \neq 1\big\}\right]\cdot\eta(\mathbf{X})\\
        =& \mathbbm{1}\big\{\text{sgn}\{f(\mathbf{X})\} \neq \text{sgn}\{f^\ast(\mathbf{X})\}\big\}\cdot \left[\mathbbm{1}\big\{\text{sgn}\{f(\mathbf{X})\} \neq 1\big\} - \mathbbm{1}\big\{\text{sgn}\{f^\ast(\mathbf{X})\} \neq 1\big\}\right]\cdot\eta(\mathbf{X}).
    \end{split}
\end{equation*}
By construction of $f^\ast$, we have
\begin{equation*}
\begin{split}
   &\mathbbm{1}\big\{\text{sgn}\{f(\mathbf{X})\} \neq \text{sgn}\{f^\ast(\mathbf{X})\}\big\}\cdot \left[\mathbbm{1}\big\{\text{sgn}\{f(\mathbf{X})\} \neq 1\big\} - \mathbbm{1}\big\{\text{sgn}\{f^\ast(\mathbf{X})\} \neq 1\big\}\right]\cdot\eta(\mathbf{X}) \\
   =& \begin{cases}
\mathbbm{1}\big\{\text{sgn}\{f(\mathbf{X})\} \neq \text{sgn}\{f^\ast(\mathbf{X})\}\big\}\cdot\eta(\mathbf{X}), \quad&\text{if}~~\eta(\mathbf{X}) \geq 0\\
\mathbbm{1}\big\{\text{sgn}\{f(\mathbf{X})\} \neq \text{sgn}\{f^\ast(\mathbf{X})\}\big\}\cdot\{-\eta(\mathbf{X})\}, \quad&\text{if}~~\eta(\mathbf{X}) < 0
\end{cases}\\
=&\mathbbm{1}\big\{\text{sgn}\{f(\mathbf{X})\} \neq \text{sgn}\{f^\ast(\mathbf{X})\}\big\}\cdot|\eta(\mathbf{X})|
\end{split}
\end{equation*}
Therefore, we have established\[
\mathcal{R}_{\text{upper}}(f) - \mathcal{R}^\ast_{\text{upper}}(f) = \mathbb{E}\left[\mathbbm{1}\big\{\text{sgn}\{f(\mathbf{X})\} \neq \text{sgn}\{f^\ast(\mathbf{X})\}\big\}\cdot|\eta(\mathbf{X})|\right].
\]

\subsection*{C.4: Proof of Theorem \ref{thm: relation between surrogate and true loss}}
Fix $\mathbf{x} \in \mathcal{X}$. The risk under the surrogate loss conditional on $x$ is 
\begin{equation}
\begin{split}
    \text{Conditional $\phi$-Risk} &= |U(\mathbf{x})|\cdot \phi\{f(\mathbf{x})\}\cdot\mathbbm{1}\{L(\mathbf{x}) > 0\} + |L(\mathbf{x})|\cdot \phi\{-f(\mathbf{x})\}\cdot\mathbbm{1}\{U(\mathbf{x}) < 0\}\\ 
    +& \left[|L(\mathbf{x})| + (|U(\mathbf{x})| - |L(\mathbf{x})|)\cdot \phi\{f(\mathbf{x})\}\right]\cdot\mathbbm{1}\big\{[L(\mathbf{x}), U(\mathbf{x})] \ni 0, |U(\mathbf{x})| \geq |L(\mathbf{x})|\big\}\\
    +& \left[|U(\mathbf{x})| + (|L(\mathbf{x})| - |U(\mathbf{x})|)\cdot \phi\{-f(\mathbf{x})\}\right]\cdot\mathbbm{1}\big\{[L(\mathbf{x}), U(\mathbf{x})] \ni 0, |U(\mathbf{x})| < |L(\mathbf{x})|\big\}.
\end{split}
\end{equation}
We follow \citet{bartlett2006convexity} and think of the above conditional $\phi-$risk in terms of a generic classifier value $f(\mathbf{x}) = \alpha \in \mathbb{R}$ and generic $L = L(\mathbf{x}) \in\mathbb{R}$ and $U = U(\mathbf{x}) \in \mathbb{R}$ values. To this end, we define the \emph{generic conditional $\phi$-risk}:
\begin{equation*}
    \begin{split}
        C_{L, U}(\alpha) &= |U|\cdot \phi(\alpha)\cdot\mathbbm{1}\{L > 0\} + |L|\cdot \phi(-\alpha)\cdot\mathbbm{1}\{U < 0\}\\ 
    &+ \{|L| + (|U| - |L|)\cdot \phi(\alpha)\}\cdot\mathbbm{1}\{[L, U] \ni 0, |U| \geq |L|\}\\
    &+ \{|U| + (|L| - |U|)\cdot \phi(-\alpha)\}\cdot\mathbbm{1}\{[L, U] \ni 0, |U| < |L|\}. 
    \end{split}
\end{equation*}

The \emph{optimal conditional $\phi$-risk} is defined to be
\begin{equation*}
    H(L, U) = \underset{\alpha \in \mathbb{R}}{\text{inf}}C_{L, U}(\alpha),
\end{equation*}
and the optimal $\phi-$risk can be written as
\begin{equation}
\label{eqn: optimal phi risk and H}
    \mathcal{R}^{h, \ast}_{\text{upper}} = \mathbb{E}\{H(L(\mathbf{X}), U(\mathbf{X}))\}.
\end{equation}

Straightforward calculation shows that 
\[
H(L, U) = \begin{cases}
0, \quad&L > 0 ~\text{or}~U < 0\\
|L|, &[L, U] \ni 0, ~|U| \geq |L| \\
|U|, &[L, U] \ni 0, ~|U| < |L|.
\end{cases}
\]

For the surrogate loss to be useful, we need to make sure that the optimal conditional $\phi$-risk can be achieved with an $\alpha$ that has the same sign as the optimal rule. To this end, we follow \citet{bartlett2006convexity} and define
\[
H^-(L, U) = \text{inf}\{C_{L, U}(\alpha): \alpha \cdot \eta \leq 0\},
\]where $\eta$ is a shorthand of $\eta(\mathbf{x})$ defined in Proposition \ref{prop: Bayes decision rule}. It is straightforward to show that $H^-(L, U)$ attains its minimum at $\alpha = 0$, with the optimal value: 
\[
H^-(L, U) = |L|\cdot\mathbbm{1}\{[L, U] \ni 0, |U| \geq |L|\} + |U|\cdot\mathbbm{1}\{[L, U] \ni 0, |U| < |L|\}.
\]

We can now compute 
\begin{equation*}
\begin{split}
    H^-(L, U) - H(L, U) &= |U|\cdot\mathbbm{1}\{L > 0\} + |L|\cdot\mathbbm{1}\{U < 0\} + (|U| - |L|)\cdot \mathbbm{1}\{[L, U] \ni 0, |U| \geq |L|\} \\
    &+ (|L| - |U|)\cdot \mathbbm{1}\{[L, U] \ni 0, |U| < |L|\}.
\end{split}
\end{equation*}
It is clear that for any $L \in \mathbb{R}$ and $U \in \mathbb{R}$ such that $L \neq U$, we have $H^-(L, U) - H(L, U) > 0$. This suggests that the surrogate loss is \emph{classification-calibrated} in the terminology of \citet{bartlett2006convexity}. Moreover, observe that 
\[
H^-(L, U) - H(L, U) = \begin{cases}
|U|, \quad&L > 0 \\
|L|, &U < 0\\
||U| - |L||, &[L, U] \ni 0,
\end{cases}
\] and we have
\begin{equation}
\label{eqn: H^- - H = eta}
    H^-(L, U) - H(L, U) = |\eta|.
\end{equation}

Now we are ready to put together everything and prove the proposition:
\begin{equation*}
    \begin{split}
        \mathcal{R}_{\text{upper}}(f) -  \mathcal{R}^\ast_{\text{upper}} &= \mathbb{E}\left[\mathbbm{1}\{f(\mathbf{X}) \neq f^\ast(\mathbf{X})\}\cdot|\eta(\mathbf{X})|\right]\\
        &=\mathbb{E}\left[\mathbbm{1}\{f(\mathbf{X}) \neq f^\ast(\mathbf{X})\}\cdot \big\{H^-\{L(\mathbf{X}), U(\mathbf{X})\} - H\{L(\mathbf{X}), U(\mathbf{X})\}\big\}\right]\\
        &\leq \mathbb{E}\left[C_{L(\mathbf{X}), U(\mathbf{X})}\{f(\mathbf{X})\} - H\{L(\mathbf{X}), U(\mathbf{X})\}\right]\\
        &=\mathbb{E}\left[C_{L(\mathbf{X}), U(\mathbf{X})}\{f(\mathbf{X})\}\right] - \mathbb{E}\left[H\{L(\mathbf{X}), U(\mathbf{X})\}\right]\\
        &=\mathcal{R}^h_{\text{upper}}(f) - \mathcal{R}^{h, \ast}_{\text{upper}},
    \end{split}
\end{equation*}
where the first equality is by Proposition \ref{prop: excess risk}; the second equality is by equation (\ref{eqn: H^- - H = eta}); the third inequality is by the fact that $H^-(L, U)$ minimizes the conditional $\phi$-risk $C_{L, U}$ when $\mathbbm{1}\{f(\mathbf{X}) \neq f^\ast(\mathbf{X})\} = 1$, i.e., $f$ and the optimal rule $f^\ast$ disagree; the last equality is by definition and equation (\ref{eqn: optimal phi risk and H}). This completes the proof of Theorem \ref{thm: relation between surrogate and true loss}.

\clearpage
\begin{center}
{\large\bf Supplementary Material D: Proof of Proposition \ref{prop: approx error} and Convergence Rate Results on $L(\mathbf{X})$ and $U(\mathbf{X})$
}
\end{center}

\subsection*{D.1: Proof of Proposition \ref{prop: approx error}}

\begin{proof}[of the first part]
First, we have
\begin{align*}
\mathcal{R}_{\text{upper}}(f;L(\cdot),U(\cdot))  
=&  \mathbb{E}\left[\sup_{C'(\mathbf{X}) \in [L(\mathbf{X}), U(\mathbf{X})]}  |C'(\mathbf{X})|\cdot \mathbbm{1}\big\{\text{sgn}\{f(\mathbf{X})\} \neq \text{sgn}\{C'(\mathbf{X})\}\big\}\right]\\
\geq&\mathbb{E}\left[ |C(\mathbf{X})|\cdot \mathbbm{1}\big\{\text{sgn}\{f(\mathbf{X})\} \neq \text{sgn}\{C(\mathbf{X})\}\big\}\right]\\
=& \mathcal{R}(f).
\end{align*}
On the other hand we claim that for any $\mathbf{X} = \mathbf{x}$,
\begin{align}
&\sup_{C'(\mathbf{x}) \in [L(\mathbf{x}), U(\mathbf{x})]}  |C'(\mathbf{x})|\cdot \mathbbm{1}\big\{\text{sgn}\{f(\mathbf{x})\} \neq \text{sgn}\{C'(\mathbf{x})\}\big\}-|C(\mathbf{x})|\cdot \mathbbm{1}\big\{\text{sgn}\{f(\mathbf{x})\} \neq \text{sgn}\{C(\mathbf{x})\}\big\}\nonumber\\
&\leq U(\mathbf{x})-L(\mathbf{x}).\label{eqn: loss upper bound}
\end{align}
If $[L(\mathbf{x}),U(\mathbf{x})]$ does not cover $0$, then for any $C'(\mathbf{x})\in [L(\mathbf{x}),U(\mathbf{x})]$, we have
$\mathbbm{1}\big\{\text{sgn}\{f(\mathbf{x})\} \neq \text{sgn}\{C'(\mathbf{x})\}\big\}=\mathbbm{1}\big\{\text{sgn}\{f(\mathbf{x})\} \neq \text{sgn}\{C(\mathbf{x})\}\big\}$. Then we have
\begin{align*}
&\text{LHS of } (\ref{eqn: loss upper bound}) 
\leq \sup_{C'(\mathbf{x}) \in [L(\mathbf{x}), U(\mathbf{x})]}  |C'(\mathbf{x})|-|C(\mathbf{x})|\leq U(\mathbf{x})-L(\mathbf{x}).
\end{align*}
If $[L(\mathbf{x}), U(\mathbf{x})]$ covers $0$, then we have
\begin{align*}
&\text{LHS of }(\ref{eqn: loss upper bound}) \\
\leq& \sup_{C'(\mathbf{x}) \in [L(\mathbf{x}), U(\mathbf{x})]}|C'(\mathbf{x})|\cdot \mathbbm{1}\big\{\text{sgn}\{f(\mathbf{x})\} \neq \text{sgn}\{C'(\mathbf{x})\}\big\}\\
\leq&\sup_{C'(\mathbf{x}) \in [L(\mathbf{x}), U(\mathbf{x})]}|C'(\mathbf{x})|\\
\leq& U(\mathbf{x})-L(\mathbf{x}).
\end{align*}
Combine these results, we prove the first part of the proposition.
\end{proof}

\begin{proof}[of the second part]
To prove the result we only need to show that for every $\mathbf{X} = \mathbf{x}$, we have
\begin{equation}\small
\begin{split}
&\left[|C(\mathbf{x})|\cdot \mathbbm{1}\big\{\text{sgn}\{C(\mathbf{x})\} \neq \text{sgn}\{f^\ast(\mathbf{x})\}\big\}\right] \\
\leq &\left[{\mathbbm{1}\left\{L(\mathbf{x})<0<U(\mathbf{x})\right\}}\cdot{\left\{U(\mathbf{x})-L(\mathbf{x})\right\}}\cdot
{\left\{\frac{1-\rho(\mathbf{x};U, L)}{2}\right\}}\cdot{\mathbbm{1}\left\{\rho^{c}(\mathbf{x};U,L,C)>\rho(\mathbf{x};U,L)\right\}}\right].\label{eqn: part2}
\end{split}
\end{equation}
If $[L(\mathbf{x}), U(\mathbf{x})]$ does not cover $0$, then both LHS and RHS of \eqref{eqn: part2} are $0$ and the inequality is trivially satisfied. If $[L(\mathbf{x}), U(\mathbf{x})] \ni 0$ and $\text{sgn}\{C(\mathbf{x})\} = \text{sgn}\{f^\ast(\mathbf{x})\}$, then we have
\begin{align*}
\text{LHS of } (\ref{eqn: part2})=0\leq \text{RHS of } (\ref{eqn: part2}).    
\end{align*}
  If $[L(\mathbf{x}), U(\mathbf{x})] \ni 0$ and $\text{sgn}\{C(\mathbf{x})\} \neq \text{sgn}\{f^\ast(\mathbf{x})\}$, note that  $\text{sgn}\{f^\ast(\mathbf{x})\}=\text{sgn}\{U(\mathbf{x})+L(\mathbf{x})\}$, and we have $\text{sgn}\{C(\mathbf{x})\} \neq \text{sgn}\left\{\frac{L(\mathbf{x})+U(\mathbf{x})}{2}\right\}$. It follows that
  \begin{align*}
      |C(\mathbf{x})|< \left|C(\mathbf{x})-\frac{L(\mathbf{x})+U(\mathbf{x})}{2}\right|,
  \end{align*}
and thus we have $\mathbbm{1}\left\{\rho^{c}(\mathbf{x};U,L,C)>\rho(\mathbf{x};U,L)\right\}=1$. Therefore, we have
\[
\text{LHS of } (\ref{eqn: part2})=|C(\mathbf{x})|, 
\] and
\[
\text{RHS of } (\ref{eqn: part2})=\left\{U(\mathbf{x})-L(\mathbf{x})\right\}\cdot
{\left\{\frac{1-\rho(\mathbf{x};U, L)}{2}\right\}}.
\]
 Observe that in this case $|C(\mathbf{x})|\leq |U(\mathbf{x})|\wedge |L(\mathbf{x})|$ and $|U(\mathbf{x})|\wedge |L(\mathbf{x})|=\left\{U(\mathbf{x})-L(\mathbf{x})\right\}\cdot
{\left\{\frac{1-\rho(\mathbf{x};U, L)}{2}\right\}}$ we conclude $\text{LHS of } (\ref{eqn: part2})\leq \text{RHS of } (\ref{eqn: part2})$.

\end{proof}

\subsection*{D.2: Convergence Rate of $L(\mathbf{X})$ and $U(\mathbf{X})$ for Balke-Pearl Bounds, Siddique Bounds, and Manski-Pepper bounds}

Lemma D2 states that if $K$ functions are all $n^{-\theta}$ estimable then their linear combinations, maximum/minimum, and product are also $n^{-\theta}$ estimable. Recall that for a binary outcome, the Balke-Pearl bounds and the Siddique bounds, are compositions of a series of functions using maximum/minimum and linear combinations. To obtain an estimate of $L(\mathbf{X})$ and $U(\mathbf{X})$ that satisfy Assumption \ref{assump: rate of convergence of L and U} for Balke-Pearl bounds and Siddique bounds, it suffices to first obtain an estimate for each constituent part, i.e. $\{{p}_{y, a \mid z, \mathbf{X}},~y=\pm 1, a=\pm 1, z=\pm 1\}$, that converges in $L_1$ with $n^{-\theta}$, and then plug all estimates into the $L(\mathbf{X})$ and $U(\mathbf{X})$ expressions. Analogously, for a continuous outcome and the Manski-Pepper bounds, we only need to obtain estimates for $P(Z = z \mid \mathbf{X})$, $\mathbb{E}\{Y \mid Z = z, A = a \mid \mathbf{X}\}$, and $P(A =  a \mid Z = z, \mathbf{X})$ that converge in $L_1$ with $n^{-\theta}$ rate, and then plug these estimates into the expressions of the Manski-Pepper bounds; see Supplementary Material A. 

\begin{customlemma}{D2}
Let $g_i:\mathcal{X}\mapsto \mathbb{R}, i=1,\dots,K$, be $K$ functions of $\mathbf{X}$ to be estimated. Suppose we have estimators  $\{\hat{g}_i,i=1,\dots,K\}$ s.t.
$$\mathbb{E}\left[\big|\hat{g}_i(\mathbf{X})-g_i(\mathbf{X})\big|\right]\leq \delta_i, i=1,2\dots,K$$.
Then we have the following:\par

\begin{enumerate}
    \item For any constant $c_1,c_2,\dots,c_K$: \begin{equation*}
    \begin{split}
        &\mathbb{E}\left[\left|\{c_1 g_1(\mathbf{X})+c_2 g_2(\mathbf{X})+\dots+c_K g_K(\mathbf{X})\}-\{c_1 \hat{g}_1(\mathbf{X})+c_2 \hat{g}_2(\mathbf{X})+\dots +c_K g_K(\mathbf{X})\}\right|\right]\\
        &\leq |c_1|\cdot \delta_1 +|c_2|\cdot \delta_2+\dots+|c_K|\cdot \delta_K.
    \end{split}
\end{equation*}
\item \[
\mathbb{E}\left[|\max\{ g_1(\mathbf{X}), g_2(\mathbf{X}),\dots,g_K(\mathbf{X})\}-\max\{ \hat{g}_1(\mathbf{X}), \hat{g}_2(\mathbf{X}),\dots,\hat{g}_K(\mathbf{X})\}|\right]\leq \delta_1+\delta_2+\dots+\delta_K.
\]
\item 
\[
\mathbb{E}\left[|\min\{ g_1(\mathbf{X}), g_2(\mathbf{X}),\dots,g_K(\mathbf{X})\}-\min\{ \hat{g}_1(\mathbf{X}), \hat{g}_2(\mathbf{X}),\dots,g_K(\mathbf{X})\}|\right]\leq \delta_1+\delta_2+\dots+\delta_K.
\]
\item If we further assume that $|g_{i}(\textbf{X})|,|\hat{g}_i(\textbf{X})|\leq C$ with probability 1 for some constant C and $i=1, 2$, then
\[
\mathbb{E}\left\{|g_1(\textbf{X})g_2(\textbf{X})-\hat{g}_1(\textbf{X})\hat{g}_2(\textbf{X})|\right\}\leq C\cdot(\delta_1+\delta_2).
\]
\end{enumerate}
\end{customlemma}

\begin{proof}
Observe that
\begin{align*}
&\mathbb{E}\left[|\{c_1 g_1(\mathbf{X})+c_2 g_2(\mathbf{X})\}-\{c_1 \hat{g}_1(\mathbf{X})+c_2 \hat{g}_2(\mathbf{X})\}|\right] \\
\leq& \mathbb{E}\left[|c_1|\cdot| g_1(\mathbf{X})-\hat{g}_1(\mathbf{X})|+|c_2|\cdot |g_2(\mathbf{X})-\hat{g}_2(\mathbf{X})|\right] \\
\leq& |c_1|\cdot \delta_1 +|c_2|\cdot \delta_2.
\end{align*}
Apply this result iteratively and part (a) is proved by induction.

Next, we prove an upper bound for $|\max\{g_1(\mathbf{X}),g_2(\mathbf{X})\}-\max\{\hat{g}_1(\mathbf{X}),\hat{g}_2(\mathbf{X})\}|$ by cases . If $g_1(\mathbf{X})\geq g_2(\mathbf{X})$, then we can deduce that
\begin{align*}
&\max\{g_1(\mathbf{X}),g_2(\mathbf{X})\}-\max\{\hat{g}_1(\mathbf{X}),\hat{g}_2(\mathbf{X})\} \\
=&g_1(\mathbf{X})-\max\{\hat{g}_1(\mathbf{X}),\hat{g}_2(\mathbf{X})\}\\
\leq&g_1(\mathbf{X})-\hat{g}_1(\mathbf{X})\\
\leq& |g_1(\mathbf{X})-\hat{g}_1(\mathbf{X})|.
\end{align*}
By symmetry, if $g_1(\mathbf{X})< g_2(\mathbf{X})$, then we can prove
$$\max\{g_1(\mathbf{X}),g_2(\mathbf{X})\}-\max\{\hat{g}_1(\mathbf{X}),\hat{g}_2(\mathbf{X})\} \leq |g_2(\mathbf{X})-\hat{g}_2(\mathbf{X})|.$$
Combine these two results we have 
\begin{align}\max\{g_1(\mathbf{X}),g_2(\mathbf{X})\}-\max\{\hat{g}_1(\mathbf{X}),\hat{g}_2(\mathbf{X})\} \leq |g_2(\mathbf{X})-\hat{g}_2(\mathbf{X})|+|g_1(\mathbf{X})-\hat{g}_1(\mathbf{X})|.\label{eq: max-max upper}
\end{align}
On the other hand, if $\hat{g}_1(\mathbf{X})\geq \hat{g}_2(\mathbf{X})$, then we can deduce that
\begin{align*}
&\max\{g_1(\mathbf{X}),g_2(\mathbf{X})\}-\max\{\hat{g}_1(\mathbf{X}),\hat{g}_2(\mathbf{X})\} \\
=&\max\{g_1(\mathbf{X}),g_2(\mathbf{X})\}-\hat{g}_1(\mathbf{X})\\
\geq&g_1(\mathbf{X})-\hat{g}_1(\mathbf{X})\\
\geq& -|g_1(\mathbf{X})-\hat{g}_1(\mathbf{X})|.
\end{align*}
By symmetry, if $\hat{g}_1(\mathbf{X})< \hat{g}_2(\mathbf{X})$, then we can prove
$$\max\{g_1(\mathbf{X}),g_2(\mathbf{X})\}-\max\{\hat{g}_1(\mathbf{X}),\hat{g}_2(\mathbf{X})\} \geq -|g_2(\mathbf{X})-\hat{g}_2(\mathbf{X})|.$$
Therefore we have
\begin{align}\max\{g_1(\mathbf{X}),g_2(\mathbf{X})\}-\max\{\hat{g}_1(\mathbf{X}),\hat{g}_2(\mathbf{X})\} \geq -|g_2(\mathbf{X})-\hat{g}_2(\mathbf{X})|-|g_1(\mathbf{X})-\hat{g}_1(\mathbf{X})|.\label{eq: max-max lower}
\end{align}
Finally combine inequalities (\ref{eq: max-max upper}) and (\ref{eq: max-max lower}), we have
\begin{align}|\max\{g_1(\mathbf{X}),g_2(\mathbf{X})\}-\max\{\hat{g}_1(\mathbf{X}),\hat{g}_2(\mathbf{X})\}|\leq |g_2(\mathbf{X})-\hat{g}_2(\mathbf{X})|+|g_1(\mathbf{X})-\hat{g}_1(\mathbf{X})|.\label{eq: lemma upper}\end{align}

Use inequality (\ref{eq: lemma upper}) iteratively and we have
\begin{align*}
&|\max\{ g_1(\mathbf{X}), g_2(\mathbf{X}),\dots,g_K(\mathbf{X})\}-\max\{ \hat{g}_1(\mathbf{X}), \hat{g}_2(\mathbf{X}),\dots,\hat{g}_K(\mathbf{X})\}|   \\
=&|\max\left\{\max\{ g_1(\mathbf{X}), g_2(\mathbf{X}),\dots,g_{K-1}(\mathbf{X})\},g_K(\mathbf{X})\right\}-\max\left\{\max\{ \hat{g}_1(\mathbf{X}), \hat{g}_2(\mathbf{X}),\dots,\hat{g}_{K-1}(\mathbf{X})\},\hat{g}_K(\mathbf{X})\right\}|\\
\leq&|\max\{ g_1(\mathbf{X}), g_2(\mathbf{X}),\dots,g_{K-1}(\mathbf{X})\}-\max\{ \hat{g}_1(\mathbf{X}), \hat{g}_2(\mathbf{X}),\dots,\hat{g}_{K-1}(\mathbf{X})\}|+|g_K(\mathbf{X})-\hat{g}_K(\mathbf{X})|\\
\leq&\dots\\
\leq&\sum_{i=1}^{K}|g_i(\mathbf{X})-\hat{g}_i(\mathbf{X})|.
\end{align*}
Then we can conclude:
\begin{align*}
&\mathbb{E}\left[\left|\max\{ g_1(\mathbf{X}), g_2(\mathbf{X}),\dots,g_K(\mathbf{X})\}-\max\{ \hat{g}_1(\mathbf{X}), \hat{g}_2(\mathbf{X}),\dots,\hat{g}_K(\mathbf{X})\}\right|\right]    \\
\leq&\mathbb{E}\left[\sum_{i=1}^{K}|g_i(\mathbf{X})-\hat{g}_i(\mathbf{X})|\right]\\
\leq&\sum_{i=1}^{K}\delta_i,
\end{align*}
and part (b) is proved. Proof of part (c) is analogous to that of part (b) and is omitted.
\par

Finally, to prove part (d), it suffices to observe that
\begin{align*}
&\mathbb{E}\left\{|g_1(\textbf{X})g_2(\textbf{X})-\hat{g}_1(\textbf{X})\hat{g}_2(\textbf{X})|\right\}\\
\leq&\mathbb{E}\left[|g_1(\textbf{X})\left\{g_2(\textbf{X})-\hat{g}_2(\textbf{X})\right\}|+|\left\{\hat{g}_1(\textbf{X})-g_1(\textbf{X})\right\}\hat{g}_2(\textbf{X})|\right]\\
\leq& C\cdot \mathbb{E}\left[|g_2(\textbf{X})-\hat{g}_2(\textbf{X})|+|\hat{g}_1(\textbf{X})-g_1(\textbf{X})|\right]\\
\leq&C\cdot (\delta_1+\delta_2).
\end{align*}
\end{proof}

\bigskip
\bigskip
\begin{center}
{\large\bf Supplementary Material E: Proofs of Lemmas and Theorem \ref{thm: conv rate} }
\end{center}

\subsection*{E.1: Proof of Theorem \ref{thm: conv rate}}
First notice that $\widehat{f}_n ^{\lambda_n}$ is the minimizer of $$\frac{1}{|I_2|}\sum_{i\in I_2}  
l(\textbf{X}_i ;\widehat{w},\widehat{e},f)+
\frac{\lambda_{n}}{2}||f||^{2}.$$
Consider a function $f_0 = 0$ everywhere, then 
\begin{equation}
\begin{split}
\frac{1}{|I_2|}\sum_{i\in I_2}  
l(\textbf{X}_i ;\widehat{w},\widehat{e},f)+
\frac{\lambda_{n}}{2}||\widehat{f}_n ^{\lambda_n}||^{2} &\leq \frac{1}{|I_2|}\sum_{i\in I_2}  
l(\textbf{X}_i ;\widehat{w},\widehat{e},f_0)+
\frac{\lambda_{n}}{2}||f_0||^{2} \\
&=\frac{1}{|I_2|}\sum_{i\in I_2}\widehat{w}(\textbf{X}_i).
\label{ineql: bound of f hat}
\end{split}
\end{equation}
According to Assumption \ref{assump: boundedness of L U estimates}, we have $|\widehat{L}|,|\widehat{U}|\leq M_3$; therefore $|\widehat{w}(\textbf{X}_i)|\leq 2M_3$.
Plug his into (\ref{ineql: bound of f hat}) and we have:
$$2M_3 \geq \frac{1}{|I_2|}\sum_{i\in I_2}  
l(\textbf{X}_i ;\widehat{w},\widehat{e},f)+
\frac{\lambda_{n}}{2}||\widehat{f}_n ^{\lambda_n}||^{2} \geq \frac{\lambda_{n}}{2}||\widehat{f}_n ^{\lambda_n}||^{2}.$$ 
This implies that $||\widehat{f}_n ^{\lambda_n}||\leq \frac{2\sqrt{M_3}}{\sqrt{\lambda_n}}$. 
Let $f_n ^*$ be the optimal function with respect to penalized risk:
$$f_n ^*= \underset{f\in \mathcal{F}}{\text{argmin}}~\mathbb{E}l(\textbf{X};w,e,f)+\frac{\lambda_n}{2}||f||^2
$$
When the penalty factor is set to be $\lambda_n$, $f_{n}^{*}$ is the best we can get. 
Similar norm bound results for $f_n ^*$ exist:
\begin{equation*}
\begin{split}
\mathbb{E}l(\textbf{X};w,e,f_n ^*)+\frac{\lambda_n}{2}||f_n ^* ||^2&\leq \mathbb{E}l(\textbf{X};w,e,f_0 )+\frac{\lambda_n}{2}||f_0 ||^2\\
&\leq 2M_1.
\end{split}
\end{equation*}
Therefore $||f_n ^*||\leq \frac{2\sqrt{M_1}}{\sqrt{\lambda_n}}$.
According to assumption we have $||f||\geq M_4 ||f||_\infty$ for any $f\in \mathcal{F}$. Therefore, $||f_n ^*||_\infty, ||\widehat{f}_n ^{\lambda_n}||_{\infty} \leq \frac{2\sqrt{M_1 \vee M_3}}{M_4\sqrt{\lambda_n}}$.
 Define $B_{n}$ to be the set of functions $f$ s.t. $f\in \mathcal{F}$ and $||f||_{\infty}\leq  \frac{2\sqrt{M_3 \vee M_1}}{M_4\sqrt{\lambda_n}}$. Obviously $f_n ^* , \widehat{f}_n ^{\lambda_n}\in B_n$. By assumption \ref{assump: entropy}, we also have $ \widehat{f}_n ^{\lambda_n}\in B_n^*$.
 According to Assumption \ref{assump: existence of fnite minimizer}, there exists an optimal rule $f_\mathcal{F}^*$ with finite norm. We naturally have the following risk decomposition:
\[
\begin{split}
&\mathcal{R}^{\text{h}}_{\text{upper}}(\widehat{f}_{n}^{\lambda_n}) - \inf_{f\in \mathcal{F}}\mathcal{R}^{\text{h}}_{\text{upper}}(f) \\
=& \mathbb{E} l(\textbf{X}; w, e, \widehat{f}_{n}^{\lambda_n})-\mathbb{E} l(\textbf{X}; w, e, f_\mathcal{F}^{\ast})\\
=&  \underbrace{\mathbb{E} l(\textbf{X}; w, e, \widehat{f}_{n}^{\lambda_n})-\mathbb{E} l(\textbf{X}; w, e, f_n ^*)-\frac{\lambda_n}{2}||f_n ^*||^2}_{Q1} + \underbrace{\mathbb{E} l(\textbf{X}; w, e, f_n ^*)+\frac{\lambda_n}{2}||f_n ^*||^2-\mathbb{E} l(\textbf{X}; w, e, f_\mathcal{F}^{\ast})}_{Q2}.
\end{split}
\]
We  bound Q1 and Q2 separately.  Note that $\widehat{w}$ and $\widehat{e}$ are also random variables in addition to $\textbf{X}$. We first bound $Q_1$ by decomposing $Q_1$ into five terms. Note that
\[
\begin{split}
Q1 \leq & \mathbb{E} l(\textbf{X}; w, e, \widehat{f}_{n}^{\lambda_n})+\frac{\lambda_n}{2}||\widehat{f}_n ^{\lambda_n}||^2 -\mathbb{E} l(\textbf{X}; w, e, f_n ^*)-\frac{\lambda_n}{2}||f_n ^*||^2\\
=& \underbrace{\mathbb{E}\left\{\sum_{i\in I_2}\frac{1}{|I_2 |}l(\textbf{X}_i ;w,e,f_n ^* )-\mathbb{E} l(\textbf{X};w,e,f_n ^* )\right\}}_{Q11}\\
&+ \underbrace{\mathbb{
E}\left\{ \sum_{i\in I_2}\frac{1}{|I_2 |}l(\textbf{X}_i ;\widehat{w},\widehat{e},f_n ^* )- \sum_{i\in I_2}\frac{1}{|I_2 |}l(\textbf{X}_i ;w,e,f_n ^* )\right\}}_{Q12}\\
&+ \underbrace{\mathbb{E}\left\{\frac{1}{|I_2 |}\sum_{i\in I_2} l(\textbf{X}_i ; \widehat{w},\widehat{e},\widehat{f}_n ^{\lambda_n})+\frac{\lambda_n}{2}||\widehat{f}_n ^{\lambda_n}||^2- \frac{1}{|I_2 |}\sum_{i\in I_2} l(\textbf{X}_i ; \widehat{w},\widehat{e},f_n ^* )-\frac{\lambda_n}{2}||f_n ^* ||^2\right\}}_{Q13}\\
&+ \underbrace{\mathbb{E}\left\{\frac{1}{|I_2 |}\sum_{i\in I_2} l(\textbf{X}_i ; w,e,\widehat{f}_n ^{\lambda_n}) -\sum_{i\in I_2} \frac{1}{|I_2 |}l(\textbf{X}_i ; \widehat{w},\widehat{e},\widehat{f}_n ^{\lambda_n})\right\}}_{Q14}\\
&+ \underbrace{\mathbb{E}\left\{\mathbb{E}l(\textbf{X} ; w,e,\widehat{f}_n ^{\lambda_n}) - \sum_{i\in I_2}\frac{1}{|I_2 |} l(\textbf{X}_i ; w,e,\widehat{f}_n ^{\lambda_n})\right\}}_{Q15}
\end{split}
\]
Below we will bound $Q_{11}$, $Q_{12}$, $Q_{13}$, $Q_{14}$, and $Q_{15}$ separately.  

First, we have
\begin{equation}    \label{proof part1}
\begin{split}
  Q_{15}
    &\leq \mathbb{E}_{\textbf{X}_{|I_1|}}\mathbb{E}_{\textbf{X}_{|I_2|}}\left|\mathbb{E}l(\textbf{X} ; w,e,\hat{f}_{n}^{\lambda_n}) - \sum_{i\in I_2}\frac{1}{|I_2 |} l(\textbf{X}_i ; w,e,\hat{f}_{n}^{\lambda_n})\right|\\
    &\leq\mathbb{E}_{\textbf{X}_{|I_1 |}} \mathbb{E}_{\textbf{X}_{[I_2 ]}}\sup_{f\in B_n^*}\left|\mathbb{E}l(\textbf{X} ; w,e,f) - \sum_{i\in I_2}\frac{1}{|I_2 |} l(\textbf{X}_i ; w,e,f)\right|(\text{ by $\hat{f}_{n}^{\lambda_n}\in B_n^*$}) \\
    &= \mathbb{E}_{\textbf{X}_{[I_1]}}Q_{11}' \leq O\left(\frac{1}{\sqrt{n}\cdot\lambda_n}\right),
\end{split}
\end{equation}
where the last inequality follows from Lemma \ref{lemma: lemma 2}. Besides we have $Q_{11}=0$ by simple calculation.\par
Second, we have
\begin{equation}\label{proof part2} 
    \begin{split}
      Q_{14}
      &\leq \mathbb{E}_{\textbf{X}_{[I_1]}}\mathbb{E}_{\textbf{X}_{[I_2]}}\left\{\frac{1}{|I_2 |}\sum_{i\in I_2} l(\textbf{X}_i ; w,e,\hat{f}_{n}^{\lambda_n}) -\sum_{i\in I_2} \frac{1}{|I_2 |}l(\textbf{X}_i ; \widehat{w},\widehat{e},\hat{f}_{n}^{\lambda_n})\right\} \\
      &\leq \mathbb{E}_{\textbf{X}_{[I_1]}}\sup_{f\in B_n }\left|\mathbb{E}_{\textbf{X}_{[I_2]}}\left\{\frac{1}{|I_2 |}\sum_{i\in I_2} l(\textbf{X}_i ; w,e,f) -\sum_{i\in I_2} \frac{1}{|I_2 |}l(\textbf{X}_i ; \widehat{w},\widehat{e},f)\right\}\right| \\
        &\leq O\left(\lambda_n^{-\frac{1}{2}}n^{-(\alpha\wedge \beta)}\right),
    \end{split}
\end{equation}
where the last inequality follows from Lemma \ref{lemma: difference caused by hat w and w}. Similarly,
\begin{equation}
    \begin{split}
        Q_{12} &\leq \mathbb{E}_{\textbf{X}_{[I_1]}}\sup_{f\in B_n }\left|\mathbb{E}_{\textbf{X}_{[I_2]}}\left\{\frac{1}{|I_2 |}\sum_{i\in I_2} l(\textbf{X}_i ; w,e,f) -\sum_{i\in I_2} \frac{1}{|I_2 |}l(\textbf{X}_i ; \widehat{w},\widehat{e},f)\right\}\right| \\
        &\leq O\left(\lambda_n^{-\frac{1}{2}}n^{-(\alpha\wedge \beta)}\right)
    \end{split}
\end{equation}

Third, by definition of $\widehat{f}_{n}^{\lambda_n}$, we have
\begin{align}
    Q_{13}\leq 0.
\end{align}

Finally, we have
\begin{equation}
    \begin{split}
        Q_{2}=&\mathbb{E} l(\textbf{X}; w, e, f_n ^*)+\frac{\lambda_n}{2}||f_n ^*||^2-\mathbb{E} l(\textbf{X}; w, e, f_\mathcal{F}^{\ast})\\
=&   \mathbb{E} l(\textbf{X}; w, e, f_n ^*)+\frac{\lambda_n}{2}||f_n ^*||^2-\mathbb{E} l(\textbf{X}; w, e, f_\mathcal{F}^{\ast})-\frac{\lambda_n}{2}||f_\mathcal{F}^{\ast}||^2+\frac{\lambda_n}{2}||f_\mathcal{F}^{\ast}||^2\\
\leq & \frac{\lambda_n}{2}||f_\mathcal{F}^{\ast}||^2 ~~~(\text{By definition of $f_{n}^{*}$}) \\
\leq & O\left(\lambda_n\right)
~~(\text{By assumption \ref{assump: existence of fnite minimizer}}).    \end{split}
\end{equation}
Combine all the results above, we conclude
\begin{equation*}
    \begin{split}
 &\mathcal{R}^{\text{h}}_{\text{upper}}(\widehat{f}_{n}^{\lambda_n}) - \inf_{f\in \mathcal{F}}\mathcal{R}^{\text{h}}_{\text{upper}}(f)\\
 \leq & Q_{11}+Q_{12}+Q_{13}+Q_{14}+Q_{15}+Q_{2}\\
 \leq & O(\lambda_n+n^{-\frac{1}{2}}\lambda_n^{-1}+\lambda_n^{-\frac{1}{2}}(n^{-\alpha}+n^{-\beta}),
\end{split}
\end{equation*} and this completes the proof of Theorem 2.

\subsection*{E.2: Proof of Lemma \ref{lemma: difference caused by hat w and w}}
Without loss of generality, let $1\in I_2$. Note that
\begin{equation}
\begin{split}
    &\mathbb{E}_{\textbf{X}_{[I_1]}}\sup_{f\in B_n }\left|\mathbb{E}_{\textbf{X}_{[I_2]}}\left\{\frac{1}{|I_2 |}\sum_{i\in I_2} l(\textbf{X}_i ; w,e,f) -\sum_{i\in I_2} \frac{1}{|I_2 |}l(\textbf{X}_i ; \widehat{w},\widehat{e},f)\right\}\right| \\
  &\leq \mathbb{E}_{\textbf{X}_{[I_1]}}\sum_{i\in I_2}\sup_{f\in B_n }\left|\mathbb{E}_{\textbf{X}_{i}}\left\{\frac{1}{|I_2 |} l(\textbf{X}_i ; w,e,f) - \frac{1}{|I_2 |}l(\textbf{X}_i ; \widehat{w},\widehat{e},f)\right\}\right| \\
  &=\mathbb{E}_{\textbf{X}_{[I_1]}}\sup_{f\in B_n }\left|\mathbb{E}_{\textbf{X}_{1}}\{ l(\textbf{X}_1 ; w,e,f) - l(\textbf{X}_1 ; \widehat{w},\widehat{e},f)\}\right|\\
   &\leq \mathbb{E}_{\textbf{X}_{[I_1]}}\sup_{f\in B_n }\mathbb{E}_{\textbf{X}_{1}}| l(\textbf{X}_1 ; w,e,f) - l(\textbf{X}_1 ; \widehat{w},\widehat{e},f)|.\\
  \label{eqn: lemma1 partial}
\end{split}
\end{equation}
Moreover, note that
\begin{equation}
\begin{split}
&|l(x;w,e,f)-l(x;\widehat{w},\widehat{e},f)|\\
=&|\{w(x)-\widehat{w}(x)\} + f(x)\{w(x)e(x)-\widehat{w}(x)\widehat{e}(x)\}|\\
\leq &|\{w(x)-\widehat{w}(x)\}|+|f(x)\{w(x)e(x)-\widehat{w}(x)\widehat{e}(x)\}|\\
\leq &|\{w(x)-\widehat{w}(x)\}| + O\left(\frac{1}{\sqrt{\lambda_n}}\right)|w(x)e(x)-\widehat{w}(x)\widehat{e}(x)|,
\end{split}
\end{equation}
where the last inequality follows because $f \in B_n$.

Plug this into (\ref{eqn: lemma1 partial}) and we have
\begin{equation*}
\begin{split}
    &\mathbb{E}_{\textbf{X}_{[I_1]}}\sup_{f\in B_n }\mathbb{E}_{\textbf{X}_{1}}| l(\textbf{X}_1 ; w,e,f) - l(\textbf{X}_1 ; \widehat{w},\widehat{e},f)|\\
    \leq & \mathbb{E}_{\textbf{X}_{[I_1]}}\sup_{f\in B_n }\mathbb{E}_{\textbf{X}_{1}}\left\{\left|w(\textbf{X}_1)-\widehat{w}(\textbf{X}_1)\right|+O\left(\frac{1}{\sqrt{\lambda_n}}\right)\left|w(\textbf{X}_1)e(\textbf{X}_1)-\widehat{w}(\textbf{X}_1)\widehat{e}(\textbf{X}_1)\right|\right\}\\
    =& \mathbb{E}_{\textbf{X}_{[I_1]}}\mathbb{E}_{\textbf{X}_{1}}\left\{\left|w(\textbf{X}_1)-\widehat{w}(\textbf{X}_1)\right|+O\left(\frac{1}{\sqrt{\lambda_n}}\right)\left|w(\textbf{X}_1)e(\textbf{X}_1)-\widehat{w}(\textbf{X}_1)\widehat{e}(\textbf{X}_1)\right|\right\}.\\
    =& \mathbb{E}\left\{\left|w(\textbf{X}_1)-\widehat{w}(\textbf{X}_1)\right|+O\left(\frac{1}{\sqrt{\lambda_n}}\right)\left|w(\textbf{X}_1)e(\textbf{X}_1)-\widehat{w}(\textbf{X}_1)\widehat{e}(\textbf{X}_1)\right|\right\}\\
    \leq & n^{-(\alpha \wedge \beta)}O\left(1+\frac{1}{\sqrt{\lambda_n}}\right)~~(\text{by Lemma \ref{lemma: what and w} to be stated and proved in C.3})\\
    \leq &O\left(n^{-(\alpha \wedge \beta)}/\sqrt{\lambda_n}\right)~~(\text{because}~\lambda_n\ \text{is\ } o(1)).
\end{split}
\end{equation*}

\subsection*{E.3: Lemma \ref{lemma: what and w} and Proof}
\begin{customlemma}{E1}\label{lemma: what and w}\rm
Under Assumption \ref{assump: rate of convergence of L and U}, we have
\begin{align*}
    \mathbb{E}\left\{|\widehat{w}(\textbf{X})-w(\textbf{X})|\right\} \leq 
    \mathbb{E}\left\{|\widehat{w}(\textbf{X})\widehat{e}(\textbf{X})-w(\textbf{X})e(\textbf{X})|\right\} \leq O\left(n^{-(\alpha \wedge \beta)}\right)
\end{align*}
\end{customlemma}
\par
\begin{proof}
Consider 
\[
\begin{split}
w'(x) =w(x)e(x) = &-U(x)\cdot\mathbbm{1}\{{L}(x)  > 0\} - L(x) \cdot\mathbbm{1}\{U(x)<0 \} \\
&+ \{-|U(x) | + |L(x)|\}\cdot\mathbbm{1}\{[L(x), U(x) ] \ni 0\},
\end{split}
\] and 
\[
\begin{split}
\widehat{w}'(x)=\widehat{w}(x)\widehat{e}(x)= &-\widehat{U}(x)\cdot\mathbbm{1}\{{\widehat{L}}(x)  > 0\} - \widehat{L}(x) \cdot\mathbbm{1}\{\widehat{U}(x)<0 \} \\
&+ \{-|\widehat{U}(x)| + |\widehat{L}(x) |\}\cdot\mathbbm{1}\{[\widehat{L}(x), \widehat{U}(x) ] \ni 0\}.
\end{split}
\]
Then it can be easily verified that
\begin{align*}
    w(x)=|w'(x)| \quad \text{and} \quad e(x) = \text{sgn}\{w'(x)\},
\end{align*}
and similarly
\begin{align*}
    \widehat{w}(x)=|\widehat{w}'(x)| \quad \text{and} \quad \widehat{e}(x) = \text{sgn}\{\widehat{w}'(x)\}.
\end{align*}
Therefore, we have
\begin{align*}
    |\widehat{w}(x)\widehat{e}(x)-w(x)e(x)|=|\widehat{w}'(x)-w'(x)|
\end{align*}
and
\begin{align*}
    |w(x)-\widehat{w}(x)|=\big||w'(x)|-|\widehat{w}'(x)|\big|\leq |\widehat{w}'(x)-w'(x)|,
\end{align*}
which directly implies that
\begin{equation}
\mathbb{E}\left\{|\widehat{w}(\textbf{X})-w(\textbf{X})|\right\}\leq 
    \mathbb{E}\left\{|\widehat{w}(\textbf{X})\widehat{e}(\textbf{X})-w(\textbf{X})e(\textbf{X})|\right\}.
    \label{eqn: E[|w_hat - w|]}
\end{equation}

Below we will bound $|\widehat{w}'(x)-w'(x)|$. Let
\begin{equation}
    \begin{split}
        S = &|\widehat{w}'(x)-w'(x)|\\
= &\bigg|\left[-U(x)\cdot\mathbbm{1}\{{L}(x)  > 0\} - L(x) \cdot\mathbbm{1}\{U(x)<0 \} + \{-|U(x) | + |L(x) |\}\cdot\mathbbm{1}\big\{[L(x), U(x)] \ni 0\big\}\right]\notag\\
&-\left[-\widehat{U}(x)\cdot\mathbbm{1}\{{\widehat{L}}(x)  > 0\} - \widehat{L}(x) \cdot\mathbbm{1}\{\widehat{U}(x)<0 \} + \{-|\widehat{U}(x) | + |\widehat{L}(x)|\}\cdot\mathbbm{1}\big\{[\widehat{L}(x), \widehat{U}(x) ] \ni 0\big\}\right]\bigg|
\end{split}
\label{eqn: S}
\end{equation}
We claim 
\begin{align}
    S\leq |\widehat{U}(x)-U(x)|+|\widehat{L}(x)-L(x)|.\label{ineq: lemma 3 case}
\end{align}
We prove this inequality case by case:
\begin{itemize}
    \item[Case 1:] If one of the following three clauses holds, it's straightforward to verify (\ref{ineq: lemma 3 case}):
\begin{itemize}
    \item $L(x)>0$ and $\widehat{L}(x)>0$;
    \item $U(x)<0$ and $\widehat{U}(x)<0$;
    \item $[{L}(x), {U}(x) ] \ni 0$ and $[\widehat{L}(x), \widehat{U}(x) ] \ni 0$.
\end{itemize}
\item[Case 2:] If $L(x)>0$ and $[\widehat{L}(x), \widehat{U}(x) ] \ni 0$, 
\begin{align*}
    S=|\widehat{U}(x)-U(x)+\widehat{L}(x)|\leq |\widehat{U}(x)-U(x)|-\widehat{L}(x)\leq |\widehat{U}(x)-U(x)|+|\widehat{L}(x)-L(x)|.
\end{align*}\par
By symmetry, if $\widehat{L}(x)>0$ and $[L(x), {U}(x) ] \ni 0$, (\ref{ineq: lemma 3 case}) holds too.
\item[Case 3:] If $L(x)>0$ and $\widehat{U}(x)<0$, then
\begin{align*}
    S=-\widehat{L}(x)+U(x)\leq U(x)-\widehat{U}(x)+L(x)-\widehat{L}(x)\leq |\widehat{U}(x)-U(x)|+|\widehat{L}(x)-L(x)|.
\end{align*}\par
By symmetry, if $\widehat{L}(x)>0$ and ${U}(x)<0$ (\ref{ineq: lemma 3 case}) holds too.
\item[Case 4:] If $\widehat{U}(x)<0$ and $[{L}(x), {U}(x) ] \ni 0$, then
\begin{align*}
    S=|-L(x)-U(x)+\widehat{L}(x)|\leq |\widehat{L}(x)-L(x)|+U(x)\leq |\widehat{U}(x)-U(x)|+|\widehat{L}(x)-L(x)|
\end{align*}\par
By symmetry, if ${U}(x)<0$ and $[\widehat{L}(x), \widehat{U}(x) ] \ni 0$, (\ref{ineq: lemma 3 case}) holds too.
\end{itemize}
Therefore, we claim that 
\begin{align*}
    S\leq |\widehat{U}(x)-U(x)|+|\widehat{L}(x)-L(x)|.
\end{align*}

Combine (\ref{eqn: E[|w_hat - w|]}) and (\ref{eqn: S}) and we finally prove Lemma 3:
\begin{equation*}
    \begin{split}
        &\mathbb{E}\left\{|\widehat{w}(\textbf{X})\widehat{e}(\textbf{X})-w(\textbf{X})e(\textbf{X})|\right\} \\
        = &\mathbb{E}\left\{|\widehat{w}'(\textbf{X})-w'(\textbf{X})|\right\}\\
\leq &\mathbb{E} \left\{\left|\widehat{U}(\textbf{X})-U(\textbf{X})\right| + \left|\widehat{L}(\textbf{X})-L(\textbf{X})\right|\right\}\\
\leq &O(n^{-(\alpha \wedge \beta)}).
\end{split}
\end{equation*}
\end{proof}

\subsection*{E.4: Proof of Lemma \ref{lemma: lemma 2}}
We prove Lemma 2 using theorem 2.14.1 of \citet{van1996weak}.
Consider the function class $\mathcal{F}_{w,e,n}=\{l(x;w,e,f):f\in B_n^*\}$. Then 
$$Q_{11}'=\frac{1}{\sqrt{|I_2|}}\mathbb{E}||\mathbb{G}_{n}||_{\mathcal{F}_{w,e,n}}.$$
Define a constant function $F(x)=M_1 \left(1+\frac{2\sqrt{M_3 \vee M_1}}{M_4\sqrt{\lambda_n }} \right)$, and we have
\begin{equation*}
    \begin{split}
        |l(x;w,e,f)|=&|w(x)\{1+e(x)f(x)\}^+| \\
\leq & |w(x)\{1+e(x)f(x)\}|\\
\leq & M_1 \cdot \left(1+\frac{2\sqrt{M_3 \vee M_1}}{M_4\sqrt{\lambda_n }}\right)\\
=& F(x)=O\left(1/\sqrt{\lambda_n}\right).
    \end{split}
\end{equation*}
Therefore, $F$ is an envelop function for $\mathcal{F}_{w,e,n}$.\par

According to theorem 2.14.1 of \citet{van1996weak}, we have
\begin{align*}
    \mathbb{E}||\mathbb{G}_{n}||_{\mathcal{F}_{w,e,n}}\leq O\left(\sup_{Q}\int_{0}^{1}\sqrt{1+\log N(\epsilon\cdot ||F||_{Q,2},\mathcal{F}_{w,e,n},L_2(Q))}d\epsilon\cdot ||F||_{P,2}\right)\label{eqn: rate},
\end{align*}
where the supremum is taken over all discrete probability measures Q with $||F||_{Q,2}>0$ and the inequality holds up to a universal constant. Note $||F||_{Q,2}\propto \frac{1}{\sqrt{\lambda_n}}$ and
 by Assumption \ref{assump: entropy} we have 
 \begin{align*}
\log N(\epsilon\cdot ||F||_{Q,2},\mathcal{F}_{w,e,n},L_2(Q))
\leq& \log N(\epsilon\cdot ||F||_{Q,2},\frac{C_0}{\lambda_n}\mathcal{F}_0,L_2(Q))  \\
\leq& \log N(\epsilon\cdot ||F||_{Q,2}\cdot \frac{\lambda_n}{C_0},\mathcal{F}_0,L_2(Q))  \\
\leq& O\big((\epsilon\cdot \sqrt{\lambda_n})^{-v}\big).
 \end{align*}
It follows that 
$$\sup_{Q}\int_{0}^{1}\sqrt{1+\log N(\epsilon\cdot ||F||_{Q,2},\mathcal{F}_{w,e,n},L_2(Q))}d\epsilon\leq O(\lambda_n^{-v/4})\leq \lambda_n^{-\frac{1}{2}}.$$
Therefore, we have
\begin{align*}
    \mathbb{E}||\mathbb{G}_{n}||_{\mathcal{F}_{w,e,n}}\leq O\left(\frac{1}{\sqrt{\lambda_n}}\cdot||F||_{P,2}\right)\leq O\left(1/{\lambda_n}\right),
\end{align*}
which implies that $Q_{11}'\leq O\left(\frac{1}{\sqrt{n}\cdot\lambda_n}\right)$.

Finally we have
\begin{equation*}
    \begin{split}
\mathbb{E}_{\textbf{X}_{[I_1]}}\mathbb{E}_{\textbf{X}_{[I_2 ]}}\sup_{f\in B_n^*}\left|\mathbb{E}l(\textbf{X} ; w,e,f) - \sum_{i\in I_2}\frac{1}{|I_2 |} l(\textbf{X}_i ; w,e,f)\right| \leq \mathbb{E}Q_{11}'\leq O\left(\frac{1}{\sqrt{n}\cdot\lambda_n}\right).
    \end{split}
\end{equation*}

\clearpage
\begin{center}
{\large\bf Supplementary Material F: Additional Simulations}
\end{center}

\subsection*{F.1: Data Generating Processes and Simulation Results when $g_2(\mathbf{X}, U, A)$ is Model (2)}
Data generating processes considered in the main article does not necessarily satisfy the IV identification assumptions underpinning the approach proposed by \citet{cui2019semiparametric}. To see this, we follow \citet{cui2019semiparametric} and \citet{wang2018bounded} and denote
\[
\Tilde{\delta}(\mathbf{X}, U) = P(A = 1 \mid Z = 1, \mathbf{X}, U) - P(A = 1 \mid Z = -1, \mathbf{X}, U),
\] and 
\[
\Tilde{\gamma}(\mathbf{X}, U) = \mathbb{E}[Y(1)\mid \mathbf{X}, U] - \mathbb{E}[Y(-1) \mid \mathbf{X}, U].
\]
\citet{wang2018bounded} showed that one sufficient condition to point identify the treatment effect is the following: $\text{Cov}(\Tilde{\delta}(\mathbf{X}, U), \Tilde{\gamma(\mathbf{X}, U)} \mid \mathbf{X}) = 0$. It is easy to see that according to the DGP in Section \ref{sec: simulation}, we have
\[
\Tilde{\delta}(X, U) = \text{expit}\{8 + X_1 - 7X_2 + \lambda(1+X_1)U\} - \text{expit}\{-8 + X_1 - 7X_2 + \lambda(1+X_1)U\},
\]
which depends on $U$ unless $\lambda = 0$. Similarly, observe that 
\[
\Tilde{\gamma}(X, U) = \text{expit}\{g_1(\mathbf{X}, U) + g_2(\mathbf{X}, U, 1)\} - \text{expit}\{g_1(\mathbf{X}, U) + g_2(\mathbf{X}, U, -1)\}
\] depends on $U$ except when $\xi = \delta = 0$. Therefore, we have $\text{Cov}(\Tilde{\delta}(\mathbf{X}, U), \Tilde{\gamma(\mathbf{X}, U)} \mid \mathbf{X}) \neq 0$ in general and the data-generating processes considered in the simulation section does \emph{not} always satisfy the IV identification assumptions studied in \citet{wang2018bounded} and \citet{cui2019semiparametric}.

Table \ref{tbl: simulation 2} is analogous to Table \ref{tbl: simulation 1} in the main article; it summarizes the simulation results when $g_2(\mathbf{X}, U, A)$ is set to be Model (2).

\begin{table}[h]
\centering
\caption{\small Estimated weighted misclassification error for different $(\lambda, \delta, \xi)$ and $g_1(\mathbf{X}, U)$ combinations. We take $g_2(\mathbf{X}, U, A)$ to be Model (2) throughout. Training data sample size $n_{\text{train}} = 300$ or $500$. Each number in the cell is averaged over $500$ simulations. Standard errors are in parentheses.}
\label{tbl: simulation 2}
\resizebox{\columnwidth}{!}{%
\fbox{%
\begin{tabular}{ccccccc} \\[-0.8em]
   &\multicolumn{2}{c}{IV-PILE-RF} &\multicolumn{2}{c}{OWL-RF} &COIN-FLIP &$C_{\text{DGP}}$\\[-0.9em] \\ \\[-1em] \\[-0.8em]
$\xi = 0$, $g_1 = \text{Model}~(1)$& $n_{\text{train}} = 300$ & $n_{\text{train}} = 500$ & $n_{\text{train}} = 300$ & $n_{\text{train}} = 500$ \\[-0.9em] \\ \\[-1em] \cline{1-1} \\[-0.8em]
$(\lambda, \delta) = (0.5, 0.5)$  &0.019 (0.019) &0.017 (0.016) & 0.194 (0.011) &0.195 (0.008) &0.111 (0.000) & 0.001\\
$(\lambda, \delta) = (1.0, 1.0)$  &0.023 (0.021) &0.022 (0.019) &0.194 (0.011) &0.196 (0.008) &0.114 (0.000) & 0.003\\
$(\lambda, \delta) = (1.5, 1.5)$  &0.034 (0.026) &0.031 (0.020) &0.198 (0.011) &0.199 (0.009) &0.119 (0.000) & 0.009\\
$(\lambda, \delta) = (2.0, 2.0)$  &0.043 (0.033) &0.042 (0.024) &0.199 (0.012) &0.202 (0.008) &0.126 (0.000) & 0.018\\ \\[-0.8em]\hline 
\\[-0.8em]
$\xi = 1$, $g_1 = \text{Model}~(1)$ \\[-0.9em] \\ \cline{1-1} \\[-0.9em]
$(\lambda, \delta) = (0.5, 0.5)$  &0.026 (0.024) &0.020 (0.019) &0.184 (0.010) &0.185 (0.008) &0.105 (0.000) & 0.000\\
$(\lambda, \delta) = (1.0, 1.0)$  &0.032 (0.028) &0.029 (0.024) &0.183 (0.010) &0.183 (0.008) &0.107 (0.000) & 0.003\\
$(\lambda, \delta) = (1.5, 1.5)$ &0.042 (0.034) &0.038 (0.028) &0.181 (0.010) &0.181 (0.008) &0.109 (0.000) & 0.005\\
$(\lambda, \delta) = (2.0, 2.0)$ &0.063 (0.042) &0.052 (0.032) &0.183 (0.011) &0.182 (0.008) &0.114 (0.000) & 0.009\\ \\[-0.8em]\hline 
\\[-0.8em]
$\xi = 0$, $g_1 = \text{Model}~(2)$  \\[-0.9em] \\ \cline{1-1} \\[-0.9em]
$(\lambda, \delta) = (0.5, 0.5)$  &0.047 (0.046) &0.039 (0.036) & 0.214 (0.011) &0.215 (0.009) &0.123 (0.000) & 0.001\\
$(\lambda, \delta) = (1.0, 1.0)$  &0.060 (0.051) &0.046 (0.037) &0.217 (0.011) &0.217 (0.008) &0.127 (0.000) & 0.005\\
$(\lambda, \delta) = (1.5, 1.5)$  &0.075 (0.054) &0.067 (0.045) &0.220 (0.011) &0.221 (0.008) &0.133 (0.000) & 0.012\\
$(\lambda, \delta) = (2.0, 2.0)$  &0.102 (0.060) &0.095 (0.050) &0.227 (0.011) &0.226 (0.009) &0.141 (0.000) & 0.021\\ \\[-0.8em]\hline 
\\[-0.8em]
$\xi = 1$, $g_1 = \text{Model}~(2)$  \\[-0.9em] \\ \cline{1-1} \\[-0.9em]
$(\lambda, \delta) = (0.5, 0.5)$  &0.046 (0.044) &0.036 (0.033) &0.213 (0.011) &0.214 (0.008) &0.123 (0.000) & 0.001\\
$(\lambda, \delta) = (1.0, 1.0)$  &0.066 (0.055) &0.048 (0.039) &0.216 (0.010) &0.216 (0.008) &0.126 (0.000) & 0.004\\
$(\lambda, \delta) = (1.5, 1.5)$ &0.079 (0.056) &0.067 (0.044) &0.220 (0.011) &0.220 (0.009) &0.132 (0.000) &0.012\\
$(\lambda, \delta) = (2.0, 2.0)$ &0.105 (0.063) &0.099 (0.052) &0.226 (0.010) &0.226 (0.008) &0.140 (0.000) & 0.021\\
\end{tabular}}}
\end{table}

\subsection*{F.2: Simulation Results Comparing IV-PILE and OWL when relevant probabilities in $L(\mathbf{X})$ and $U(\mathbf{X})$ are Estimated via Multinomial Logistic Regression and Random Forest with Different Node Sizes}

In this section, we report simulations results for IV-PILE and OWL when relevant conditional probabilities in $L(\mathbf{X})$ and $U(\mathbf{X})$ are estimated with simple but misspecified parametric regression models and random forests with different node sizes. The default node size setting as implemented in the \textsf{randomforest} package in \textsf{R} is $1$ for classification problems, and simulation results corresponding to using this default setting has been reported in the main article. We consider three additional settings in this section: 1) fitting relevant conditional probabilities with multinomial logistic regression; 2) fitting all relevant conditional probabilities with random forest model with node size = 5; 3) fitting all relevant conditional probabilities with random forest model with node size = 10. We report simulations results with $n_{\text{train}} = 300$ in Table \ref{tbl: app n_train 300} and $n_{\text{train}} = 500$ in Table \ref{tbl: app n_train 500}. Note that $C_{\text{DGP}}$ in these cases are the same as before as $C_{\text{DGP}}$ depends only on the data-generating process and not on the specific algorithms. All qualitative trends summarized in the main article apply in these additional simulations. Notably, IV-PILE outperforms OWL in all additional simulation settings, when relevant conditional probabilities are fitted using the same method. We found that tuning the node size in the random forest classifier might improve the performance in some settings, and random forest with properly selected node size outperformed the misspecified parametric models in general.

\begin{table}[ht]
\centering
\caption{\small Estimated weighted misclassification error for different combinations of $(\lambda, \delta, \xi)$, $g_2(\mathbf{X}, U, A)$, and methods for estimating relevant conditional probabilities. LOGIT: all relevant conditional probabilities are estimated via multinomial logistic regression;  RF-5: random forest with node size = 5; RF-10: random forest with node size = 10. We take $g_1(\mathbf{X}, U)$ to be Model (1) throughout. Training data sample size $n_{\text{train}} = 300$. Each number in the cell is averaged over $500$ simulations. Standard errors are in parentheses.}
\label{tbl: app n_train 300}
\resizebox{\columnwidth}{!}{
\begin{tabular}{ccccccc}\hline \\[-0.8em]
   &\multicolumn{3}{c}{IV-PILE-RF} &\multicolumn{3}{c}{OWL-RF} \\[-0.9em] \\ \\[-1em] \\[-0.8em]
$(\lambda, \delta)$ & LOGIT & RF-5 & RF-10 & LOGIT & RF-5 & RF-10\\[-0.9em] \\ \\[-1em] \\[-0.8em] \hline \\[-0.8em]
&&&$\xi = 0, g_2 = (1)$ \\[-0.9em] \\  \\[-0.9em]
$ (0.5, 0.5)$  &0.006 (0.004) &0.005 (0.001) &0.006 (0.002) &0.020 (0.003) &0.015 (0.003) &0.016 (0.003)\\
$ (1.0, 1.0)$  &0.010 (0.003) &0.008 (0.001) & 0.009 (0.002) &0.023 (0.003) &0.018 (0.003) &0.019 (0.003)  \\ 
$ (1.5, 1.5)$  &0.015 (0.003) & 0.014 (0.001) & 0.014 (0.003) & 0.028 (0.003) &0.023 (0.003) & 0.024 (0.003)  \\
$ (2.0, 2.0)$  &0.022 (0.003) & 0.021 (0.001) & 0.021 (0.003) & 0.034 (0.003) &0.029 (0.003) & 0.030 (0.003) \\ \\[-0.8em]\hline 
\\[-0.8em]
&&&$\xi = 0$, $g_2 = (2)$   \\[-0.9em] \\  \\[-0.9em]
$ (0.5, 0.5)$   &0.004 (0.007) &0.007 (0.009) & 0.003 (0.005) & 0.160 (0.018) & 0.189 (0.013) & 0.153 (0.052)\\
$ (1.0, 1.0)$   &0.007 (0.008) &0.010 (0.011) &0.006 (0.007) & 0.161 (0.016) &0.190 (0.012) & 0.184 (0.013)\\
$ (1.5, 1.5)$   &0.013 (0.007) &0.017 (0.013) &0.012 (0.007) & 0.166 (0.014) & 0.195 (0.011) &0.189 (0.012)\\
$ (2.0, 2.0)$   &0.022 (0.010) &0.027 (0.017) &0.021 (0.010) & 0.172 (0.015) & 0.198 (0.012) & 0.193 (0.013)\\ \\[-0.8em]\hline 
\\[-0.8em]
&&&$\xi = 1, g_2 = (1)$   \\[-0.9em] \\  \\[-0.9em]
$(0.5, 0.5)$  &0.007 (0.005) &0.004 (0.001) &0.004 (0.002) &0.019 (0.003) &0.014 (0.004) &0.015 (0.004) \\
$ (1.0, 1.0)$  &0.010 (0.006) &0.007 (0.001) &0.008 (0.002) &0.020 (0.002) &0.016 (0.003) &0.017 (0.003) \\
$(1.5, 1.5)$  &0.013 (0.002) & 0.012 (0.001) & 0.012 (0.002) &0.023 (0.002) & 0.020 (0.003) & 0.020 (0.003)\\
$ (2.0, 2.0)$  &0.019 (0.002) & 0.018 (0.001) & 0.019 (0.002) & 0.028 (0.002) & 0.025 (0.002) & 0.026 (0.002) \\ \\[-0.8em]\hline 
\\[-0.8em]
&&&$\xi = 1$, $g_2 = (2)$   \\[-0.9em] \\  \\[-0.9em]
$ (0.5, 0.5)$   &0.018 (0.043) & 0.015 (0.026) & 0.004 (0.009) & 0.155 (0.015) & 0.180 (0.011) & 0.126 (0.066)\\
$ (1.0, 1.0)$  &0.027 (0.041) &0.012 (0.016) &0.006 (0.009) &0.155 (0.014) &0.179 (0.011) &0.175 (0.012)\\
$ (1.5, 1.5)$  &0.039 (0.045) & 0.018 (0.019) & 0.010 (0.012) & 0.156 (0.012) & 0.178 (0.011) & 0.174 (0.011)\\
$ (2.0, 2.0)$  &0.038 (0.033) & 0.028 (0.023) & 0.016 (0.013) & 0.159 (0.012) & 0.179 (0.011) & 0.175 (0.012)\\\hline 
\end{tabular}}
\end{table}

\begin{table}[ht]
\centering
\caption{\small Estimated weighted misclassification error for different combinations of $(\lambda, \delta, \xi)$, $g_2(\mathbf{X}, U, A)$, and methods for estimating relevant conditional probabilities. LOGIT: all relevant conditional probabilities are estimated via multinomial logistic regression;  RF-5: random forest with node size = 5; RF-10: random forest with node size = 10. We take $g_1(\mathbf{X}, U)$ to be Model (1) throughout. Training data sample size $n_{\text{train}} = 500$. Each number in the cell is averaged over $500$ simulations. Standard errors are in parentheses.}
\label{tbl: app n_train 500}
\resizebox{\columnwidth}{!}{
\begin{tabular}{ccccccc}\hline \\[-0.8em]
   &\multicolumn{3}{c}{IV-PILE-RF} &\multicolumn{3}{c}{OWL-RF} \\[-0.9em] \\ \\[-1em] \\[-0.8em]
$(\lambda, \delta)$ & LOGIT & RF-5 & RF-10 & LOGIT & RF-5 & RF-10\\[-0.9em] \\ \\[-1em] \\[-0.8em]\hline \\[-0.8em]
&&&$\xi = 0, g_2 = (1)$ \\[-0.9em] \\  \\[-0.9em]
$ (0.5, 0.5)$  &0.006 (0.004) &0.005 (0.001) &0.005 (0.001) &0.021 (0.002) &0.015 (0.003) &0.016 (0.003)\\
$ (1.0, 1.0)$  &0.010 (0.006) &0.008 (0.001) &0.008 (0.002) &0.024 (0.002) &0.018 (0.003) &0.019 (0.003)  \\ 
$ (1.5, 1.5)$  &0.014 (0.003) &0.013 (0.001) &0.013 (0.001) &0.029 (0.002) &0.023 (0.003) & 0.024 (0.003) \\ 
$ (2.0, 2.0)$  &0.021 (0.003) &0.020 (0.001) &0.020 (0.002) &0.035 (0.002) &0.030 (0.003) & 0.030 (0.003) \\ \\[-0.8em]\hline 
\\[-0.8em]
&&&$\xi = 0$, $g_2 = (2)$   \\[-0.9em] \\ \\[-0.9em]
$ (0.5, 0.5)$   &0.014 (0.045) & 0.002 (0.003) & 0.003 (0.005) & 0.155 (0.014) & 0.192 (0.009) & 0.153 (0.059)\\
$ (1.0, 1.0)$   &0.018 (0.035) &0.009 (0.009) &0.006 (0.006) & 0.158 (0.013) &0.192 (0.009) &0.187 (0.010)\\
$ (1.5, 1.5)$  &0.030 (0.041) &0.016 (0.012) &0.012 (0.007) & 0.161 (0.013) & 0.194 (0.009) & 0.189 (0.104)\\
$ (2.0, 2.0)$  &0.038 (0.050) &0.025 (0.013) &0.020 (0.007) & 0.167 (0.012) & 0.200 (0.009) &0.194 (0.010) \\ \\[-0.8em]\hline 
\\[-0.8em]
&&&$\xi = 1, g_2 = (1)$   \\[-0.9em] \\  \\[-0.9em]
$ (0.5, 0.5)$  &0.007 (0.005) &0.004 (0.001) &0.004 (0.002) &0.019 (0.003) &0.014 (0.004) &0.015 (0.004) \\
$ (1.0, 1.0)$  &0.009 (0.006) &0.007 (0.001) &0.007 (0.001) &0.021 (0.002) &0.016 (0.003) &0.017 (0.003)\\
$ (1.5, 1.5)$  &0.012 (0.002) & 0.012 (0.001) &0.012 (0.001) & 0.024 (0.002) & 0.020 (0.003) & 0.021 (0.002)\\
$ (2.0, 2.0)$ &0.019 (0.002) &0.018 (0.000) & 0.018 (0.001) &0.029 (0.002) &0.025 (0.002) &0.026 (0.002) \\ \\[-0.8em]\hline 
\\[-0.8em]
&&&$\xi = 1$, $g_2 = (2)$   \\[-0.9em] \\  \\[-0.9em]
$(0.5, 0.5)$  &0.030 (0.061) & 0.002 (0.004) & 0.003 (0.007) &0.152 (0.012) & 0.182 (0.008) & 0.159 (0.031)\\
$ (1.0, 1.0)$  &0.036 (0.047) &0.012 (0.012) &0.006 (0.008) &0.153 (0.010) &0.181 (0.007) &0.177 (0.008)  \\
$ (1.5, 1.5)$  &0.052 (0.050) &0.017 (0.016) & 0.010 (0.010) & 0.154 (0.010) & 0.180 (0.008) & 0.176 (0.009)\\
$ (2.0, 2.0)$ &0.042 (0.045) &0.025 (0.019) &0.016 (0.012) &0.156 (0.009) & 0.181 (0.008) & 0.177 (0.008)\\\hline 
\end{tabular}}
\end{table}

\clearpage
\subsection*{F.3: Simulation Results of OWL-RF and IV-PILE-RF when $n_{\text{train}}$ = 1000, 2000, and 3000}

We report additional simulation results when $n_{\text{train}} = 1000$, $2000$, and $3000$. We let $g_1$ be Model (1), $g_2$ be Model $(1)$, $\xi = 0$, $\lambda = \delta = 0.5$, and all relevant conditional probabilities are estimated via random forests with default settings (node size = 1). Table \ref{tbl: app n_train 1000-3000} summarizes these results. Again, we point out that $C_{\text{DGP}}$ does not depend on $n_{\text{train}}$ and the same as before.

\begin{table}[h]
\centering
\caption{\small Estimated weighted misclassification error for larger $n_{\text{train}}$. $g_1$ is taken to be Model (1); $g_2$ is taken to be Model (1), $\xi = 0$, $(\lambda, \delta) = (0.5, 0.5)$. All relevant conditional probabilities are estimated via random forests with default settings. Each number in the cell is averaged over $500$ simulations. Standard errors are in parentheses.}
\label{tbl: app n_train 1000-3000}
\begin{tabular}{cccc}\hline \\[-0.8em]
$n_{\text{train}}$  &\multicolumn{1}{c}{IV-PILE-RF} &\multicolumn{1}{c}{OWL-RF} \\[-0.9em] \\ \\[-1em] \\[-0.8em]
$n = 1000$ &0.005 (0.000) & 0.015 (0.003)\\
$n = 2000$ &0.005 (0.000) & 0.015 (0.002)  \\ 
$n = 3000$ &0.005 (0.000) & 0.015 (0.002)  \\ \\[-0.8em]\hline 
\end{tabular}
\end{table}

\subsection*{F.4: Simulation Results of OWL and IV-PILE via Sample Splitting}
In Section \ref{sec: theory}, we studied a sample splitting version of the IV-PILE estimator. In this section, we report simulation results of this version of IV-PILE estimator. We consider the following five classifiers:
\begin{enumerate}
    \item SS-IV-PILE-RF: Sample-Splitting IV-PILE with relevant conditional probabilities in $L(\mathbf{X})$ and $U(\mathbf{X})$ estimated via random forests;
    \item SS-IV-PILE-LOGIT: Sample-Splitting IV-PILE with relevant conditional probabilities in $L(\mathbf{X})$ and $U(\mathbf{X})$ estimated via multinomial logistic regressions;
    \item OWL-RF: OWL with the propensity score estimated via random forests;
    \item OWL-LOGIT: OWL with the propensity score estimated via logistic regressions;
    \item COIN-FLIP: Classifier based on random coin flips.
\end{enumerate}
Note that the sample splitting version of IV-PILE is similar to IV-PILE except that we use half of the data ($0.5\times n_{\text{train}}$) to estimate the relevant conditional probability models in $L(\mathbf{X}), U(\mathbf{X})$, either via a random forest or a multinomial regression model, and then estimate the Balke-Pearl bound for the second half of the training data ($0.5\times n_{\text{train}}$) using the estimated $L(\mathbf{X})$ and $U(\mathbf{X})$ models and finally compute the IV-PILE estimator using this second half. We did not perform sample splitting for the OWL estimator as the theoretical derivation underpinning OWL algorithm does not require sample splitting. To conclude, we are using half of the training data to build the SS-IV-PILE-RF and SS-IV-PILE-LOGIT, and all of the training data to build OWL-RF and OWL-LOGIT. The simulation set-up is as follows: we let $g_1$ be Model (1) with $\xi = 0$, $g_2$ be Model (1), $\lambda = \delta = 0.5$, $1.0$, $1.5$, and $2.0$, and $n_{\text{train}} = 300$ or $500$. All random forest models are fit with the default settings, and all classifiers adopt a Gaussian RBF kernel. The test dataset is generated in the same way as described in Section \ref{subsubsec: train and test dataset}. Again we point out that $C_{\text{DGP}}$ in these cases are the same as before as $C_{\text{DGP}}$ depends only on the data-generating process and not the specific algorithms. 

Table \ref{tbl: simulation sample split} reports the simulation results. Compare entries under ``SS-IV-PILE-RF" and ``SS-IV-PILE-LOGIT" to the corresponding entries in Table \ref{tbl: simulation 1}, \ref{tbl: app n_train 300} and \ref{tbl: app n_train 500}, and we observed that the sample-splitting version of IV-PILE had slightly inferior performance compared to the non-sample-splitting version. However, the difference was minor and not consequential; the sample-splitting version of IV-PILE still largely outperformed OWL in simulation settings considered here.

\begin{table}[ht]
\centering
\caption{\small Estimated weighted misclassification error for different $(\lambda, \delta)$ combinations. We take $g_1(\mathbf{X}, U)$ to be Model (1) with $\xi = 0$ and $g_1(\mathbf{X}, U, A)$ to be Model (1) throughout. Training data sample size $n_{\text{train}} = 300$ or $500$. Each number in the cell is averaged over $500$ simulations. Standard errors are in parentheses.}
\label{tbl: simulation sample split}
\resizebox{\columnwidth}{!}{%
\fbox{%
\begin{tabular}{cccccc} \\[-0.8em]
   &\multicolumn{2}{c}{SS-IV-PILE-RF} &\multicolumn{2}{c}{OWL-RF} &COIN-FLIP\\[-0.9em] \\ \\[-1em] \\[-0.8em]
 & $n_{\text{train}} = 300$ & $n_{\text{train}} = 500$ & $n_{\text{train}} = 300$ & $n_{\text{train}} = 500$ \\[-0.9em] \\ \\[-1em] \cline{1-1} \\[-0.8em]
$(\lambda, \delta) = (0.5, 0.5)$  &0.008 (0.007) &0.007 (0.005) & 0.014 (0.004) & 0.015 (0.003) &0.030 (0.000) \\
$(\lambda, \delta) = (1.0, 1.0)$  &0.010 (0.007) &0.010 (0.005) & 0.017 (0.004) &0.018 (0.003) &0.033 (0.000)  \\ 
$(\lambda, \delta) = (1.5, 1.5)$  &0.016 (0.004) &0.015 (0.004) & 0.022 (0.003) & 0.023 (0.003) &0.037 (0.000)\\ 
$(\lambda, \delta) = (2.0, 2.0)$  &0.024 (0.008) &0.022 (0.004) & 0.029 (0.004) &0.029 (0.003) &0.043 (0.000)  \\ \\[-0.8em]\hline 
\\[-0.8em]
 &\multicolumn{2}{c}{SS-IV-PILE-LOGIT} &\multicolumn{2}{c}{OWL-LOGIT} &COIN-FLIP\\[-0.9em] \\ \\[-1em] \\[-0.8em]
 & $n_{\text{train}} = 300$ & $n_{\text{train}} = 500$ & $n_{\text{train}} = 300$ & $n_{\text{train}} = 500$ \\[-0.9em] \\ \\[-1em] \cline{1-1} \\[-0.8em]
$(\lambda, \delta) = (0.5, 0.5)$  &0.008 (0.005) &0.007 (0.004) &0.021 (0.002) &0.014 (0.003) &0.030 (0.000) \\
$(\lambda, \delta) = (1.0, 1.0)$  &0.011 (0.005) &0.010 (0.004) &0.023 (0.003) &0.024 (0.002) &0.033 (0.000) \\
$(\lambda, \delta) = (1.5, 1.5)$  &0.016 (0.004) &0.015 (0.004) &0.028 (0.003) &0.029 (0.002) &0.037 (0.000) \\
$(\lambda, \delta) = (2.0, 2.0)$  &0.023 (0.003) &0.022 (0.004) &0.034 (0.003) &0.035 (0.002) &0.043 (0.000) \\ \\[-0.8em]
\end{tabular}}}
\end{table}

\subsection*{F.5: Simulation Results of OWL and IV-PILE when the IV is Weak}
In this section, we consider simulation scenarios when the putative IV is valid but may be weak, in the sense that the association between the IV $Z$ and the treatment $A$ may be small. We consider the same data-generating process as in Section \ref{subsec: simulation setup}, except that $\alpha$, the association between $Z$ and $A$ in the following ``A-model" will be varied:
\begin{equation*}
    \begin{split}
        &P(A = 1 \mid \mathbf{X}, U, Z) = \text{expit}\{\alpha Z + X_1 - 7X_2 + \lambda (1 + X_1) U\}.
    \end{split}
\end{equation*}
Specifically, we let $\alpha = 2$, $4$, $8$ (as in the previous simulations), and $12$. The rest of the simulation set-up is as follows: we let $g_1$ be Model (1) with $\xi = 0$, $g_2$ be Model (1), $\lambda = \delta = 0.5$, $1.0$, $1.5$, and $2.0$, and $n_{\text{train}} = 500$. We consider the following five classifiers:
\begin{enumerate}
    \item IV-PILE-RF: IV-PILE with relevant conditional probabilities in $L(\mathbf{X})$ and $U(\mathbf{X})$ estimated via random forests;
    \item IV-PILE-LOGIT: IV-PILE with relevant conditional probabilities in $L(\mathbf{X})$ and $U(\mathbf{X})$ estimated via multinomial logistic regressions;
    \item OWL-RF: OWL with the propensity score estimated via random forests;
    \item OWL-LOGIT: OWL with the propensity score estimated via logistic regressions;
    \item COIN-FLIP: Classifier based on random coin flips.
\end{enumerate}
All random forest models are fit with the default settings, and all classifiers adopt a Gaussian RBF kernel. The test dataset is generated in the same way as described in Section \ref{subsubsec: train and test dataset}. In addition to the generalization error, we also report the average estimated compliance rate for each setting. Table \ref{tbl: simulation weak IV} summarizes the results. Again we point out that $C_{\text{DGP}}$ does not depend on the strength of the IV and remains the same as before in these cases.

When $\alpha = 2$ and the compliance is low $(\approx 0.14)$, we would expect an IV-optimal rule to have poor performance, as little information (without additional assumptions) can be learned about the CATE, the partial identification intervals are wide, and the gap between ``IV-optimality" and ``optimality" is large. Simulation results under $\alpha = 2$ corroborate this. When $\alpha$ grows larger and the estimated compliance becomes larger, e.g., $\alpha = 8$ and $\alpha = 12$, we would expect that the ``IV-optimal" rule began to have a favorable performance and largely outperformed the naive OWL-based methods. Again, this is verified by simulation results under $\alpha = 8$ and $\alpha = 12$. In empirical studies, an IV with estimated compliance around $0.1 \sim 0.2$ is often considered a weak IV; an IV with estimated compliance around $0.5$ is considered a relatively strong IV (\citealp{ertefaie2018quantitative}). In clinical trials, compliance can be even significantly higher than $0.5$. Table \ref{tbl: simulation weak IV} suggests that targeting an ``IV-optimal" ITR is a sensible thing to do when the compliance is a relatively high, and might not yield a favorable performance when the IV is not informative and compliance low. In order to obtain ``higher-quality" ITR, researchers are advised to leverage a stronger IV, or an IV such that additional IV identification assumptions that help narrow down partial identification intervals hold.

\begin{table}[ht]
\centering
\caption{\small Estimated weighted misclassification error for different $\alpha$ and $(\lambda, \delta)$. We take $g_1(\mathbf{X}, U)$ to be Model (1) with $\xi = 0$ and $g_2(\mathbf{X}, U, A)$ to be Model (1) throughout. Training data sample size $n_{\text{train}} = 500$. Each number in the cell is averaged over $500$ simulations. Standard errors are in parentheses.}
\label{tbl: simulation weak IV}
\resizebox{\columnwidth}{!}{%
\fbox{%
\begin{tabular}{ccccccc} \\[-0.8em]
   &IV-PILE-LOGIT &IV-PILE-RF &OWL-LOGIT &OWL-RF &COIN-FLIP & Compliance\\[-0.9em] \\ \\[-1em] \\[-0.8em]
 \multirow{2}{*}{\begin{tabular}{c}Weak\\ ($\alpha = 2$)\end{tabular}}  & \\ \\[-0.9em] \\ \\[-1em] \cline{1-1} \\[-0.8em]
$(\lambda, \delta) = (0.5, 0.5)$  & 0.034 (0.003) & 0.035 (0.003) & 0.033 (0.002) & 0.034 (0.001) & 0.030 (0.000) & 0.143 (0.045) \\
$(\lambda, \delta) = (1.0, 1.0)$  & 0.037 (0.003) & 0.038 (0.003) & 0.035 (0.002) & 0.036 (0.001) & 0.033 (0.000)  & 0.141 (0.047)\\ 
$(\lambda, \delta) = (1.5, 1.5)$  & 0.041 (0.003) & 0.042 (0.002) & 0.039 (0.002) & 0.040 (0.001) & 0.037 (0.000) & 0.143 (0.046)\\ 
$(\lambda, \delta) = (2.0, 2.0)$  & 0.046 (0.003) & 0.047 (0.003) & 0.044 (0.002) & 0.046 (0.001) & 0.043 (0.000) & 0.143 (0.045)  \\ \\[-0.8em]\hline 
\\[-0.8em]
  \multirow{2}{*}{\begin{tabular}{c}Moderately Weak\\ ($\alpha = 4$)\end{tabular}}  & \\ \\[-0.9em] \\ \\[-1em] \cline{1-1} \\[-0.8em]
$(\lambda, \delta) = (0.5, 0.5)$  & 0.027 (0.005) & 0.026 (0.004) & 0.027 (0.002) & 0.028 (0.002) & 0.030 (0.000) & 0.281 (0.041)\\
$(\lambda, \delta) = (1.0, 1.0)$  & 0.030 (0.004) & 0.028 (0.004) & 0.029 (0.002) & 0.031 (0.002) & 0.033 (0.000) & 0.278 (0.042)  \\
$(\lambda, \delta) = (1.5, 1.5)$ & 0.034 (0.004) & 0.033 (0.004) & 0.034 (0.002) & 0.035 (0.002) & 0.037 (0.000) & 0.281 (0.042) \\
$(\lambda, \delta) = (2.0, 2.0)$ & 0.040 (0.004) & 0.038 (0.004) & 0.039 (0.002) & 0.041 (0.002) & 0.043 (0.000) & 0.278 (0.041) \\ \\[-0.8em]\hline 
\\[-0.8em]
 \multirow{2}{*}{\begin{tabular}{c}Moderately Strong\\ ($\alpha = 8$)\end{tabular}}  & \\ \\[-0.9em] \\ \\[-1em] \cline{1-1} \\[-0.8em]
$(\lambda, \delta) = (0.5, 0.5)$ & 0.005 (0.003) & 0.005 (0.000) & 0.021 (0.002) & 0.015 (0.003) & 0.030 (0.000) & 0.473 (0.033) \\
$(\lambda, \delta) = (1.0, 1.0)$  & 0.009 (0.002) & 0.008 (0.000) & 0.024 (0.002) & 0.018 (0.003) & 0.033 (0.000) & 0.473 (0.035) \\
$(\lambda, \delta) = (1.5, 1.5)$  & 0.014 (0.003) & 0.014 (0.000) & 0.029 (0.002) & 0.023 (0.003) & 0.037 (0.000) & 0.474 (0.032) \\
$(\lambda, \delta) = (2.0, 2.0)$  & 0.021 (0.003) & 0.020 (0.000) & 0.035 (0.002) & 0.029 (0.003) & 0.043 (0.000) & 0.465 (0.032) \\ \\[-0.8em]\hline
\\[-0.8em]
\multirow{2}{*}{\begin{tabular}{c}Moderately Strong\\ ($\alpha = 12$)\end{tabular}}  & \\ \\[-0.9em] \\ \\[-1em] \cline{1-1} \\[-0.8em]
$(\lambda, \delta) = (0.5, 0.5)$  & 0.005 (0.002) & 0.005 (0.000) & 0.021 (0.002) & 0.012 (0.003) & 0.030 (0.000) & 0.499 (0.033) \\
$(\lambda, \delta) = (1.0, 1.0)$  & 0.008 (0.002) & 0.008 (0.000) & 0.024 (0.002) & 0.015 (0.003) & 0.033 (0.000) & 0.499 (0.032) \\
$(\lambda, \delta) = (1.5, 1.5)$  & 0.013 (0.002) & 0.014 (0.000) & 0.029 (0.002) & 0.020 (0.003) & 0.037 (0.000) & 0.499 (0.031) \\
$(\lambda, \delta) = (2.0, 2.0)$  & 0.020 (0.002) & 0.020 (0.000) & 0.035 (0.002) & 0.026 (0.003) & 0.043 (0.000) & 0.500 (0.032) \\ \\[-0.8em]
\end{tabular}}}
\end{table}

\subsection*{F.6: Simulation Results of OWL and IV-PILE when the IV is Invalid}
We investigate the performance of IV-PILE when the putative IV is invalid in this section. Specifically, we consider a situation where the exclusion restriction assumption is violated so that the IV $Z$ has a direct effect on the outcome $Y$. We consider the following data-generating process which is slightly modified from that in Section \ref{subsubsec: data generating process}:
\begin{equation*}
    \begin{split}
        &Z \sim \text{Bern}(0.5), ~ X_1,\cdots,X_{10} \sim \text{Unif}~[-1, 1], ~ U \sim \text{Unif}~[-1, 1],\\
        &P(A = 1 \mid \mathbf{X}, U, Z) = \text{expit}\{8Z + X_1 - 7X_2 + \lambda (1 + X_1) U\},\\
        &P(Y = 1 \mid A, \mathbf{X}, U) = \text{expit}\{1 - X_1 + X_2 + 0.442(1 - X_1 + X_2 + \delta U)A + cZ\}.
    \end{split}
\end{equation*}
In the above data-generating process, we allow $Z$ to have a direct effect on $Y$ by adding a ``cZ" term in the ``Y-model". Similar to the previous simulations, we consider the following three classifiers: IV-PILE-RF, OWL-RF, and COIN-FLIP, and tabulate the results for each classifier against different choices of $(c, \lambda, \delta)$. We let $n_\text{train} = 500$. We calculate $C_{\text{DGP}}$ for the new data-generating processes. Table \ref{tbl: simulation invalid IV} summarizes the results for $c = -1$, $1$, and $2$, representing mild violations of the exclusion restriction assumption. We observed that the IV-PILE algorithm seemed to be relatively robust to mild violations of the exclusion restriction assumption.

\begin{table}[ht]
\centering
\caption{\small Estimated weighted misclassification error for different $c$ and $(\lambda, \delta)$. We take $g_1(\mathbf{X}, U)$ to be Model (1) with $\xi = 0$ and $g_2(\mathbf{X}, U, A)$ to be Model (1) throughout. Training data sample size $n_{\text{train}} = 500$. Each number in the cell is averaged over $500$ simulations. Standard errors are in parentheses.}
\label{tbl: simulation invalid IV}
\fbox{
\begin{tabular}{cccccccc} \\[-0.8em]
   &IV-PILE-RF &OWL-RF &COIN-FLIP &$C_{\text{DGP}}$\\[-0.9em] \\ \\[-1em] \\[-0.8em]
  
  $c = -1$  & \\[-0.9em] \\ \\[-1em] \cline{1-1} \\[-0.8em]
$(\lambda, \delta) = (0.5, 0.5)$  & 0.012 (0.008) & 0.023 (0.005) & 0.038 (0.000) & 0.001\\
$(\lambda, \delta) = (1.0, 1.0)$  & 0.014 (0.008) & 0.026 (0.005) & 0.040 (0.000) & 0.004\\ 
$(\lambda, \delta) = (1.5, 1.5)$  & 0.018 (0.007) & 0.030 (0.004) & 0.045 (0.000) & 0.009\\ 
$(\lambda, \delta) = (2.0, 2.0)$  & 0.024 (0.007) & 0.035 (0.004) & 0.050 (0.000) & 0.014\\ \\[-0.8em]\hline 
\\[-0.8em]
 
 $c = 1$  & \\[-0.9em] \\ \\[-1em] \cline{1-1} \\[-0.8em]
$(\lambda, \delta) = (0.5, 0.5)$  & 0.005 (0.000) & 0.011 (0.002) & 0.023 (0.000) & 0.001\\
$(\lambda, \delta) = (1.0, 1.0)$  & 0.008 (0.000) & 0.014 (0.002) & 0.026 (0.000) & 0.004\\ 
$(\lambda, \delta) = (1.5, 1.5)$  & 0.013 (0.000) & 0.018 (0.002) & 0.029 (0.000) & 0.008\\ 
$(\lambda, \delta) = (2.0, 2.0)$  & 0.019 (0.000) & 0.024 (0.002) & 0.034 (0.000) & 0.014\\ \\[-0.8em]\hline 
\\[-0.8em]
  $c = 2$  & \\[-0.9em] \\ \\[-1em] \cline{1-1} \\[-0.8em]
$(\lambda, \delta) = (0.5, 0.5)$  & 0.004 (0.000) & 0.009 (0.002) & 0.019 (0.000) & 0.001\\
$(\lambda, \delta) = (1.0, 1.0)$  & 0.007 (0.000) & 0.011 (0.002) & 0.021 (0.000) & 0.004\\ 
$(\lambda, \delta) = (1.5, 1.5)$  & 0.011 (0.000) & 0.015 (0.002) & 0.026 (0.000) & 0.008\\ 
$(\lambda, \delta) = (2.0, 2.0)$  & 0.016 (0.000) & 0.020 (0.002) & 0.028 (0.000) & 0.012\\ \\[-0.8em]
\\[-0.8em]
\end{tabular}}
\end{table}

\subsection*{F.7: Simulation Results of OWL and IV-PILE when the Outcome is Continuous but Bounded}
We consider the following data-generating process with a binary IV $Z$, a binary treatment $A$, a continuous but bounded outcome $Y$, a $10$-dimensional observed covariates $\mathbf{X}$, and an unmeasured confounder $U$:
\begin{equation*}
    \begin{split}
        &Z \sim \text{Bern}(0.5), ~ X_1,\cdots,X_{10} \sim \text{Unif}~[-1, 1], ~ U \sim \text{Unif}~[-1, 1],\\
        &P(A = 1 \mid \mathbf{X}, U, Z) = \text{expit}\{8Z + X_1 - 7X_2 + \lambda (1 + X_1) U\},\\
        &Y \mid A, \mathbf{X}, U \sim \text{Truncated Normal}(\mu(\mathbf{X}, U, A), \sigma = 1, a = -3, b = 4),
    \end{split}
\end{equation*}
where 
\[
\mu(\mathbf{X}, U, A) = g_1(\mathbf{X}, U) + 3g_2(\mathbf{X}, U, A)
\]
with the same choices of $g_1(\mathbf{X}, U)$:
\begin{equation*}
    \begin{split}
    &\text{Model}~(1): \qquad g_1(\mathbf{X}, U) = 1 - X_1 + X_2 + \xi U,
    \end{split}
\end{equation*}
and $g_2(\mathbf{X}, U, A)$ as before:
\begin{equation*}
    \begin{split}
    &\text{Model}~(1): \qquad g_2(\mathbf{X}, U, A) = 0.442(1 - X_1 + X_2 + \delta U)A, \\
    &\text{Model}~(2): \qquad g_2(\mathbf{X}, U, A) = (X_2 - 0.25X_1^2 - 1 + \delta U)A.
    \end{split}
\end{equation*}
In a truncated normal distribution, $\mu(\mathbf{X}, U, A)$ and $\sigma$ specify the mean and standard deviation of the ``parent" normal distribution, and $[a, b]$ specifies the truncation interval. Here, we have adopted a truncated normal distribution to model a continuous but bounded outcome of interest. The rest of the simulation setup is the same as described in Section \ref{subsec: simulation setup}, except that we use the Manski-Pepper bound described in \eqref{eqn: Mansk-Pepper lower bound} and \eqref{eqn: Mansk-Pepper upper bound} (see Supplementary Material A) to calculate the partial identification interval of the CATE. We calculate $C_{\text{DGP}}$ for each data-generating process.

We consider the following three classifiers:
\begin{enumerate}
    \item IV-PILE-RF: IV-PILE with relevant conditional probabilities/expectations in $L(\mathbf{X})$ and $U(\mathbf{X})$ estimated via random forests;
    \item OWL-RF: OWL with the propensity score estimated via a random forest;
    \item COIN-FLIP: Classifier based on random coin flips.
\end{enumerate}
Table \ref{tbl: simulation cont} summarizes simulation results. Again, we observed qualitatively similar behaviors as in the binary outcome simulations.

\begin{table}[ht]
\centering
\caption{\small Simulation results when $Y$ is continuous but bounded. Estimated \emph{weighted misclassification error} for different $(\lambda, \delta)$, $\xi$, and $g_2(\mathbf{X}, U, A)$ combinations. Training data sample size $n_{\text{train}} = 1000$. Each number in the cell is averaged over $500$ simulations. Standard errors are in parentheses.}
\label{tbl: simulation cont}
\fbox{%
\begin{tabular}{ccccccc} \\[-0.8em]
   &\multicolumn{1}{c}{IV-PILE-RF} &\multicolumn{1}{c}{OWL-RF} &COIN-FLIP &$C_{\text{DGP}}$\\[-0.9em] \\ \\[-1em] \\[-0.8em]
 $\xi = 0$, $g_2 = \text{Model}~(1)$&  \\[-0.9em] \\ \\[-1em] \cline{1-1} \\[-0.8em]
$(\lambda, \delta) = (0.5, 0.5)$  & 0.060 (0.029) & 0.261 (0.048) & 0.732 (0.001) & 0.003\\
$(\lambda, \delta) = (1.0, 1.0)$  & 0.106 (0.034) & 0.307 (0.047) & 0.773 (0.001) & 0.013\\ 
$(\lambda, \delta) = (1.5, 1.5)$  & 0.184 (0.039) & 0.386 (0.047) & 0.842 (0.001) & 0.030\\ 
$(\lambda, \delta) = (2.0, 2.0)$  & 0.292 (0.054) & 0.492 (0.051) & 0.936 (0.001) & 0.049\\  \\[-0.8em]\hline 
\\[-0.8em]
$\xi = 0$, $g_2 = \text{Model}~(2)$   \\[-0.9em] \\ \cline{1-1} \\[-0.9em]
$(\lambda, \delta) = (0.5, 0.5)$  & 0.020 (0.001) & 0.028 (0.006) & 1.645 (0.002) & 0.003\\
$(\lambda, \delta) = (1.0, 1.0)$  & 0.098 (0.001) & 0.115 (0.011) & 1.723 (0.002) & 0.008\\
$(\lambda, \delta) = (1.5, 1.5)$  & 0.239 (0.001) & 0.280 (0.020) & 1.864 (0.003) & 0.019\\
$(\lambda, \delta) = (2.0, 2.0)$  & 0.443 (0.003) & 0.531 (0.035) & 2.067 (0.003) & 0.032\\ \\[-0.8em]\hline 
\\[-0.8em]
$\xi = 1$, $g_2 = \text{Model}~(1)$   \\[-0.9em] \\ \cline{1-1} \\[-0.9em]
$(\lambda, \delta) = (0.5, 0.5)$  & 0.081 (0.043) & 0.265 (0.047) & 0.732 (0.001) & 0.002\\
$(\lambda, \delta) = (1.0, 1.0)$  & 0.137 (0.053) & 0.312 (0.049) & 0.774 (0.001) & 0.012\\
$(\lambda, \delta) = (1.5, 1.5)$  & 0.221 (0.066) & 0.395 (0.051) & 0.842 (0.001) & 0.025\\
$(\lambda, \delta) = (2.0, 2.0)$  & 0.342 (0.080) & 0.507 (0.050) & 0.936 (0.001) & 0.042\\ \\[-0.8em]\hline 
\\[-0.8em]
$\xi = 1$, $g_2 = \text{Model}~(2)$   \\[-0.9em] \\ \cline{1-1} \\[-0.9em]
$(\lambda, \delta) = (0.5, 0.5)$  & 0.020 (0.000) & 0.036 (0.011) & 1.645 (0.002) & 0.002\\
$(\lambda, \delta) = (1.0, 1.0)$  & 0.098 (0.001) & 0.131 (0.017) & 1.723 (0.002) & 0.005\\
$(\lambda, \delta) = (1.5, 1.5)$ & 0.239 (0.001) & 0.314 (0.032) & 1.864 (0.002) & 0.010\\
$(\lambda, \delta) = (2.0, 2.0)$ & 0.446 (0.008) & 0.601 (0.053) & 2.067 (0.003) & 0.020 \\
\end{tabular}}
\end{table}

\clearpage
\begin{center}
{\large\bf Supplementary Material G: Plug-in Estimators of IV-Optimal Rules and Theoretical Properties}
\end{center}

\subsection*{G.1: Plug-In Estimators for IV-Optimal Rules}

One approach to optimal ITR estimation problems in non-IV settings is to specify various aspects of the conditional distribution of the outcome given covariates and treatment assignment, i.e., $C(\mathbf{X})$, and take the sign of $\widehat{C}(\mathbf{X})$ to be the optimal treatment rule. These approaches are often called \emph{indirect approaches} in the literature as they do not directly target estimating ITRs. Methodologies that fall into this line include g-estimation methods in structural nested models (\citealp{robins1989analysis}; \citealp{murphy2003optimal}; \citealp{robins2004optimal}) and Q- and A-learning (\citealp{zhao2009reinforcement}; \citealp{qian2011performance}; \citealp{moodie2012q}). We can easily adapt the idea to derive a simple plug-in estimator based on $\widehat{L}(\mathbf{X})$ and $\widehat{U}(\mathbf{X})$, estimators of $L(\mathbf{X})$ and $U(\mathbf{X})$, respectively. Proposition \ref{thm: simple plug-in rules} establishes such an estimator $\widehat{f}_{\text{plug-in}}$ for IV-optimal rules.

\begin{customprop}{G1}\rm
\label{thm: simple plug-in rules}
Let $\widehat{L}$ and $\widehat{U}$ be estimators of $L$ and $U$. The plug-in estimator
\[
\widehat{f}_{\text{plug-in}}(\mathbf{x})=\text{sgn}\big\{\widehat{U}(\mathbf{x})^{+} - \{-\widehat{L}(\mathbf{x})\}^{+}\big\}
\]
minimizes $\mathcal{R}_{\text{upper}}(f;\widehat{L}(\cdot),\widehat{U}(\cdot))$, where  $u^{+} =u \vee 0$ for any $u\in \mathbb{R}$.
\end{customprop}
\par

Theorem \ref{thm: simple plug-in rules are optimal} and  Theorem \ref{thm: simple plug-in rules are optimal, lower bound} further establish the rate of convergence of $\widehat{f}_{\text{plug-in}}$ and prove that it is rate optimal.

\begin{customthm}{G1}\rm
\label{thm: simple plug-in rules are optimal}
\begin{equation}
      \mathcal{R}_{\text{upper}}(\widehat{f}_{\text{plug-in}})-\mathcal{R}^\ast_{\text{upper}}\leq O\left(n^{-(\alpha \wedge \beta)}\right).  
\end{equation}
\end{customthm}
\par

\begin{customthm}{G2}\rm
\label{thm: simple plug-in rules are optimal, lower bound}
Let $\theta$ denote a general parameter determining the joint distribution of $(\textbf{X},Z,A,Y)$ and $\Theta$ be the set of all $\theta$ that satisfy IV assumptions $IV.A1-IV.A4$. We observe i.i.d samples of $\textbf{X},Z,A,Y$. Let $U$ and $L$ be defined by Balke-Pearl Bound in (\ref{eqn: bp lower bound}) and (\ref{eqn: bp upper bound}). Suppose that we first obtain function estimates $\widehat{U}$ and $\widehat{L}$ that satisfy Assumption (\ref{assump: rate of convergence of L and U}) and then an estimator $\widehat{f}$ that only depends on $\{(\textbf{X}_{i}, Z_i),~ i=1,\dots,n \}$, and $\widehat{U}$ and $\widehat{L}$. Let 
\begin{equation*}
    \mathcal{F}_{L,n}=\left\{\widehat{L}:\mathbb{E}\left[|\widehat{L}(\textbf{X})-L(\textbf{X})|\right]\leq 2n^{-\alpha}\right\},
\end{equation*}
\begin{equation*}
    \mathcal{F}_{U,n} = \left\{\widehat{U}:\mathbb{E}\left[|\widehat{U}(\textbf{X})-U(\textbf{X})|\right]\leq 2n^{-\beta}\right\}.
\end{equation*}
Then
\begin{equation}
   \sup_{\theta\in \Theta, \widehat{L}\in \mathcal{F}_{L,n}, \widehat{U}\in \mathcal{F}_{U,n}}\mathcal{R}_{\text{upper}}(\widehat{f})-\mathcal{R}^\ast_{\text{upper}}\geq  n^{-(\alpha \wedge \beta)}.
\end{equation}
\end{customthm}

\begin{customremark}{G1}\rm
We consider supreme excess risk over all function estimates $\widehat{L}\in \mathcal{F}_{L,n}$ and $\widehat{U}\in \mathcal{F}_{U,n}$ because we treat $\widehat{L}$ and $\widehat{U}$ as general estimates in the article and do not assume any particular structure/properties other than Assumption (\ref{assump: rate of convergence of L and U}). Also, we consider $\widehat{f}$ that depends on $\{(Y_i,Z_i),i=1,\dots n\}$ only through $\widehat{U}$ and $\widehat{L}$. If we allow $\widehat{f}$ to directly depend on $Y_{i}$ and $Z_i$ then one can construct another set of $\widehat{L}$ and $\widehat{U}$ that may converge faster using $(\textbf{X}_i,Z_i,A_i,Y_i)$. To avoid such case we make the restriction. Note both the plug-in estimator and the SVM-based estimator we propose obey this restriction.
\end{customremark}\rm

Although plug-in estimators for learning IV-optimal rules are simple and rate optimal, they may some undesirable features as elaborated in Remark \ref{remark: we also develop plug-in estimators}. In essence, such estimators do not allow empirical researchers to estimate $L(\mathbf{X})$ and $U(\mathbf{X})$ using some flexible machine learning tools, while maintain the simplicity of learned ITRs, as $\widehat{f}_{\text{plug-in}}$ follows immediately from $\widehat{L}(\mathbf{X})$ and $\widehat{U}(\mathbf{X})$.

\subsection*{G.2: Proofs of Proposition \ref{thm: simple plug-in rules}, Theorem \ref{thm: simple plug-in rules are optimal}, and Theorem \ref{thm: simple plug-in rules are optimal, lower bound}}

\subsubsection*{G.2.1: Proof of Proposition \ref{thm: simple plug-in rules}}

Observe that for any $f$:
\begin{equation*}
    \begin{split}
     &\sup_{C'(\mathbf{x}) \in [\widehat{L}(\mathbf{x}), \widehat{U}(\mathbf{x})]}  |C'(\mathbf{x})|\cdot \mathbbm{1}\big\{\text{sgn}\{f(\mathbf{x})\} \neq \text{sgn}\{C'(\mathbf{x})\}\big\}\\
    &\geq \mathbbm{1}\big\{L(\mathbf{x})<0<U(\mathbf{x})\big\}\cdot\min(|L(\mathbf{x})|,|U(\mathbf{x})|)\\
    &=\sup_{C'(\mathbf{x}) \in [\widehat{L}(\mathbf{x}), \widehat{U}(\mathbf{x})]} |C'(\mathbf{x})|\cdot \mathbbm{1}\big\{\text{sgn}\{\widehat{U}(\mathbf{x})^{+} > \{-\widehat{L}(\mathbf{x})\}^{+}\} \neq \text{sgn}\{C'(\mathbf{x})\}\big\}
   \\
    &=\sup_{C'(\mathbf{x}) \in [\widehat{L}(\mathbf{x}), \widehat{U}(\mathbf{x})]} |C'(\mathbf{x})|\cdot \mathbbm{1}\big\{\text{sgn}\{\widehat{f}_{\text{plug-in}}(\mathbf{x})\} \neq \text{sgn}\{C'(\mathbf{x})\}\big\}.
     \end{split}
\end{equation*}
Therefore we have
\begin{equation*}
    \begin{split}
        &\mathcal{R}_{\text{upper}}(f;\widehat{L}(\cdot),\widehat{U}(\cdot))\\
        = &\mathbb{E}\left[\sup_{C'(\mathbf{X}) \in [\widehat{L}(\mathbf{X}), \widehat{U}(\mathbf{X})]}  |C'(\mathbf{X})|\cdot \mathbbm{1}\{\text{sgn}\{f(\mathbf{X})\} \neq \text{sgn}\{C'(\mathbf{X})\}\}\right]\\
    \geq & \mathbb{E}\left[\sup_{C'(\mathbf{X}) \in [\widehat{L}(\mathbf{X}), \widehat{U}(\mathbf{X})]}  |C'(\mathbf{X})|\cdot \mathbbm{1}\{\text{sgn}\{\widehat{f}_{\text{plug-in}}(\mathbf{X})\} \neq \text{sgn}\{C'(\mathbf{X})\}\}\right]\\
    = &\mathcal{R}_{\text{upper}}(\widehat{f}_{\text{plug-in}};\widehat{L}(\cdot),\widehat{U}(\cdot))
    \end{split}
\end{equation*}

\subsubsection*{G.2.2: Proof of Theorem \ref{thm: simple plug-in rules are optimal}}

According to Proposition \ref{prop: excess risk}
\begin{equation}
  \mathcal{R}_{\text{upper}}(\widehat{f}_{\text{plug-in}})-\mathcal{R}^\ast_{\text{upper}}  = \mathbb{E}\left[\mathbbm{1}\big\{\text{sgn}\{\widehat{f}_{\text{plug-in}}(\textbf{X})\}\neq \text{sgn}\{f^\ast(\textbf{X})\}\big\}\cdot|\eta(\textbf{X})|\right]
  \label{eqn: plug in est excess risk}
\end{equation}
From the Proposition \ref{prop: Bayes decision rule} we know that $\text{sgn}\{f^\ast(\mathbf{x})\} = \text{sgn}\{\eta(\mathbf{x})\}$. Define
\begin{equation*}
    \widehat{\eta}(\mathbf{x})=|\widehat{U}(\mathbf{x})|\cdot\mathbbm{1}\{\widehat{L}(\mathbf{x}) > 0\} - |\widehat{L}(\mathbf{x})|\cdot\mathbbm{1}\{\widehat{U}(\mathbf{x}) < 0\} + (|\widehat{U}(\mathbf{x})| - |\widehat{L}(\mathbf{x})|)\cdot\mathbbm{1}\{[\widehat{L}(\mathbf{x}), \widehat{U}(\mathbf{x})] \ni 0\}.
\end{equation*}
Then we have 
\begin{equation*}
    \text{sgn}\{\widehat{f}_{\text{plug-in}}(\mathbf{x})\} = \text{sgn}\{\widehat{\eta}(\mathbf{x})\}.
\end{equation*}
Plug into (\ref{eqn: plug in est excess risk}), and we have 
\begin{equation*}
   \mathcal{R}_{\text{upper}}(\widehat{f}_{\text{plug-in}})-\mathcal{R}^\ast_{\text{upper}}  =\mathbb{E}\left[\mathbbm{1}\big\{\text{sgn}\{\widehat{\eta}(\textbf{X})\}\neq \text{sgn}\{\eta(\textbf{X})\}\big\}\cdot|\eta(\textbf{X})|\right]
\end{equation*}
For any $\mathbf{x}$, if $\text{sgn}\{\widehat{\eta}(\mathbf{x}))\neq \text{sgn}(\eta(\mathbf{x})\}$, we then have $|\widehat{\eta}(\mathbf{x})-\eta(\mathbf{x})|\geq  |\eta(\mathbf{x})|$, which implies
\begin{equation*}
    \mathbbm{1}\big\{\text{sgn}\{\widehat{\eta}(\textbf{x})\}\neq \text{sgn}\{\eta(\textbf{x})\}\big\}\cdot|\eta(\textbf{x})|\leq |\widehat{\eta}(\mathbf{x})-\eta(\mathbf{x})|.
    \end{equation*}
Therefore
\begin{equation*}
     \mathcal{R}_{\text{upper}}(\widehat{f}_{\text{plug-in}})-\mathcal{R}^\ast_{\text{upper}} \leq \mathbb{E}\left[|\widehat{\eta}(\textbf{X})-\eta(\textbf{X})|\right].
\end{equation*}
Moreover, it can be easily verified that $\eta(\mathbf{x})=-w(\mathbf{x})e(\mathbf{x})$ and $\widehat{\eta}(\mathbf{x})=-\widehat{w}(\mathbf{x})\widehat{e}(\mathbf{x})$
, which implies
\begin{equation*}
     \mathcal{R}_{\text{upper}}(\widehat{f}_{\text{plug-in}})-\mathcal{R}^\ast_{\text{upper}} \leq \mathbb{E}\left[|\widehat{w}(\textbf{X})\widehat{e}(\textbf{X})-w(\textbf{X})e(\textbf{X})|\right].
\end{equation*}
Apply Lemma \ref{lemma: what and w} and we have
\begin{equation*}
\mathcal{R}_{\text{upper}}(\widehat{f}_{\text{plug-in}})-\mathcal{R}^\ast_{\text{upper}} \leq O\left(n^{-(\alpha \wedge \beta)}\right).
\end{equation*}

\subsubsection*{G.2.3: Proof of Theorem \ref{thm: simple plug-in rules are optimal, lower bound}}

We will first prove 
  \begin{equation}
      \sup_{\theta\in \Theta, \widehat{L}\in \mathcal{F}_{L,n}, \widehat{U}\in \mathcal{F}_{U,n}}\mathcal{R}_{\text{upper}}(\widehat{f})-\mathcal{R}^\ast_{\text{upper}}\geq  n^{-\beta}, \label{eqn: first prove}
  \end{equation} 
then by symmetry it's easy to prove that 
$ \sup_{\theta\in \Theta, \widehat{L}\in \mathcal{F}_{L,n}, \widehat{U}\in \mathcal{F}_{U,n}}\mathcal{R}_{\text{upper}}(\widehat{f})-\mathcal{R}^\ast_{\text{upper}}\geq  n^{-\alpha}$.

To prove (\ref{eqn: first prove}), we will construct two situations where they share the same joint distribution of $(\textbf{X},Z,A,Y)$ and the same $\widehat{U}$ and $\widehat{L}$, but have different optimal rules. 
Let $\textbf{X}$ be an independent uniform random variable on $[0,1]$ and $Z$ an independent Bernoulli($1/2$) random variable. We construct $p_{y, a \mid z, \mathbf{X}}$ as in Table \ref{tbl: joint distribution 1} for two scenarios, where $\epsilon_n = 2n^{-\beta}$. 

\begin{table}[h]
\centering
\caption{}
\label{tbl: joint distribution 1}
\fbox{
\begin{tabular}{lcc}\\[-0.8em]
$n_{\text{train}}$  &\multicolumn{1}{c}{Scenario 1} &\multicolumn{1}{c}{Scenario 2} \\[-0.9em] \\ \\[-1em] \\[-0.8em]
$p_{1,1\mid 1,\textbf{X}}$ &$\frac{1}{2}\epsilon_{n}+\frac{1}{4}$ & $\frac{1}{2}\epsilon_{n}+\frac{1}{4}$\\
$p_{-1,1\mid 1,\textbf{X}}$ &$\epsilon_{n}+\frac{1}{4}$ & $\frac{1}{4}$  \\ 
$p_{1,-1\mid 1,\textbf{X}}$ & $-\epsilon_{n}+\frac{1}{4}$ & $\frac{1}{4}$\\
$p_{-1,-1\mid 1,\textbf{X}}$ & $-\frac{1}{2}\epsilon_{n}+\frac{1}{4}$ &$-\frac{1}{2}\epsilon_{n}+\frac{1}{4}$\\
$p_{1,1\mid -1,\textbf{X}}$ &$-\frac{1}{2}\epsilon_{n}+\frac{1}{4}$ &$-\frac{1}{2}\epsilon_{n}+\frac{1}{4}$\\
$p_{1,-1\mid -1,\textbf{X}}$ &$-\epsilon_{n}+\frac{1}{4}$ & $\frac{1}{4}$\\
$p_{1,-1\mid -1,\textbf{X}}$ &$\epsilon_{n}+\frac{1}{4}$ &$\frac{1}{4}$\\
$p_{-1,-1\mid -1,\textbf{X}}$ & $\frac{1}{2}\epsilon_{n}+\frac{1}{4}$ & $\frac{1}{2}\epsilon_{n}+\frac{1}{4}$\\ 
\end{tabular}}
\end{table}

It's easy to verify that both scenarios yield well-defined joint distributions of $(\textbf{X},Z,A,Y)$ that satisfy IV assumptions IV.A1-IV.A2. According to the Balke-Pearl Bound, in Scenario 1, $L^{(1)}(x)=\epsilon_n - \frac{1}{2}$ and $U^{(1)}(x)= \frac{1}{2}-2\epsilon_n$ for all $x\in [0,1]$. In Scenario 2, $L^{(2)}(x)=\epsilon_n - \frac{1}{2}$ and $U^{(2)}(x)=\frac{1}{2}$ for all $x\in [0,1]$. Let $\widehat{L}(x)=\epsilon_n -\frac{1}{2}$ and $\widehat{U}(x)=\frac{1}{2}-\epsilon_n$, and it holds that $\widehat{L}\in \mathcal{F}_{L,n}$ and $\widehat{U}\in \mathcal{F}_{U,n}$ for both scenarios.

Therefore in both scenarios we observe the same $\{(\textbf{X}_i , Z_i,A_i,Y_i),i=1,\dots,n\}$, and $\widehat{U}$ and $\widehat{L}$. This implies that we would have the same $\widehat{f}$. However, the optimal rules in these two scenarios are precisely the opposite: in Scenario 1, the optimal rule should be always negative while in Scenario 2 it should be always positive. Consider $\eta^{(1)}(\mathbf{x})$ and $\eta^{(2)}(\mathbf{x})$ as defined in proposition \ref{prop: Bayes decision rule} for both scenarios. Then
\begin{equation*}
    \eta^{(2)}(\mathbf{x})=\epsilon_n =-\eta^{(1)}(\mathbf{x}),
\end{equation*}
which implies that
\begin{equation}
    \mathbbm{1}\big\{\text{sgn}\{\widehat{f}(\mathbf{x})\}\neq \text{sgn}\{\eta^{(1)}(\mathbf{x})\}\big\}\cdot|\eta^{(1)}(\mathbf{x})|+ \mathbbm{1}\big\{\text{sgn}\{\widehat{f}(\mathbf{x})\}\neq \text{sgn}\{\eta^{(2)}(\mathbf{x})\}\big\}\cdot|\eta^{(2)}(\mathbf{x})|=\epsilon_n
\end{equation}
Let $\mathcal{R}^{\ast,(1)}_{\text{upper}}$ and $\mathcal{R}^{\ast,(2)}_{\text{upper}}$ denote the optimal risk for Scenario 1 and 2. Then we have
\begin{equation*}
    \begin{split}
        &\mathcal{R}_{\text{upper}}(\widehat{f};L^{(1)}(\cdot),U^{(1)}(\cdot))- \mathcal{R}^{\ast,(1)}_{\text{upper}}+\mathcal{R}_{\text{upper}}(\widehat{f};L^{(2)}(\cdot),U^{(2)}(\cdot))- \mathcal{R}^{\ast,(2)}_{\text{upper}}\\
        =& \mathbb{E}\left[\mathbbm{1}\big\{\text{sgn}\{\widehat{f}(\mathbf{X})\}\neq \text{sgn}\{\eta^{(1)}(\mathbf{X})\}\big\}\cdot|\eta^{(1)}(\mathbf{X})|\right]+ \mathbb{E}\left[\mathbbm{1}\big\{\text{sgn}\{\widehat{f}(\mathbf{X})\}\neq \text{sgn}\{\eta^{(2)}(\mathbf{X})\}\big\}\cdot|\eta^{(2)}(\mathbf{X})|\right]\\
        =&\epsilon_n.
    \end{split}
\end{equation*}
Finally, we have
\begin{equation*}
    \begin{split}
        &\sup_{\theta\in \Theta, \widehat{L}\in \mathcal{F}_{L,n}, \widehat{U}\in \mathcal{F}_{U,n}}\mathcal{R}_{\text{upper}}(\widehat{f})-\mathcal{R}^\ast_{\text{upper}}\\
        \geq &\frac{1}{2}\left[\mathcal{R}_{\text{upper}}(\widehat{f};L^{(1)}(\cdot),U^{(1)}(\cdot))- \mathcal{R}^{\ast,(1)}_{\text{upper}}+\mathcal{R}_{\text{upper}}(\widehat{f};L^{(2)}(\cdot),U^{(2)}(\cdot))- \mathcal{R}^{\ast,(2)}_{\text{upper}}\right]\\
        =&\frac{1}{2}\epsilon_n = n^{-\beta}
    \end{split}
\end{equation*}

Analogously, we can prove $\sup_{\theta\in \Theta, \widehat{L}\in \mathcal{F}_{L,n}, \widehat{U}\in \mathcal{F}_{U,n}}\mathcal{R}_{\text{upper}}(\widehat{f})-\mathcal{R}^\ast_{\text{upper}}\geq n^{-\alpha}$ by considering the conditional distributions in two scenarios as in Table \ref{tbl: joint distribution 2}:

\begin{table}[h]
\centering
\caption{}
\label{tbl: joint distribution 2}
\fbox{
\begin{tabular}{lcc}\\[-0.8em]
$n_{\text{train}}$  &\multicolumn{1}{c}{Scenario 1} &\multicolumn{1}{c}{Scenario 2} \\[-0.9em] \\ \\[-1em] \\[-0.8em]
$p_{1,1\mid 1,\textbf{X}}$ &$\frac{1}{4}$ & $\epsilon_{n}+\frac{1}{4}$\\
$p_{-1,1\mid 1,\textbf{X}}$ &$\frac{1}{2}\epsilon_{n}+\frac{1}{4}$ & $\frac{1}{2}\epsilon_{n}+\frac{1}{4}$ \\ 
$p_{1,-1\mid 1,\textbf{X}}$ &$-\frac{1}{2}\epsilon_{n}+\frac{1}{4}$ & $-\frac{1}{2}\epsilon_{n}+\frac{1}{4}$\\
$p_{-1,-1\mid 1,\textbf{X}}$ & $\frac{1}{4}$ &$-\epsilon_{n}+\frac{1}{4}$\\
$p_{1,1\mid -1,\textbf{X}}$ & $\frac{1}{4}$ &$-\epsilon_{n}+\frac{1}{4}$\\
$p_{1,-1\mid -1,\textbf{X}}$ &$-\frac{1}{2}\epsilon_{n}+\frac{1}{4}$ & $-\frac{1}{2}\epsilon_{n}+\frac{1}{4}$\\
$p_{1,-1\mid -1,\textbf{X}}$ &$\frac{1}{2}\epsilon_{n}+\frac{1}{4}$ & $\frac{1}{2}\epsilon_{n}+\frac{1}{4}$\\
$p_{-1,-1\mid -1,\textbf{X}}$ & $\frac{1}{4}$ & $\epsilon_{n}+\frac{1}{4}$\\ 
\end{tabular}}
\end{table}

It suffices to use the same proof as above in order to establish that \[
\sup_{\theta\in \Theta, \widehat{L}\in \mathcal{F}_{L,n}, \widehat{U}\in \mathcal{F}_{U,n}}\mathcal{R}_{\text{upper}}(\widehat{f})-\mathcal{R}^\ast_{\text{upper}}\geq n^{-\alpha}. 
\]
Finally, combine the results and we conclude:
\begin{equation*}
    \sup_{\theta\in \Theta, \widehat{L}\in \mathcal{F}_{L,n}, \widehat{U}\in \mathcal{F}_{U,n}}\mathcal{R}_{\text{upper}}(\widehat{f})-\mathcal{R}^\ast_{\text{upper}}\geq n^{-(\alpha \wedge \beta)}.
\end{equation*}
\end{document}